\documentclass{llncs}
\usepackage{amssymb}
\usepackage{amsmath}

\usepackage{color}
\usepackage{url}
\usepackage{ifpdf}
\usepackage{graphicx}
\usepackage{wrapfig}
\usepackage{xspace}
\usepackage{multirow}
\ifpdf
\usepackage[colorlinks=true,linkcolor=blue,citecolor=blue]{hyperref}
\fi

\usepackage{verbatim}

\long\def\comment#1{\relax}

\newcommand{\etc}{{etc.}}

\newcommand{\Rscr}{\mathcal{R}}

\newcommand{\Escr}{\mathcal{E}}

\renewenvironment{paragraph}[1]{\smallskip\noindent{\it #1}~}

\usepackage{tweaklist}
\renewcommand{\enumhook}{
	\setlength{\topsep}{0em}
	\setlength{\itemsep}{0em}
	\setlength{\partopsep}{0em}
	\setlength{\parskip}{0em}
	\setlength{\parsep}{0em}
	\addtolength{\leftmargin}{-0.75em}
}
\renewcommand{\itemhook}{
	\setlength{\topsep}{0em}
	\setlength{\itemsep}{0em}
	\setlength{\partopsep}{0em}
	\setlength{\parskip}{0em}
	\setlength{\parsep}{0em}
	\addtolength{\leftmargin}{-0.75em}
}

\newcommand\lra{\longrightarrow}

\newcommand\lb{\{}
\newcommand\rb{\}}
\newcommand\sbsim{\approx}
\newcommand\sbsimu{\precsim}

\newcommand\relQ{\sim?}

\newcommand\Dia{\ensuremath{\mathsf{D}}\xspace}

\newcommand\sysI{\ensuremath{\mathtt{mkISys}}\xspace}
\newcommand\mtC{\ensuremath{\mathtt{mtC}}\xspace}
\newcommand\att{\ensuremath{\mathsf{att}}\xspace}

\newcommand\rcv{\ensuremath{\mathtt{rcv}}\xspace}

\newcommand\drop{\ensuremath{\mathtt{drop}}\xspace}
\newcommand\delay{\ensuremath{\mathtt{delay}}\xspace}
\newcommand\replay{\ensuremath{\mathtt{replay}}\xspace}
\newcommand\edit{\ensuremath{\mathtt{edit}}\xspace}
\newcommand\divert{\ensuremath{\mathtt{divert}}\xspace}
\newcommand\create{\ensuremath{\mathtt{create}}\xspace}

\newcommand{\qedd}{\nobreak \ifvmode \relax \else
      \ifdim\lastskip<1.5em \hskip-\lastskip
      \hskip1.5em plus0em minus0.5em \fi \nobreak
      \vrule height0.75em width0.5em depth0.25em\fi}

\long\def\omitthis#1{\relax}

\title{Dialects for  CoAP-like Messaging Protocols}

\author{Carolyn Talcott\inst{1}}
\institute{
SRI International, Menlo Park, USA, \email{carolyn.talcott@sri.com} 
}

\begin{document}
{\def\addcontentsline#1#2#3{}\maketitle}  
\pagestyle{plain} 

\begin{abstract}

Messaging protocols for resource limited systems
such as distributed IoT systems are often
vulnerable to attacks due to security choices made
to conserve resources such as time, memory, or
bandwidth. For example, use of secure layers such
as DTLS are resource expensive and can sometimes
cause service disruption. Protocol dialects are
intended as a light weight, modular mechanism to
provide selected security guarantees, such as
authentication. In this report we study the CoAP
messaging protocol and define two attack models
formalizing different vulnerabilities. We propose a
generic dialect for CoAP messaging. The CoAP
protocol, dialect, and attack models are formalized
in the rewriting logic system Maude. A number of
case studies are reported illustrating
vulnerabilities and effects of applying the dialect.
We also prove (stuttering) bisimulations between
CoAP messaging applications and dialected versions,
thus ensuring that dialecting preserves LTL
properties (without Next) of CoAP applications.
To support search for attacks in complex messaging 
situations we specify a simple application layer to
drive the CoAP messaging and generalize the attack
model to support a form of symbolic search for attacks.
Two case studies are presented to illustrate the
more general attack search.
\end{abstract}

\tableofcontents
\newpage
\section{Introduction}
\label{sec:intro}

There is a rapidly growing number of networks of
IoT devices that impact our daily lives. They
enable smart homes, offices, factories, and
infrastructure. They are increasingly used in
health care, precision agriculture, logistics,
supply chain management, and situation awareness.
The IoT devices sense and act on the physical
environment and must operate using limited
resources (energy, bandwidth, memory, compute
power, \dots). The vulnerabilities of small
inexpensive devices are magnified when deployed at
scale. Security is crucial for safe operation but
security is typically resource intensive. Design of
such systems must balance balance resources used
for messaging and for security.

Protocol dialects are a light weight mechanism that
can provide authentication and possibly additional
security services such as integrity. A dialect
transforms the underlying protocol to mutate
messages before sending such that only dialect
partners can revert the mutation for delivery to
the intended receiver. A key feature of dialects is
to rapidly change the mutation used to prevent
attackers from using any decoding information they
might gain. This \emph{moving target defense} means
that the complexity of mutation processing can be
kept low without compromising the security
guarantee. The need to change mutations presents a
significant challenge to synchronize change across
distributed network nodes.

Existing dialect works, for example
\cite{sjoholmsierchio-etal-2021netsoft,gogineni-etal-2022arXiv,ren-etal-2023seccom}, usually consider
protocols running over TCP or other reliable
transport. In this case mutations can reasonably be
synchronized based on time. In the case of
unreliable transport where messages can be delayed,
reordered, lost, or repeated, time-based
synchronization is problematic. Unreliable
transport is mentioned in \cite{GEMM-2023esorics}
as one classification parameter, but
synchronization mechanisms for this case are not
treated.

In this work we study dialects for protocols running
over unreliable transport, using the CoAP messaging
protocol \cite{RFC7252} as prototypical example. We
investigate the use of counters, analogous to those
used to prevent replay in security protocols. for
synchronizing. We consider two methods of evaluating
the CoAP specification and dialect wrapper. One is
based on the CoAP vulnerabilities discussed in
\cite{coap-attacks} and the other is based on
properties a CoAP enabled endpoint could assume when
running over an unreliable transport layer such as
UDP. 

For the vulnerability case, we formalize the
attack scenarios presented in \cite{coap-attacks} as
motivation for strengthening the specification. These
scenarios consider an attacker that can only
block/drop or delay messages. We use reachability
analysis to look for attacks. In some cases the Maude
endpoints exhibit the expected vulnerability. in some
cases because the Maude specification implements some
of the stronger mechanisms, the attack fails, but
there are related attacks. In some cases an
attacker that can redirect messages can realize the
attack, but dialecting mitigates such attacks.

For the CoAP properties case the goal is
to identify guarantees provided by CoAP messaging
running on an unreliable transport, identify attack
points, and corresponding attack model, and prove
that dialecting provides adequate defense.
This is formalized using stuttering bisimulation
relations.

\paragraph{Contributions.}
The main contributions of this report are:
\begin{itemize}
\item  
  Executable specification of the CoAP messaging protocol (Section \ref{sec:coap-spec}).
\item  
  Specification of two families of resource limited attack models (Section \ref{sec:attack-models}):
an active attacker that can modify messages in transit
and a reactive attacker that can observe messages in
transit and transmit (modified) copies. 
\item  
Specification of a generic dialect layer and transform for messaging protocols (Section \ref{sec:coap-dialect-spec}).
\item  
Reachability analysis demonstrating CoAP vulnerabilities identified in \cite{coap-attacks} (Appendix \ref{apx:coap-vulnerabilities}). These require an active attacker.
\item  
Reachability analysis demonstrating diverse reactive
attacks (Section \ref{sec:reactive-attacks}).
\item  
 An analysis of protection provided by the proposed dialect class in unreliable/reliable networks (Section 
 \ref{sec:dialect-properties}).
\item  
Proof of stuttering bisimilarity theorems characterizing effects of the dialect transformation
(Appendix \ref{apx:sbsim-proofs}).
\item
A simple application layer is specified to drive complex
messaging scenarios and basic properties defined from
which to build invariant specifications.
The reactive attack
model is generalized to support a form of symbolic search 
for attacks. (Section \ref{apx:app-experiments}).
\end{itemize}
  
\paragraph{Plan.}
Section~\ref{sec:background} provides background information on CoAP, dialects, and rewriting logic/Maude. 
Section~\ref{sec:coap-spec} describes the Maude specification of CoAP messaging.
Resource limited attack models are presented in
Section \ref{sec:attack-models}  with examples
illustrating a variety of attacks presented in
Section \ref{sec:reactive-attacks}.
In Section~\ref{sec:dialect-fns} we present the abstract mutation functions of the proposed dialect scheme.
The Maude specification of the CoAP dialect is given in Section~\ref{sec:coap-dialect-spec}.
\omitthis{The reachability analysis experiment are summarized in Section~\ref{sec:experiments}.}
In Section~\ref{sec:dialect-properties} we discuss the
properties of dialects such as the CoAP dialect.
The relation to \cite{GEMM-2023esorics} is discussed in
Section~\ref{sec:related}.  Concluding remarks
are given in Section~\ref{sec:concl}.

In Appendix~\ref{apx:exe-props} we spell out key
properties of CoAP system executions, including
characterizing the minimal time elapse function. In
Appendix~\ref{apx:coap-vulnerabilities} we specify
the scenarios illustrating CoAP vulnerabilities
listed in \cite{coap-attacks} and do reachability
analysis to determine which vulnerabilities are
realized by the Maude specification. We also analyze
the dialected versions of each scenario. In
Appendix~\ref{apx:sbsim-proofs} we prove two
bisimulation theorems for dialected CoAP systems,
verifying relations proposed in
Section~\ref{sec:dialect-properties}.
The application layer, generalized attack model and
associated case studies are presented in Appendix \ref{apx:app-experiments}.

The Maude specification and case studies can be found at 
\url{https://github.com/SRI-CSL/VCPublic} in the folder
CoAPDialect.

\section{Background}
\label{sec:background}

\subsection{The CoAP protocol}
\label{subsec:coap}

 The Constrained Application Protocol (CoAP)
\cite{RFC7252,coap-attacks} is an HTTP-like
client-server Protocol for use by resource-constrained
devices (e.g. low power) and networks (low bandwidth,
lossy). CoAP provides a request/response RESTFUL
interaction model between application endpoints.
Servers hold resources that can be generated, updated,
accessed by clients. An endpoint can be a client, a
server, or both. A design goal is to keep message
overhead small, to limit need for fragmentation.
Note that server resources may be connected to 
actuators having effects on the environment.

CoAP features include UDP binding (reliability
supported at the protocol level by a retry mechanism)
and asynchronous two-way message exchange. A binding
to Datagram Transport Layer Security (DTLS) provides
security, but with substantial overhead. Expected
applications include: smart buildings, instrumented
infrastructure, medical devices, etc.

RFC7252 specifies a binary format for CoAP messages.
Here we present an abstract data-type view. A CoAP
message consists of four parts: Header, Token,
Options, and Payload. The header has four parts:
Version, Type, Code, and MsgId (we will ignore the
version).

The header type is one of \verb|{CON, NON, ACK, RST}|.
A message of type \texttt{CON} is confirmable. The
receiver must acknowledge with a message of type
\texttt{ACK}, in addition to a response
if appropriate. The sender should resend a
\texttt{CON} message if no ACK is received within
\texttt{ACK-TIMEOUT} time. There is a bound on the
number of resends (a configuration parameter). A
message of type \texttt{NON} is non-confirmable. If
received it may require that a response is sent, but
no separate acknowledgment is sent. A message of type
\texttt{RST} is a reset message, used to indicate
receipt of message that can't be handled, also use as
a kind of ping.

Code is an HTTP like code that determines if the
message is a request or response. In the case of a
request the code also determines the method--one of
\texttt{GET,POST,PUT,DELETE}. In the case of a
response, the code determines if the result is a
success (with value in the payload), a client error,
or a server error.

The \texttt{MsgId} header element is a string that
uniquely identifies the message (among messages from
the sender). It is used to match acknowledgements to
\texttt{CON} messages. It also is used to provide
limited replay protection. Message identifiers have a
limited lifetime (a configuration parameter) after
which they can be reused.

The token message component is a unique identifier
generated by a requestor to match request to response.
Options is a list of name-value pairs used to specify
additional information, including: content-format,
request Uri, and response Uri (of a created resource).
Payload contains additional request information (for
example the value to PUT) or response value (for
example the result of a GET).

The rules for handling requests with a given method
are similar to those for HTTP. A response can be
\emph{piggy backed} on the acknowledgement of a
\texttt{CON} request. A separate response can be of type
\texttt{CON} or \texttt{NON}. Servers are required to
respond to legitimate requests, but they should only
process a request once. Thus if a confirmable request
is resent, for example because the acknowledgement was
lost or delayed, the server may just resend the
\texttt{ACK} or it may resend the actual response, but
it should not repeat the actions needed to carry out
the response.

\subsection{Dialects}
\label{subsec:dialects}

A protocol dialect is a variation on protocol messages
intended as a light weight mechanism to provide some
additional security properties such as forms of
authentication or message integrity.

A dialect uses a family of lingos to transform protocol messages. A lingo provides a pair of functions that are inverse of one another, possibly modulo parameters. Dialects provide moving target defense in the sense that lingo parameters and choice of lingo change over time in a dialect specified manner. This means weaker forms of message transformation can be used, because even if a message transformation is determined, it doesn't help as that transformation will not be used on later messages. It also means that honest participants must synchronize on the choice of lingo and lingo parameters.

A typical application of dialects is to networks where
at least some nodes are resource limited devices. An
example would be smart infrastructure with sensors and
actuators running on battery power, and with limited
memory and cpu capability. Thus a dialect should not
require substantial resources. In addition to
providing new security properties, a dialect should
preserve important properties of the underlying
protocol. The protocol may be running on a network
with limited bandwith, thus a dialect should not
generate significant additional traffic, and should
not significantly increase message size.

The following are some examples of dialects and their application.
Dialects to improve SDN (Software defined network)
security are defined and evaluated in
\cite{sjoholmsierchio-2019nps,sjoholmsierchio-etal-2021netsoft}. Two transformations based on HMAC over the
message data are proposed. One places the HMAC in the
session id slot. This limits the HMAC size and thus the
hash key is changed every second. For messages that
use this slot, a 512-bit HMAC tag is appended to data
packets. Sequence numbers are included in the hashed
content to prevent replay. Keys are generated from
pre-shared symmetric keys, delivered out-of-band.
Dialects are implemented using proxies inserted into
the communication path to carry out the message
transformations. One goal is to protect against TLS
downgrade attacks by providing authentication for
initial TLS communications.

Dialects for FTP and MQTT protocols using bit
shuffling or packet splitting are studied in
\cite{gogineni-etal-2022arXiv}. A self-synchronization
mechanisms is propsed. In
\cite{gogineni-etal-2021site} the authors propose a
Deep Neural Network mechanism for synchronization. A
dialect is implemented by (conceptually) inserting a
separate process between the protocol and the network.

Dialecting is proposed in \cite{ren-etal-2023seccom}
to mitigate known exploitable vulnerabilities in IoT
devices deployed in large numbers thus enabling
massive attacks. The idea is to transform functions
that construct and parse message data structures to
construct and parse mutated data structures. The
message constructor/parser functions are identified
automatically and message mutations are typically per
message field. Specific mutations are designed using
testing. Dialects are implemented by modifying the
identified message constructor and parser functions of
the protocol implementation. The method is evaluated
on implementations of MQTT, LwM2M, and DDS. In the
experiments, the mutations successfully blocked
reproducible attacks corresponding to a selected set
of CVEs reported for the protocols

The above works propose dialects and evaluate them
experimentally. In \cite{GEMM-2023esorics} a formal
framework for specifying and applying dialects is
proposed. The ideas are illustrated using the MQTT
protocol. This work and its relation to the current
work is discussed in detail in section
\ref{sec:related}.

\subsection{Rewriting logic and Maude}
\label{subsec:rwl-maude}

Rewriting logic \cite{meseguer-92unified-tcs} is a logic
for specifying concurrent and distributed systems.
A rewrite theory has the form $(\Sigma,\Escr,\Rscr)$
where $\Sigma$ is a signature specifying a partial order
of sorts, and a set of operators (the name, argument sorts
and result sort). $(\Sigma, \Escr)$ is the equational sub-theory, where the equations $\Escr$ define functions
and properties declared in the signature.  $\Rscr$ is
the set of rewrite rules of the form
$l: lhs \Longrightarrow rhs~~\mathit{if}~c.$  Here $lhs,rhs$ are terms, possibly with variables, and $c$ is
a boolean term, the condition, which is optional.
An important feature of
rewriting logic is that rules are applied locally,
i.e. to matching sub-terms of the term being rewritten.
Given terms (of the language given by $\Sigma$) $t$, $t'$
we have $t\Longrightarrow t'$ ($t$ rewrites to $t'$) using
rule $l$ just if there is a position $p$ of $t$ and a
substitution $\sigma$ matching $lhs$ to the subterm of $t$
at $p$ ($\sigma(lhs) =_{\Escr} t\downarrow p$) such that
$\sigma(c)$ holds and
$t'$ is the result of replacing the subterm of $t$
at $p$ by $\sigma(rhs)$.

Maude \cite{clavel-etal-07maudebook} is a language and tool implementing rewriting logic. Maude can be downloaded
from the Maude website \cite{maudeweb}.
Rewrite theories are represented by Maude modules.
Functional modules (declared using the keyword 
\texttt{fmod} at the beginning and \texttt{endfm} at the
 end) represent equational theories (with no rules).
System modules (declared using the keyword 
\texttt{mod} at the beginning and \texttt{endm} at the
 end) represent rewrite theories with rules.
We introduce the main Maude syntax elements with examples
from the CoAP specification (discussed in detail in 
Section \ref{sec:coap-spec}).
Sorts and their ordering are declared using the keywords 
\texttt{sort} and \texttt{subsort}. For example, the following
declares two sorts \texttt{Msg} (the message sort) and
\texttt{MsgS} (multisets of messages) with \texttt{Msg}
(singleton sets) being a subsort of \texttt{MsgS}.

\begin{small}
\begin{verbatim}
    sort Msg MsgS .
    subsort Msg < MsgS .
\end{verbatim}
\end{small}

\noindent
Functions/operators are declared using the keyword
\texttt{op}. For example, the operator \texttt{m}
constructs a message from two strings (the target and
source) and a term of sort \texttt{Content}. The
attribute \texttt{ctor} says that \texttt{m} is a
constructor and will not be subject any equations.
\texttt{mtM} is a constant (no arguments) denoting the
empty message set. The operator \verb|__| has empty
syntax, the blanks denoting argument positions. It
correspond to union of multi sets. That union is
associative and commutative with unit the empty set is
expressed by the attributes \texttt{assoc},
\texttt{comm}, and i\texttt{d: mtM}. The \texttt{ctor}
attribute says that no \texttt{other} equations apply to
elements of \texttt{MsgS}. Thus multisets are generated
from the empty set and singletons by union.

\begin{small}
\begin{verbatim}
    op m : String String Content -> Msg [ctor] .
    op mtM : -> MsgS [ctor] .
    op __ : MsgS MsgS -> MsgS [ctor assoc comm id: mtM] .
\end{verbatim}
\end{small}
\noindent
This is a general pattern used to specify (multi)sets of
a given sort.  A similar pattern is used to specify
lists of elements, omitting the \texttt{comm} attribute.

Equations are introduced with the key word \texttt{eq}
(or \texttt{ceq} for conditional equations). For
example, the function \texttt{getTgt}, that gets the
target of a message, is defined by the following
declaration and equation.
\begin{small}
\begin{verbatim}
    op getTgt : Msg -> String .
    eq getTgt(m(tgt:String,src:String,c:Content)) = tgt:String .
\end{verbatim}
\end{small}
\noindent
Here we use inline variable declaration.  For example,
\texttt{tgt:String} declares \texttt{tgt} to be a variable of
sort \texttt{String}.  Alternatively, we could declare
the variable globally.
\begin{small}
\begin{verbatim}
    vars tgt src : String .
    var c : Content .
    eq getTgt(m(tgt,src,c)) = tgt .
\end{verbatim}
\end{small}

Rewrite rules are signalled by the keyword \texttt{rl}
(\texttt{crl} for conditional rules), optionally followed by
a label in []s.  The rule labeled
\texttt{net}  applies to any configuration that has a
network element of the form \texttt{net(dmsgs0
dmsg,dmsgs1)}, and moves \texttt{dmsg} from the first to the
second argument of the constructor \texttt{net}.
\footnote{A delayed message \texttt{dmsg} has the form
\texttt{msg @ d} where \texttt{d} is a natural number
representing the amount of time it will take the message
to traverse the network.}
\begin{small}
\begin{verbatim}
    rl[net]:
    net(dmsgs0 dmsg,dmsgs1) => net(dmsgs0,dmsgs1 dmsg)
    [print "\n[net] "]  .
\end{verbatim}
\end{small}
\noindent
The rule attribute \verb|[print "\n[net] "]| causes Maude
to print to the terminal ``[net]'' on a new line each time
the rule is applied in an execution. This print feature
allows the user to define a model specific execution trace
that is useful for debugging or demonstrating features of
an execution.

There are many Maude tools for static and dynamic
analysis of a system specification. In the present work
we only need the rewrite command for prototyping and
the search command for reachability analysis.
Assume you have defined a
module \texttt{SCENARIO} that includes all the modules
defining rules and data structures defining your model.
Assume further that system configurations are of sort
\texttt{Sys} and the module \texttt{SCENARIO} defines a
constant \texttt{initSys} of sort \texttt{Sys}.
Then we can use the command 
\begin{small}
\begin{verbatim}
    rew [n] initSys .
\end{verbatim}
\end{small}
\noindent
to find one way that \texttt{initSys} might evolve
for \texttt{n} steps.   If you omit \texttt{[n]} the
execution will run until a state is reached that has no
successors (a terminal state).  This should be used with
care as there may not be a reachable terminal state!
If a boolean function \texttt{P(sys:Sys)} is defined,
one can use the search command to carryout reachability
analysis.   The following search command will return the
first \texttt{n} distinct reachable states that satisfy
the property \texttt{P}.  
\begin{small}
\begin{verbatim}
  search [n] initSys =>+ sys:Sys such that P(sys:Sys) .
\end{verbatim}
\end{small}
\noindent
If \texttt{[n]} is omitted Maude will search for all
reachable states satisfying \texttt{P}.  Again caution
is suggested as there may be infinitely many such states.

Example initial configurations and properties for our CoAP
specification are defined in section \ref{sec:coap-spec}
and reachability analysis experiments for different attack
models are given in Sections \ref{sec:reactive-attacks} and
Appendix \ref{apx:coap-vulnerabilities}.

\begin{small}
\begin{verbatim}

\end{verbatim}
\end{small}

\section{Executable Specification of CoAP}
\label{sec:coap-spec}

The Maude specification of the CoAP messaging protocol
consists of specifications of data structures to
represent execution state \ref{subsec:data}, rules
specifying the protocol behavior \ref{subsec:rules},
and an attack specification framework
\ref{subsec:attack-specification}. The data structures
include a message data type and data structures
representing endpoint and network state. The basic
CoAP rules model sending and receiving messages, flow
of messages in the network, and passing of time. The
attack model consists of a specification of attacker
capabilities and a rule for application of a
capability to modify communications. The Maude
specification of the two families of attack models is
given in Section \ref{subsec:attack-model-spec} The
CoAP specification concludes with definitions of
functions for generating application message lists and
scenario initial states ( Section
\ref{subsec:coap-scenarios}) and some small example
scenarios and analyses to illustrate the ideas
(Section \ref{subsec:sample-scenarios}).

\subsection{Data types}\label{subsec:data}
\paragraph{Message data type.}

The specification of the message data type is a direct
formalization of the (abstract) structure of messages
as specified in \cite{RFC7252}.
Messages (sort \texttt{Msg}) are constructed by the operation

\begin{small}
\begin{verbatim}
    op m : String String Content -> Msg [ctor] .
\end{verbatim}
\end{small}

\noindent
thus a typical message term has the form
\texttt{m(tgt,src,content)} with \texttt{tgt},
\texttt{src} being strings identifying the receiver
(target) and sender (source) respectively.
 
The sort \texttt{DMsg} of delayed messages consists of
terms of the form \texttt{msg @ d} where \texttt{msg}
is of sort \texttt{Msg} and \texttt{d} is the delay, a
natural number (sort \texttt{Nat}). The delay is used
to model the time between sending and receiving a
message (network latency). It is also used to model the
timer controlling resend of a confirmable message.

As discussed in Section \ref{subsec:coap} message
content (sort \texttt{Content}) is constructed from
sorts \texttt{Head}, multisets of \texttt{Option}s, and
\texttt{Body} (aka payload).

\begin{small}
\begin{verbatim}
    op c : Head String Options Body -> Content [ctor] .
\end{verbatim}
\end{small}

\noindent
We define a constant \texttt{mtBody} of sort \texttt{Body}
to stand for the empty payload, and assume non empty
payload is represented by strings: \texttt{b(str)}.
Elements of sort \texttt{Head} are constructed by

\begin{small}
\begin{verbatim}
    op h : String String String -> Head  [ctor] .
\end{verbatim}
\end{small}
 
\noindent
where in \texttt{h(type,code,mid)} the first string
is the message type (\texttt{"CON"}, \texttt{``NON''}, 
\texttt{"ACK"} or \texttt{"RST"}), the second string
is the HTTP-like code, and the third string is the
message unique identifier, generated by the sender,
as discussed in \ref{subsec:coap}

The \texttt{String} element of a message content is the
token used to match a response to the corresponding request.  Each request from a given source to a given
target should have a unique token generated for it.

Sort \texttt{Options} is a multiset of sort \texttt{Option} where \texttt{mtO} is the empty option set. An individual
option has the form \texttt{o(oname,oval)} where
\texttt{oname} is a string naming the option and \texttt{oval}
is the option value of sort \texttt{String} or \texttt{Nat}.
The current specification supports options representing
request URIs and URIs in a response resulting from creation
of a resource.   For example the equation 

\begin{small}
\begin{verbatim}
    eq getPath(opts o("Uri-Path",path)) = path .
\end{verbatim}
\end{small}

\noindent
defines how to access the Uri-path of a request.
Multisets of delayed messages (sort \texttt{DMsgS}, id:
\texttt{mtDM}) are defined similarly to multisets of
messages shown above. 

Selectors are defined for each
component of a message. For example the selector for
message type is defined as follows:

\begin{small}
\begin{verbatim}
      op getType : Msg -> String .
      op getType : Content -> String .
      op getType : Head -> String .

      eq getType(m(dst:String,src:String,c:Content)) 
                = getType(c:Content) .
      eq getType(c(h:Head,tok:String,opts:Options,body:Body)) 
                = getType(h:Head) .
      eq getType(h(type,code,mid)) = type .
\end{verbatim}
\end{small}

The CoAP protocol itself does not generate requests, rather
it packages and transmits application level messages. For
testing and analysis purposes we represent the application by
a list of application level messages (requests) (sorts
\texttt{AMsg}, \texttt{AMsgL}) to be transmitted.
An application message consists of six strings and a
message body.  The strings specify: the application id,
the message target and type, the method, and the
resource path and query parameters if any.
 
\begin{small}
\begin{verbatim}
   ****     appid    tgt   type   meth   path  qparams body
  op amsg : String String String String String String Body
          -> AMsg [ctor] .
\end{verbatim}
\end{small}

\noindent
The empty \texttt{AMsgL} is \texttt{nilAM} and list concatenation is \verb|_;_|.

\paragraph{Endpoint state.}

A CoAP network system consists of a set of endpoints
(devices running CoAP) and the network. The network is
modeled as a pair of delayed message sets
\texttt{net(dmsgs0,dmsgs1)}. Newly sent messages enter
the first component, are moved by the net (or an
attacker) to the second component from which they are
delivered once the delay is zero.

An endpoint (also referred to as a device) is a term of
sort \texttt{Agent} of the form 
\begin{small}
\begin{verbatim}
        [epid | attrs]
\end{verbatim}
\end{small}
\noindent
where \texttt{epid} is a string identifying the endpoint
and \texttt{attrs} is a set of attributes representing the
endpoints current state. CoAP endpoint attributes include
attributes to track status of msgs sent and received in
order to implement the rules for reliability (confirmable
messages) and deduplication.  They also include
configuration parameters allowing each endpoint to be
configured differently if desired.

\begin{itemize}
\item
\texttt{w4Ack(dmsgs)}: \texttt{dmsgs} is the set of 
confirmable messages sent and not yet acknowledged.  The delay
part is the amount of time to wait before resending the
message;
\item
\texttt{w4Rsp(msgs)}: a message in \texttt{msgs} is either
a non-confirmable request waiting for a response, or
a confirmable request for which an ACK has been received
but the response is still pending;
\item
\texttt{rspRcd(msgs)}:  logs responses received to avoid reprocessing repeated responses;
\item
\texttt{rspSntD(dmsgs)}: logs responses sent where the
delay represents the message id/token lifetime. It can
be used to resend the response in case the original is
lost, and to detect replay attempts within the message
id lifetime.
\item
\texttt{rsrcs(rmap)}:  \texttt{rmap} is a map from resource paths to values, representing server resources state;
\item
\texttt{ctr(n)}:  \texttt{n} is a natural number used to model generating fresh/unique message Id strings (\texttt{genMid(prefix,n)}) and unique tokens (\texttt{genTok(prefix,n)});
\item
\texttt{sendReqs(amsgl)}: a list of application messages
to be transmitted
\item
\texttt{config(cbnds)}: a map from global parameter names to configured values, examples include message send delay (abstracts time in transit, \texttt{cb("msgSD",n)}), maximun number of resends (\verb|cb("MAX_RETRANSMIT", n0)|), backoff parameter for computing delay between resends (\texttt{cb("msgQD",n1)}), congestion control parameters such as time to wait for an ACK before resending a confirmable message (\verb|cb("ACK_TIMEOUT", n2)|), \etc.;
\item
\texttt{toSend(dmsg)}: an attribute used to store information
returned by auxiliary functions
\end{itemize}
  
A number of auxiliary functions on attributes are defined for
use in specifying rules.   We discuss two examples.
The function \texttt{getMsgSndDelay} looks up the \texttt{"msgSD"} configuration parameter.  This parameter
models average network transit time of a message.
\begin{small}
\begin{verbatim}
    op getMsgSndDelay : Attrs -> Nat .
    op getMsgSndDelay : CBnds -> Nat .

    eq getMsgSndDelay(attrs config(cbnds)) =  getMsgSndDelay(cbnds) .
    eq getMsgSndDelay(cb("msgSD",n) cbnds) = n .
\end{verbatim}
\end{small}

The function \texttt{findRspRcd} checks if a response message
from \texttt{dst} has been received in response to a request
with identifier \texttt{mid} and token \texttt{tok}. It uses
the message selectors \texttt{getSrc} which would be the
destination of the corresponding request, and \texttt{getTok}
to use to match the response to a request.
\texttt{findRspRcd} is used to avoid reprocessing responses
and also to be able to resend and ACK to a confirmable
response.  

\begin{small}
\begin{verbatim}
    ****                    dst  mid     tok
    op findRspRcd : Attrs String String String -> Bool .
    op findRspRcd : MsgS String String String -> Bool .
    eq findRspRcd(devatts rspRcd(msgs), dst,mid,tok) =
          findRspRcd(msgs, dst,mid,tok) .
    ceq findRspRcd(msgs msg, dst,mid,tok) = true
      if getSrc(msg) == dst
      /\ getTok(msg) == tok .  
    eq findRspRcd(msgs, dst,mid,tok) = false [owise] .
\end{verbatim}
\end{small}

\subsection{Rules}\label{subsec:rules}

Five rules are used to specify CoAP executions. The rules
labelled \texttt{devsend} and \texttt{rcv} specify
transmission and receipt of messages by an endpoint. The
rule labelled \texttt{ackTimeout} specifies when a
confirmable message is resent. The rule labelled
\texttt{net} specifies the network action of moving a
delayed message from the input side to the output side.
Finally, the rule labelled \texttt{tick} specifies the
passing of time. Rules \texttt{devsend}, \texttt{rcv} and
\texttt{ackTimeout} operate on sub-configurations
consisting of an endpoint and the network, the rule
\texttt{net} simply transforms the network configuration
element, while \texttt{tick} rule requires the full system
configuration \verb|{conf}|, formed by encapsulating a
configuration term in curly braces.  The \texttt{net} rule was discussed in Section~\ref{subsec:rwl-maude}.  The
remaining rules are discussed in the remainder of
this (sub)section.

\paragraph{Sending a message.}
The rule \texttt{devsend} sends the first element,
\texttt{amsg}, of the application message list attribute. The
variable \texttt{devatts} matches the remaining attributes of the sending endpoint state. The rule can only fire if \texttt{noW4Ack(devatts)} is true
and
sending is enabled (\texttt{canSend(devatts)}).
These tests model the CoAP congestion control.
\texttt{noW4Ack(devatts)} holds if the number of
confirmable messages awaiting acknowledgements is
not greater than the configuration parameter \texttt{"w4AckBd"} and \texttt{canSend} checks
that the \texttt{sndCtr} attribute of \texttt{devatts} is \texttt{0}.
The function \texttt{sndAMsg} constructs the delayed message to send, \texttt{dmsgs}, and updates \texttt{devatts}. To
construct \texttt{dmsgs}, a message id and token are
generated and the method string is converted to a code, then
the header, options and body are used to produce the message
content. The message delay is obtained from the configuration
parameter \texttt{"msgSD"}. If the message is type
\texttt{CON} it is added to the \texttt{w4Ack} attribute of
\texttt{devatts} with a delay given by the configuration
parameter \verb|"ACK_TIMEOUT"|, otherwise the message is
added to the \texttt{w4Rsp} attribute.  The \texttt{ctr} attribute is incremented by 2 (the two uses to generate
a message id and a token), and the \texttt{sndCtr} is reset using the configuration
parameter \texttt{"msgQD"}.

\begin{small} 
\begin{verbatim} 
    crl[devsend]: 
      [epid | sendReqs(amsg ; amsgl) devatts ] 
      net(dmsgs0,dmsgs1) 
      => 
      [epid | sendReqs(amsgl) devatts1] 
      net(dmsgs0 dmsgs,dmsgs1) 
    if noW4Ack(devatts) 
    /\ canSend(devatts) 
    /\ devatts1 toSend(dmsgs) := sndAMsg(epid,amsg, devatts) . \end{verbatim}
\end{small}

\paragraph{Receiving a message.}

A message with 0 delay is received by the target endpoint
using the rule labelled \texttt{rcv}. The work is done by
the function \texttt{rcvMsg} which uses the message header
code to classify the message as ``Request'', ``Response'',
``Empty'', or ``UnKnown'' and calls \texttt{rcvRequest},
\texttt{rcvResponse}, \texttt{rcvEmpty}, or
\texttt{sendReset} (if confirmable). Otherwise the message
is dropped. An empty message is either an acknowledgement
or a reset. In the acknowledgement case, the corresponding
delayed message in the \texttt{w4Ack} attribute is removed
if any, and the underlying message is added to the
\texttt{w4Rsp} attribute to be able to process the pending
response. In the reset case, if the message is confirmable
an acknowledgement is sent (this implements a ``PING''
functionality) otherwise the reset is ignored.
   
\begin{small}
\begin{verbatim}
    crl[rcv]:
      [epid | devatts] net(dmsgs0,dmsgs1 msg @ 0 )
      =>
      [epid | devatts1 ] net(dmsgs0 dmsgs,dmsgs1)
    if getTgt(msg) == epid
    /\ toSend(dmsgs) devatts1 := rcvMsg(epid, devatts, msg) .
\end{verbatim}
\end{small}

The rules for handling requests are specified in
\cite{RFC7252} section 5.8.  
The function \texttt{rcvRequest} implements these rules.
It first checks whether it has
already sent a response to this request using
\texttt{matchingRsp(devatts,msg)}. If so, if the request is
confirmable, the ACK is resent, otherwise the message is
ignored. If the request is new the method is computed from
the code and the appropriate method specific function is
called to process the message. As an example, if the method
is \texttt{"GET"}, the function \texttt{rcvGet} is called.
This function sends a single message in response:

\begin{small}
\begin{verbatim}
    msg = m(src,epid,c(h(rtype,code,rmid),tok,mtO,body) .
\end{verbatim}
\end{small}

\noindent
\texttt{rmid} is a new message id, \texttt{src} is the
sender of the request, \texttt{tok} is the token of the
request message. The type, \texttt{rtype}, is
\texttt{"ACK"} if the request in confirmable. In this
case, \texttt{msg} is a combined acknowledgement and
response. The type is \texttt{"NON"} if the request is
nonconfirmable. In this case \texttt{msg} is a simple
response. \texttt{body} is the result of attempting to
access the resource at the path specified in the
request message options using
\texttt{getResourceVal(devatts,path)}. It is
\texttt{b(val)} if a value \texttt{val} is found, and
\texttt{mtBody} if the resource does not exist. The
message code is \texttt{"2.05"} if a resource value is
found, indicating successful GET. The message code is
\texttt{"4.04"} if the resource is not found,
indicating an error. \texttt{rcvGet} updates the
endpoint attribute by incrementing the counter,
\texttt{ctr}, and recording that the request has been
processed by adding \texttt{msg} to the \texttt{rspSnt}
attribute.

The function \texttt{rcvResponse} checks if it is
combined with an acknowledgement or a simple response.
In the former case, the matching request is removed
from the \texttt{w4Ack} attribute, if any, and adds the
response message to the \texttt{rspRcd} if it has not
already been seen. This abstracts interaction with the
application layer. In the case of a simple response, if
a response has already been received, an acknowledgment
is sent if the response is confirmable and otherwise
the message is ignored. If this is a new response, it
is recorded in the \texttt{rspRcd} attribute, any
occurrence of the matching request is removed from the
\texttt{w4Ack} and \texttt{w4Rsp} attributes. If the
response is confirmable an acknowledgement is sent.

\paragraph{ReSending a message.}
The rule \texttt{ackTimeout} can fire if a message in
the \texttt{w2Ack} attribute has delay \texttt{0}. The
option \texttt{o("rcnt",n)} records the number of
resends of the message. If this count is less than the
max allowed (configuration parameter
\verb|"MAX_RETRANSMIT"|) then the message is returned
to the \texttt{w4Ack} attribute with a new ack wait
time, \texttt{delay}, computed using the
\texttt{backOff} function that doubles the delay for
each resend. Finally, the message with normal sending
delay is put in the network input side.

\begin{small}
\begin{verbatim}
  crl[ackTimeout]:
    [epid | w4Ack((msg @ 0) dmsgs) config(cbnds) attrs]
    net(dmsgs0,dmsgs1)
    =>
    [epid | w4Ack(dmsgs2 dmsgs) config(cbnds) attrs]
    net(dmsgs0 dmsg, dmsgs1)
  if  m(dst,epid,c(hd,tok,opts o("rcnt",n),body)) := msg
  /\ opts0 :=  o("rcnt",s n )
  /\ delay := backOff(getAckWait4(cbnds), s n)
  /\ dmsg :=  m(dst,epid,c(hd,tok,opts,body)) @ 
                getMsgSndDelay(cbnds)
  /\ dmsgs2 :=  (if n < getMaxReSnd(cbnds) 
                 then m(dst,epid,c(hd,tok,opts opts0,body)) @ 
                      delay
                 else mtDM  --- stop resending
                 fi)  .
\end{verbatim}
\end{small}

\omitthis{
\paragraph{Network Transmission.}

The rule \texttt{net} simply moves a message from the input
side to the output side. 
\begin{small}
\begin{verbatim}
    rl[net]:
    net(dmsgs0 dmsg,dmsgs1) => net(dmsgs0,dmsgs1 dmsg)
     .
\end{verbatim}
\end{small}
}

\paragraph{Passing time.}

For managing the passing of time, we follow the
Real-time Maude approach \cite{olveczky08tacas}. The
idea is that there is an earliest time at which some
instantaneous rule will be enabled to fire, called the
minimal time elapse (\texttt{mte}). If the \texttt{mte} is zero
then the enabled rules should be applied until there
are no more. If the \texttt{mte} is greater than zero then time
passes by this amount, reaching a situation where some
rules can fire at the current time (or the state is
terminal). It is the job of the specification designer
to ensure that only a finitely many rules can fire
before time must pass. In the CoAP model, \texttt{mte}
is greater than zero if all delays of messages in the
network, or an \texttt{w4Ack} attribute, are greater
than zero, or a message send is pending but the
\texttt{sndCtr} is non-zero. The rule labelled
\texttt{tick} formalizes the above. Here \texttt{nz}
has sort \texttt{NzNat} (non-zero natural number). The
function \texttt{passTime} decrements each message
delay in the network and \texttt{w4Ack} attributes and
each \texttt{sndCtr} by \texttt{nz}. It also decrements
delays of messages in \texttt{rspSntD} attributes,
modeling the decrease in lifetime of the identifier and
token.

\begin{small}
\begin{verbatim}
    crl[tick]:
    {conf} => {passTime(conf,nz,mt)}
    if nz := mte({conf}) .
\end{verbatim}
\end{small}

\subsection{Attack specification}
\label{subsec:attack-specification}

Our specification  of attacks against CoAP messaging consists of a specification of attack
capabilities (sort \texttt{Cap}), a function
\texttt{doAttack},  an attacker agent, and a rule
for executing attacks, that invokes \texttt{doAttack}
on a target message using an available capability.
Attack models resulting from different choices of capability sets, and corresponding \texttt{doAttack} clauses, are discussed in Section \ref{sec:attack-models}.
We note that the usual symbolic model matching patterns representing attacker ability to construct messages
doesn't work because there is no a priori expectation
of what messages are expected.

An attacker agent has the same form as a device endpoint 
\begin{small}
\begin{verbatim}
      [aid  |  caps(acaps) attrs] 
\end{verbatim}
\end{small}
\noindent
but with attacker specific attributes that include, \texttt{caps(acaps)}, the capabilities  attribute specifying what the attacker can do.
The rule labelled \texttt{attack} specifies the
semantics of an attack capability. It selects a target
delayed message, \texttt{dmsg}, from the network input
component, and a capability from the attackers
capability set, and calls \texttt{doAttack} with the
current attributes, the target message and the
capability. The result is a replacement, \texttt{dmsgs}
for \texttt{dmsg} and possibly updated attributes. 
For example, if the capability corresponds to \texttt{drop} then
\texttt{dmsgs} is the empty set, and if the capability
corresponds to \texttt{replay} after delay $n$ then \texttt{dmsgs} is
\texttt{dmsg} \texttt{dmsg1} where \texttt{dmsg1} is
\texttt{dmsg} with $n$ added to its delay. \omitthis{For
the currently supported capabilities, the attacker
attribute set is updated by adding the target message
to the \texttt{kb} attribute.
}
\begin{small}
\begin{verbatim}
    crl[attack]:
      [aid  |  caps(acaps acap) attrs] 
       net(dmsgs0 dmsg,dmsgs1)
     =>
      [aid  |  attrs1] 
      net(dmsgs0,dmsgs1 dmsgs)
     if toSend(dmsgs) attrs1 := 
        doAttack(attrs caps(acaps),dmsg,acap) .
\end{verbatim}
\end{small}

\subsection{CoAP Scenario Specification}
\label{subsec:coap-scenarios}
 
A scenario is specified by an \emph{intial system
configuration}. This together with the CoAP rewrite
theory determines a set of execution traces and the
reachable states of the scenario. To support
systematic testing of CoAP messaging and analysis
of attacks we defined two functions to generate
initial configurations for scenarios \texttt{tCS2C}
(test Client Server 2 endpoints) and \texttt{tCS3C}
(test Client Server 3 endpoints). These are used in
the experiments validating the vulnerabilities
discussed in \cite{coap-attacks}, described in
Appendix \ref{apx:coap-vulnerabilities}. They were
also used in experiments testing various CoAP
messaging patterns and effects of dialecting.

In addition to having an empty network, initial configurations are characterized by the
initial attributes of the endpoints.
The function \texttt{initDevAttrs} defines the 
initial values of attributes that are the same
for all endpoints in the scenarios considered.
It has two arguments that are passed on to the
function \texttt{mkCoapConf} the defines global
configuration parameters
\begin{small}
\begin{verbatim}
    op mkInitDevAttrs : Nat Nat -> Attrs .
    eq mkInitDevAttrs(mqd:Nat,w4ab:Nat) =
      w4Ack(mtDM)    
      w4Rsp(mtM) 
      rspSntD(mtDM)
      rspRcd(mtM)
      ctr(0) 
      config(mkCoapConf(mqd:Nat,w4ab:Nat)) .
 \end{verbatim}
 \end{small}
 
\noindent
To experiment with different delays for sending
application requests \texttt{mkCoapConf} has two arguments
giving the values of the \texttt{"msgQD"} and
\texttt{"w4AckBd"} configuration parameters.  Recall that
\texttt{"msgQD"} is the delay between message sends, and
sending is blocked if there are more than \texttt{"w4AckBd"}
confirmable messages awaiting acknowledgement.
\begin{small}
\begin{verbatim}
    op mkCoapConf : Nat Nat -> CBnds .
    eq mkCoapConf(mqd:Nat,w4ab:Nat) =
      cb("ACK_TIMEOUT", 5)
      cb("ACK_RANDOM_FACTOR", 2)
      cb("MAX_RETRANSMIT", 1)
      cb("msgSD",2)
      cb("msgQD",mqd:Nat)
      cb("w4AckBd",w4ab:Nat)
      cb("ttl", 10) .
\end{verbatim}
\end{small}

\noindent
We define initial system configurations as follows.

\begin{definition}[Initial CoAP configuration]
   \label{defn:sysI}
 An initial CoAP system configuration \texttt{sysI}  contains one
 network configuration element of the form
\begin{small}
\begin{verbatim}
    net(mtDM,mtDM)
\end{verbatim}
\end{small}
\noindent 
and a finite number of endpoint agents of the form 
\begin{small}
\begin{verbatim}
    [eid | sendReqs(amsgl:AMsgL) rsrcs(rbnds:RMap)  initDevAttrs(mqd:Nat,w4ab:Nat)]
\end{verbatim}
\end{small}
\noindent
 An initial system configuration with attacker also has one attack agent of the form
\begin{small}
\begin{verbatim}
     [ "eve" | kb(dmsgs) caps(caps:Caps)]
\end{verbatim}
\end{small}
\noindent
which we omit if \texttt{caps:Caps} is empty.
\end{definition}
 
\paragraph{Initial System Configuration Generators.} 
\noindent
The scenario generators  use auxiliary functions
to generate endpoint agents (\texttt{mkDevC})  and
an optionally an attack agent (\texttt{mkAtt}).

The device agent generator \texttt{mkDevC} has arguments
specifying its identifier (\texttt{n}), initial delay before
sending messages (\texttt{j}), its list of application
messages to send \texttt{(amsgl)}, the initial state of any
resources (\texttt{rbnds}), and two arguments passed on to
\texttt{mkInitDevAttrs}.

\begin{small}
\begin{verbatim}
    op mkDevC : Nat Nat AMsgL RMap Nat Nat -> Agent .
    ceq mkDevC(n,j,amsgl,rbnds,mqd:Nat,w4ab:Nat) =
      [ epid | sendReqs(amsgl) rsrcs(rbnds)
               sndCtr(j) mkInitDevAttrs(mqd:Nat,w4ab:Nat) ]
    if epid := "dev"  + string(n,10) .
 \end{verbatim}
 \end{small}

\noindent
The attack agent generator \texttt{mkAtt} has one argument,
the capability set \texttt{caps}.

\begin{small}
\begin{verbatim}
    op mkAtt : Caps -> Agent .
    eq mkAtt(caps) = ["eve" | kb(mtDM) caps(caps)] .
 \end{verbatim}
 \end{small}

\noindent
The function \texttt{tCS2C} has three groups of arguments,
one group for each of two devices and one (singleton)
group for the attacker. 
The resul is a client device, \texttt{"dev0"}, with
a list of application messages \texttt{amsgl} to send; a
server device, \texttt{"dev1"}, with resources
\texttt{rmap}; and an attacker, \texttt{"eve"}, with
capabilities \texttt{caps}. The function simply applies
the appropriate agent constructor to each group of
arguments and adds an initial empty network state.  If there are no attack capabilities
(\texttt{mtC}), then no attack agent is generated.

\begin{small}
\begin{verbatim}
****  2 endpoint configuration, with devconfig parameters
**** w4ab = n requires w4Ack to have <= n dmsgs
**** mqd = m requires delay of m before sending next request
    op tCS2C : Nat Nat AMsgL RMap Nat Nat
               Nat Nat AMsgL RMap Nat Nat
              Caps -> Conf .
    eq tCS2C(n0,j0,amsgl0,rbnds0,mqd0:Nat,w4ab0:Nat,
             n1,j1,amsgl1,rbnds1,mqd1:Nat,w4ab1:Nat,
            caps) = 
       net(mtDM,mtDM)
       mkDevC(n0,j0,amsgl0,rbnds0,mqd0:Nat,w4ab0:Nat)  
       mkDevC(n1,j1,amsgl1,rbnds1,mqd1:Nat,w4ab1:Nat)  
       (if caps == mtC
        then mt
        else mkAtt(caps) 
        fi)  .
 \end{verbatim}
 \end{small}

\noindent
The definition of \texttt{tCS3C} is similar, it just
has one more group of device arguments and generates
a third device.

A typical two endpoint scenario has one endpoint with
some messages, \texttt{amsgl0}, to send (a client role) 
and one with initial resources \texttt{rbnds1} (a server role) with message sending delay 5 and \texttt{w4Ack} 
bound of 0.

\begin{small}
\begin{verbatim}
    op tCS : AMsgL RMap Caps -> Conf .
    eq tCS(amsgl0,rbnds1,caps) =    
        tCS2C(0,0,amsgl0,mtR,5,0,
              1,1,nilAM,rbnds1,5,0,
             caps) .
 \end{verbatim}
 \end{small}

A similar scenario but with two servers
(with initial resources \texttt{rbnds1} and \texttt{rbnds2}) is created by the function \texttt{tCSS}.

\begin{small}
\begin{verbatim}
    op tCSS : AMsgL RMap RMap Caps -> Conf .
    eq tCSS(amsgl0,rbnds1,rbnds2,caps) =    
        tCS3C(0,0,amsgl0,mtR,5,0,
              1,1,nilAM,rbnds1,5,0,
              2,1,nilAM,rbnds2,5,0,
             caps) .
 \end{verbatim}
 \end{small}

Finally, we define application message constructors to
use in creating specific scenarios. The functions
\texttt{mkGetC} / \texttt{mkGetN} construct confirmable/non-confirmable \texttt{GET} messages.
The arguments are a message identifier (\texttt{id}),
the receiver (\texttt{tgt}), and the resource path
(\texttt{path}).

\begin{small}
\begin{verbatim}
    ****         Id     Tgt   Resource
    op mkGetC : String String String -> AMsg .
    op mkGetN : String String String -> AMsg .
    eq mkGetC(id,tgt,path) =
         amsg(id,tgt,"CON","GET",path,"",mtBody) .
    eq mkGetN(id,tgt,path) =
         amsg(id,tgt,"NON","GET",path,"",mtBody) .
 \end{verbatim}
 \end{small}

\noindent
Similarly the functions \texttt{mkPutC} / \texttt{mkPutN} construct confirmable/non-confirmable \texttt{PUT} messages with an additional argument, the value \texttt{(val)} to
assign to the resource.

\begin{small}
\begin{verbatim}
    ****         Id    Tgt    Resource  Val
    op mkPutC : String String String  String -> AMsg .
    op mkPutN : String String String  String -> AMsg .

    eq mkPutC(id,tgt,path,val) =
         amsg(id,tgt,"CON","PUT",path,"",b(val)) .
    eq mkPutN(id,tgt,path,val) =
         amsg(id,tgt,"NON","PUT",path,"",b(val)) .
 \end{verbatim}
 \end{small}

\noindent 
The function \texttt{mkDelN} creates a non-confirmable
\texttt{DELETE} message with the same arguments as the
\texttt{GET} messages.

\begin{small}
\begin{verbatim}
    ****         Id    Tgt    Resource  
    op mkDelN : String String String -> AMsg .
    eq mkDelN(id,tgt,path) =
         amsg(id,tgt,"NON","DELETE",path,"",mtBody) .
 \end{verbatim}
 \end{small}

In order to search for attacks based on order of
events (for example, message receives) we introduce
a Log configuration element. 
\begin{small}
\begin{verbatim}
    sort Log .   subsort Log < ConfElt .
\end{verbatim}
\end{small}
    
\noindent
A log is a list of log items specified as follows    
\begin{small}
\begin{verbatim}
    sorts LogItem LogItemL .
    subsort LogItem < LogItemL .
    op nilLI : -> LogItemL [ctor] .  --- the empty list
    op _;_ : LogItemL LogItemL -> LogItemL 
            [ctor assoc id: nilLI] .  --- list concatenation

op log : LogItemL -> Log [ctor] .  --- the log entry
\end{verbatim}
\end{small}

\noindent 
Receive of a  request to \texttt{PUT} value \texttt{val}
at resource path \texttt{path} by endpoint
with identifier \texttt{epid} is represented by
the log item 
\begin{small}
\begin{verbatim}
        rcvP(epid,path,val) .
\end{verbatim}
\end{small}

\noindent
Logging is currently only used by the \texttt{rcvPut} function
that adds the attribute

\begin{small}
\begin{verbatim}
  toLog(rcvP(epid,path,val))
\end{verbatim}
\end{small}
\noindent
to the returned attribute set.   The receive rule  \texttt{crl[rcv]}
is augmented with the conditional clauses   

\begin{small}
\begin{verbatim}
    /\ conf1 := doLog(conf,devatts1)
    /\ devatts2 := clearToLog(devatts1) 
\end{verbatim}
\end{small}
\noindent

\noindent
to add any logitems from a \texttt{toLog} attribute to
the log configuration element, and remove any
\texttt{toLog} attribute from the device attributes.

To support defining properties based on ordering of events
we define a function \texttt{findRcvLI} that given a log
item list, a starting index, and a pattern (three string
variables) returns a pair (sort \texttt{LogItemIx})
consisting of the first matching event on or after the
starting index, together with the index of the event in
the list.

\begin{small}
\begin{verbatim}
    sort LogItemIx .
    op `{_`,_`} : LogItem Nat -> LogItemIx [ctor] .

    op findRcvLI : LogItemL Nat String String String ->  [LogItemIx] .
    op findRcvLIX : LogItemL String String String  Nat ->  [LogItemIx] .

    eq  findRcvLI(lil:LogItemL,n,epat,ppat,vpat)  
        = findRcvLIX(nthCdr(lil:LogItemL,n),epat,ppat,vpat,n) .

    eq findRcvLIX(li:LogItem ; lil:LogItemL,epat,ppat,vpat,n)
       =
      (if matchesLI(li:LogItem,epat,ppat,vpat)
       then {li:LogItem,n}  
       else findRcvLIX(lil:LogItemL,epat,ppat,vpat,s n) fi) .
\end{verbatim}
\end{small}
\noindent
The \texttt{nthCdr(lil:LogItemL,n)} function returns the
suffix of the logitem list \texttt{lil:LogItemL} beginning
at the nth element (counting from 0). The logitem
\texttt{rcvP(epid,path,val)} matches the pattern 
\texttt{(epat,ppat,vpat)} if
corresponding elements of the triple match. A string matches
a pattern string if the two are equal or the pattern is
\texttt{""}.  An error value results if no matching event is
found.

For example to check that a request to \texttt{"dev1"} to
unlock its door resource was followed by one locking that
door we can use the property

\begin{small}
\begin{verbatim}
    subLIL(c:Conf,rcvP("dev1","door","unlock") ;
                  rcvP("dev1","door","lock"))  
\end{verbatim}
\end{small}
\noindent
where \texttt{subLIL} checks that there is a sublist
of the configuration log element that matches
the second argument logitem pattern list.
It is defined using the \texttt{findRcvLI} function
as follows.

\begin{small}
\begin{verbatim}
    op subLIL : Conf LogItemL -> Bool .
    eq subLIL(c:Conf log(litl),plitl) = subLILX(litl,plitl,0) .
    op subLIXL : LogItemL LogItemL Nat -> Bool .
    ceq subLILX(litl,rcvP(epat,ppat,vpat) ; plitl,n) =
          subLILX(litl, plitl,n0) =
     if {rcvP(epid0,path0,val0), n0} :=
             findRcvLI(litl,n,epat,ppat,vpat) 
    eq subLILX(litl,nilLI,n) = true .
    eq subLILX(litl,plitl,n) = false [owise] .
\end{verbatim}
\end{small}

\subsection{Sample scenarios.}
\label{subsec:sample-scenarios}

To illustrate the use of the CoAP specification, and
introduce some of the attacks of interest we define some
small scenarios and show how to search for attacks. Attack
models are discussed in more detail and formalized in
Section~\ref{sec:attack-models}.

Consider a scenario with client, \texttt{"dev0"}, that
requests server, \texttt{"dev1"}, to unlock
(\texttt{PUTCDU}), then lock the door (\texttt{PUTCDL}).
In between the resource \texttt{"sig"} is set to
\texttt{"go"} (\texttt{PUTNSG}). After a normal execution
of this protocol, the door is locked at the end.

\begin{small}
\begin{verbatim}
      dev0           dev1
      ----           lock
       o --PUTCDU-->  o 
       o <--ack --    o
       o <--2.04--    o
       o --PUTNSG-->  o  
       o <--2.01--    o
       o --PUTCDL-->  o
       o <--ack --    o
       O <--2.04--    o
\end{verbatim}
\end{small}

We use the function \texttt{tCS} and the application
message constructors described above to create the
initial configuration for this scenario:

\begin{small}
\begin{verbatim}
    iSys0 = {tCS(mkPutC("putCDU","dev1","door","unlock") ;
                 mkPutN("putNSG","dev1","sig","go") ; 
                 mkPutC("putCDL","dev1","door","lock"),   
                 rb("door","lock"),mtC)} .
\end{verbatim}
\end{small}

\noindent
Now we ask: Can the scenario (with no attacker) end with the door unlocked if the attacker has no capabilities?
This is done by searching for a terminal state
reachable from \texttt{iSys0} in which the \texttt{"dev1"}
\texttt{"door"} resource has value \texttt{"unlock"}.

\begin{small}
\begin{verbatim}
    search [1] iSys0 =>! {c:Conf} such that
       checkRsrc(c:Conf,"dev1","door","unlock") .
\end{verbatim}
\end{small} 

\noindent
There is no solution.
Next we ask if the attacker who can drop a message can force the scenario end with the door unlocked?  \texttt{iSys1} adds
the \texttt{drop} capability to the configuration 
\texttt{iSys0}.

\begin{small}
\begin{verbatim}
    iSys1 = {tCS(mkPutC("putNDU","dev1","door","unlock") ; 
                 mkPutN("putNSG","dev1","sig","go") ; 
                 mkPutC("putNDL","dev1","door","lock"), 
            rb("door","lock"),drop)}
\end{verbatim}
\end{small}

\noindent
We search for a final state where the \texttt{"door"}
resource of \texttt{"dev1"} has value \texttt{"unlock"}.
\begin{small}
\begin{verbatim}
    search [1] iSys1 =>! {c:Conf} such that 
        checkRsrc(c:Conf,"dev1","door","unlock") .
\end{verbatim}
\end{small}

There is no solution because even if the attacker drops the
lock request, since the request is confirmable, the client
will retry. If the attacker can drop two messages then it
can force termination with the door unlocked [not shown].

Suppose the attacker can replay a message (i.e. make a copy
and delay it) as in scenario \texttt{iSys2}.

\begin{small}
\begin{verbatim}
    iSys2 = {tCS(mkPutC("putNDU","dev1","door","unlock") ; 
                 mkPutN("putNSG","dev1","sig","go") ; 
                 mkPutC("putNDL","dev1","door","lock"), 
              rb("door","lock"),replay(10))}
\end{verbatim}
\end{small}

\noindent
Can the attacker cause an unlock to follow a lock request?
We search for an execution in which a door lock request
precedes an unlock request.  (We could also require that
the unlock request was the last door request.)
\begin{small}
\begin{verbatim}
    search iSys2 =>! {c:Conf} such that 
        subLIL(c:Conf, rcvP("dev1","door","lock") ;
                       rcvP("dev1","door","unlock") ) .
\end{verbatim}
\end{small}

\noindent
There are two solutions.  The log from the second solution
follows.
\begin{small}
\begin{verbatim}
  log(rcvP("dev1","door","unlock") ; rcvP("dev1","sig","go") ; 
      rcvP("dev1","door","lock") ; rcvP("dev1","door","unlock"))
\end{verbatim}
\end{small}

Now, consider a situation where there are two servers
(\texttt{"dev1"} and \texttt{"dev2"}) with door resources.
The client wants the state of the door at \texttt{"dev1"}.
Can the attacker reroute the  \texttt{"GET"} request so that the client receives the state of the door at \texttt{"dev2"} instead?

Here we use and instance of a general attack capability
\begin{small}
\begin{verbatim}
   mc(tgt,src,b:Bool,act(tpat,spat,d:Nat))
\end{verbatim}
\end{small}
\noindent
This capability applies to messages to \texttt{tgt} from
\texttt{src}. It changes the target and source according to
\texttt{tpat} and \texttt{spat} (\texttt{""} means no
change) and adds \texttt{d:Nat} to the delay. If the flag
\texttt{b:Bool} is \texttt{true} the edited message replaces
the original, while if the flag \texttt{b:Bool} is
\texttt{false} the original message is unchanged and the
edit is a copy.

\begin{small}
\begin{verbatim}
      dev0         eve           dev1        dev2
      ----         ---          unlock       lock
       o -GETN(1)-> @ [-GETN(1)->] ?         
       |            o               -GETN(2)->o
       | <-lock(1)- @ <---lock(2)--            o 
       |          [<--unlock(1)--] o
       o <-oneof{lock,unlock}-
\end{verbatim}
\end{small}

\texttt{iSys3a} is the scenario for the case when the attacker modifies a message in transit.  Note that the attacker must
redirect both the request (changing the target) and the response
(changing the source).
\begin{small}
\begin{verbatim}
    iSys3a = tCSS(mkGetN("getN","dev1" ,"door"), 
                  rb("door","unlock"),rb("door","lock"), 
                 mc("dev1","dev0",true,act("dev2","dev0",0)) 
                 mc("dev0","dev2",true,act("dev0","dev1",0))) .
\end{verbatim}
\end{small}
\noindent
To check for an attack we search for a final state
in which the door resource of \texttt{"dev1"} is
unlocked, but \texttt{"dev0"} receives a response
(apparently) from \texttt{"dev1"} with the value \texttt{"lock"}.
\begin{small}
\begin{verbatim}
    search iSys3a =>! {c:Conf} such that 
        checkRsrc(c:Conf,"dev1","door","unlock") and 
        hasGetRsp(c:Conf,"dev0","dev1","getN","lock") .
\end{verbatim}
\end{small}

\noindent
There are 4 solutions of 33 states visited. The log is empty
because the model currently does not currently log
\texttt{"GET"} requests.

Suppose the attacker can not change messages in transit,
but can make and edit copies.  This is modelled by setting
the flag in the capabilities to \texttt{false}, as 
in the scenario \texttt{iSys3r}.
\begin{small}
\begin{verbatim}
    iSys3r = tCSS(mkGetN("getN","dev1" ,"door"), 
                  rb("door","unlock"),rb("door","lock"), 
                  mc("dev1","dev0",false,act("dev2","dev0",0)) 
                  mc("dev0","dev2",false,act("dev0","dev1",0))) .
\end{verbatim}
\end{small}

\noindent
We repeat the search for attacks starting with \texttt{iSys3r} 
using this more restricted attack.

\begin{small}
\begin{verbatim}
    search iSys3r =>! {c:Conf} such that 
        checkRsrc(c:Conf,"dev1","door","unlock") and 
        hasGetRsp(c:Conf,"dev0","dev1","getN","lock") .
\end{verbatim}
\end{small}

\noindent
There are also 4 solutions. In this case 109 states
were visited, indicating that the probability of
attacker success is less when it can only make copies
and not change the original message in transit.

\section{Attack Models}\label{sec:attack-models}

We consider two main classes of attacker:
active and reactive. The active attacker models the
attacker envisioned in \cite{coap-attacks}. Using the
formal representation of an active attacker we
replicated versions of the attacks postulated in
\cite{coap-attacks} (see Appendix
\ref{apx:coap-vulnerabilities}). These attack
capabilities (except for the rerouting extension, A1,
see below) simply amplify the effects of an
unreliable transport, dialecting can not mitigate
such attacks, and applications need to deal with the
unreliability of the network independently of
attackers.

The reactive attacker can not modify messages in
transit (can not drop or delay these messages). But
it can observe, make copies of messages in transit,
and transmit modified copies (redirecting, resending
with delay). These capabilities violate the expected
network guarantee that if an endpoint, \emph{ep1}
receives a message with source \emph{ep0}, then there
is a unique previous event in which endpoint
\emph{ep0} sent that message. They are also
capabilities that can be mitigated by dialecting as
we show in Section \ref{sec:dialect-properties}.

The \texttt{drop} capability of scenario
\texttt{iSys1} of
Section~\ref{subsec:sample-scenarios} is an example
of an active attack. The \texttt{replay} capability
of scenario \texttt{iSys2} can be used by active or
reactive attackers. The flag in the redirection
capability of scenarios \texttt{iSys3a} and
\texttt{iSys3r} determines whether the attack is
active (flag is \texttt{true}) or reactive (flag is
\texttt{false}). This is formalized in Section
\ref{subsec:attack-model-spec}

\subsection{Active attacker}

A simple active attacker (A0) has a finite set of single
actions from the capability set $\mathtt{CapsA0} =
\{\mathtt{drop}, \mathtt{delay(n)}\}$, where \texttt{drop}
(also called \texttt{block}) prevents a message from being
delivered, and \texttt{delay(n)} delays delivery of a
message by $n$ time units. $\mathtt{CapsA1}$ extends
$\mathtt{CapsA1}$ with two additional capabilities
$\{\mathtt{replay(n)}, \mathtt{reroute(tpat,spat,n)}\}$.
Applying \texttt{replay(n)} makes a copy of the message
and transmits it with $n$ added to its delay, while
\texttt{reroute(tpat,spat,n)} modifies the message target
and source according to \texttt{tpat} and \texttt{spat}
respectively and adds $n$ to the delay. The pattern
\texttt{tpat} is either a wild card (the empty identifier)
which leaves the target unchanged, or a specific endpoint
identifier that replaces the original. Similarly for the
action of \texttt{spat} on the message source.

Informally, these attack capacities can be used to 
achieve the following basic effects on CoAP messaging.
\begin{itemize}

\item \texttt{drop} can be used to prevent a request (for
some action) from being received, or to prevent a response
to some request. In the response case this leaves the client
is a state of not knowing the requested data value, or if
the requested action happened. Note that due to CoAPs
reliability mechanism, if the target message is confirmable,
the drop attack will need to be applied to all retries to
succeed.

\item The \texttt{delay} of a request, say \texttt{DM1} from
a sequence \texttt{DM1 ... DM2} until after \texttt{DM2} can
undo the effect of \texttt{DM2} (door unlock after door
lock), or could cause \texttt{DM2} to happen in the wrong
context, for example without lights turned on.
    
\item  The \texttt{delay} of a response, say to  \texttt{DM1} delivered as response to a later 
(blocked) request \texttt{DM2} could cause the requesting
client to receive out-of-date data, or to think the action
requested by \texttt{DM2} happened, when only the action of
\texttt{DM1} happened.  If \texttt{DM1}, \texttt{DM2} are \texttt{GET} requests for different resources, say r1 and r2,
the client may interpret the response received as the value of r2 when it is the value of r1.
We note that this only happens with weak implementation of the token mechanism used to link requests and responses.

\item Rerouting a request from ep0 to ep1 to a request from
ep0 to ep2 and (un)rerouting the response from ep2 to be from ep1 causes the client ep0 to interpret the response as the value of requested resource at
ep1 when it is the value at ep2.  
 
\end{itemize}

\noindent
The capability set \texttt{CapsA0} is sufficient to
demonstrate all the vulnerabilities studied in 
Appendix \ref{apx:coap-vulnerabilities},
except for the modified final example  that uses rerouting.

\subsection{Reactive attacker}

A reactive attacker can not modify or block a message
in transit. It can see/copy a message, change the
source/target or the delay amount and transmit the
result.

A simple reactive attacker has capabilities in the set
$\mathtt{CapsR0}$ that allows the attacker to make (and
edit) one copy of a target message. A given attacker
instance is limited to a finite number of attack attempts.

We also consider a multi action reactive attacker,
with capabilities in the set $\mathtt{CapsR1}$. Using
a multi-action, an attacker can make multiple copies
per message, optionally adding to the delay and/or
modifying the message source or target. We allow the
attacker to restrict attention to a given target-source
pattern. This mainly improves the attacker efficiency
(reduces the search space) but adds no power at the
granularity we consider. As for the simple reactive
attacker, each attacker instance is limited to a
finite set of multi-action capabilities. In a little
more detail a multi-action capability consists of a
condition \texttt{match(tpat,spat)} constraining the
set of messages to be attacked, and a finite set of
attack actions \texttt{act(tpat,spat,n)}. As before,
\texttt{(tpat,spat)} can be wild cards (represented by
the empty string) or specific endpoint identifiers. In
the \texttt{match} case a wild card matches any
identifier, while specific identifiers require
equality to match. In the case of \texttt{act}, a
wild card means the corresponding target or source is
not changed, while a specific identifier causes the
corresponding target or source to be replaced by the
pattern element. The copy is further delayed by the
amount given by the number $n$.

The simple reactive attacker is a special case of the multi-action reactive attacker, where the condition is
restricted to \texttt{match("","")} and the \texttt{capset}
applied to a message is a singleton.
Although a reactive attacker can't block or delay messages
in transit,  many of the active attack effects can be
achieved.

Here are examples of what a reactive attacker can do.
\begin{itemize}
\item \emph{R1. Undo/revert an action}. Suppose messages \texttt{M1}, \texttt{M2} PUT different values for the same
resource.  The attacker copies \texttt{M1} and sends the
copy with sufficient delay to arrive after \texttt{M2}, thus overriding the effect of \texttt{M2}.  
This emulates some of the reordering capability of the
active attacker, leaving the resource in a state other
that what the client planned.  It could for example
leave a process running that should have stopped.
If the client then starts a process by PUTting a different
resource in the state \texttt{"on"} after sending \texttt{M2}, the new process may run in an unexpected
context.

\item \emph{R2. Violate ordering or concurrency constraints.}
The attacker observes a message \texttt{M1} from \texttt{ep0} to \texttt{ep1}.
It makes copies redirected to \texttt{ep2} (\texttt{ep3} \ldots).  A compliant CoAP endpoint \texttt{ep0} will
ignore responses from \texttt{ep2} (\texttt{ep3} \ldots). 
But in the case of requests causing actions, there will be
unexpected actions.  For example, if \texttt{ep0} is 
activating \texttt{ep1}, \texttt{ep2}, \ldots, in sequence
then \texttt{ep2}, \ldots, will be acting out of sequence,
concurrently with \texttt{ep1},
which could lead to undesired consequences.

\item \emph{R3. Duplication of a process.}
Suppose as above \texttt{ep0} is coordinating a process by
activating \texttt{ep1}, \ldots, \texttt{epk} in sequence,
using messages \texttt{M1}, \ldots, \texttt{Mk}.
Suppose also that \texttt{ep0x}, \texttt{ep1x}, \ldots,
\texttt{epkx} are a functionally equivalent set of
endpoints, say in a different location.  Then the
attacker can copy each message and redirect to
\texttt{ep1x}, \ldots \texttt{epkx}.

\item \emph{R4. Spoofing (Redirect GET request/response).}
The attacker observes a message \texttt{M1} from
\texttt{ep0} to \texttt{ep1}. It makes a copy
redirected to \texttt{ep2}. A compliant CoAP endpoint
will reject the response from \texttt{ep2} (wrong
server). If the attacker also makes a copy of the
response, rerouting it to appear to be from
\texttt{ep1} then ep0 will accept which ever response
arrives first. If \texttt{M1} is a GET request, then
the client may receive the wrong value for the
requested resource.

\end{itemize}

R1 only needs one capability instance with one action.
R2 only needs one capability instance but with 2 or more
actions.
R3, R4 need multiple capability instances, each one
needs only a single action.

\subsection{More possibilities}
There are additional, stronger, attack models that could
be considered in the future.  One is a
reactive attacker with memory.  Here the attacker sees
all messages and can remember some.  Then at some point
the attacker can make edited copies of some of the
messages in its knowledge base and transmit them, for example
to rerun an application sequence either with the
original set of endpoints at a later time, or a separate
set (with the same behavior) at some point.  The attacker
may use observation of a message occurrence to trigger
actions, or time elapse.

A creative reactive attacker can additionally create new
messages and receive responses, i.e. becoming a rogue
endpoint.

\subsection{Attack-model specification.}
\label{subsec:attack-model-spec}

As part of the CoAP specification, we specified a
general attack framework where an attacker is endowed
with a finite set of \texttt{capabilities}. In this
section we specify specific capabilities 
corresponding to the active and reactive attack
models discussed above.  We define
capability constructors and define their semantics by
giving equations for the \texttt{doAttack} function
for each capability. Different models are obtained by
constraining the capability attribute of the attacker
agent.

We define a generic capability construction,
\texttt{mc(tpat,spat,active?,acaps)}.
\texttt{tpat},\texttt{spat} are target, source patterns
used to restrict the set of messages to be `attacked'.
\texttt{acaps}is a possibly empty set of actions. Each
action, \texttt{act(tpat1,spat1,d)}, causes a copy of
the matched message to be made, transformed, and
transmitted. \texttt{tpat1},\texttt{spat1} are
target, source patterns used to determine he target and
source of the transformed message. If a pattern is
\texttt{""} the original target or source is used,
otherwise the pattern string is used. The delay
\texttt{d} is added to the delay of the original message. If the
boolean \texttt{active?} is \texttt{true}, this is an
active capability, and the original message is removed
from the network, only the copies generated by the
actions are transmitted. If \texttt{active?} is
\texttt{false} the capability is reactive and the
original message is left in the network to be
transmitted as usual. This is formalized by the
following definitions.

\begin{small}
\begin{verbatim}
    ****      tpat   spat  delay
    op act : String String Nat -> Cap .
    ****     tpat  spat  active? actions
    op mc : String String Bool Caps -> Cap .
\end{verbatim}
\end{small}

The function \texttt{doAttack} defines the semantics
of each attack capability. It is given the attack
agent attributes, including its knowledge base of
messages seen, the attack target message,
\texttt{dmsg}, and a multi capability. If the target
message matches the selection pattern (using
\texttt{pmatch}) then the \texttt{toSend} attribute is
computed and returned along with the knowledge base
with the attacked message added. The \texttt{toSend}
attribute includes the attacked message iff the flag
\texttt{b} is false. It also includes the results of
\texttt{applyCaps(dmsg,acaps,mtDM)}, that gives the
semantics of each action.

\begin{small}
\begin{verbatim}
    ****                tgt  act   result
    op doAttack : Attrs DMsg Cap -> Attrs .
    ceq doAttack(attrs kb(dmsgs),dmsg,mc(tpat,spat,b,acaps))
           = toSend(dmsgs1) kb(dmsgs dmsg) attrs 
     if msg @ n := dmsg
     /\ pmatch(getTgt(msg),tpat)
     /\ pmatch(getSrc(msg),spat)
     /\ dmsgs0 := applyCaps(dmsg,acaps,mtDM)
     /\ dmsgs1 := dmsgs0 (if b then mtDM else dmsg fi) .
 \end{verbatim}
 \end{small}
\noindent
If the match fails doAttack returns a result that does not
match the expectation of the attack rule, and the rule
does not fire.  

The function \texttt{applyCaps} makes a  transformed copy
of the attacked message for each action in \texttt{acaps}.
The transformation sets the message target and source
according to the action patterns, and increments the
delay,
\begin{small}
\begin{verbatim}
**** multi cap attack
    op applyCaps : DMsg Caps DMsgS -> DMsgS .
    eq applyCaps(dmsg,mtC,dmsgs) = dmsgs .
    eq applyCaps(msg @ n,act(tpat,spat,n0) acaps,dmsgs) =
          applyCaps(msg @ n,acaps, dmsgs 
          (setTgtSrc(msg,tpat,spat) @ (n + n0)) ) .
\end{verbatim}
\end{small}

For early experiments and readability we defined a number
of specific (active) attack capacities.  We list them
here with their definitions as instances of the \texttt{mc} construction.
\begin{itemize}
\item  \texttt{drop = mc("","",true,mtC)} --- drop the target message
\item  \texttt{delay(n) = mc("","",true,act("","",n))} --- add n to the target message delay
\item  \texttt{divert(dst0,dst1) = mc(dst0,"",true,act(dst1,""))} --- reroute a target message with destination \texttt{dst0} to \texttt{dst1} 

\item  \texttt{undivert(src0,src1)  = mc("",src0,true,act("",src1,0))} --- make a target message from \texttt{src0} appear to be from \texttt{src1}

\item  \texttt{replay(n)= mc("","",false,act("","",n))} --- replay the target message 
  after delay n  (the original message is delivered as 
  expected)
\end{itemize}

\section{Example reactive attacks}
\label{sec:reactive-attacks}

This section describes experiments
using scenarios illustrating the 4 classes of
reactive attack discussed in Section \ref{sec:attack-models}.

\subsection{R1. Redo/revert an action}. 

The R1 scenario is a  variant of the attacks
on CoAP vulnerabilities Figures 3,4,5,6 (see Section \ref{apx:coap-vulnerabilities}).  The basic requirement
is that the door remains unlocked.  The attacker
aims for the door to be locked at the end, or before
the optional signal message.

\begin{small}
\begin{verbatim}
        dev0        eve        dev1
        ----        ---        unlock
         o -PUTNDL-> -----------> 
                    @ act("","",5/10)   
         o <-------  <-2.04--    o
          ....
         o -PUTNDU-> ----------> o 
         o <-------  <-2.04--    o
                    --- PUTNDL-> o  --- attack
    optional
         o -PUTNSO->  ------>    o  
         o <-------  <-2.01--    o
\end{verbatim}
\end{small}

The function \texttt{raR1} generates instances of this scenario.
It is parameterized by the delay between application message
sends (\texttt{mqd:Nat}), the bound on pending acknowledgements
that blocks sending an application message (\texttt{w4b:Nat}), the delay to be added to an replayed message (\texttt{d:Nat}),
and a boolean (\texttt{nso:Bool}) specifying whether the optional
\texttt{"signal"} message should be sent.  
\begin{small}
\begin{verbatim}
    op raR1 : Nat Nat Nat Bool -> Conf .
    eq raR1(mqd:Nat,w4b:Nat,d:Nat,nso:Bool) = 
     tCS2C(0,0,mkPutN("putNDL","dev1","door","lock") ;
               mkPutN("putNDU","dev1","door","unlock") ;
             (if nso:Bool 
              then mkPutN("putNS","dev1","signal","on")
              else nilAM fi) ,
              mtR,mqd:Nat,w4b:Nat,
          1,1,nilAM,rb("door","unlock")rb("sig","off"),2,0,
          mc("dev1","dev0",false, act("","",d:Nat))) .
\end{verbatim}
\end{small}

\noindent 
Can the attacker cause the execution to end with the door
locked? To check, search for a terminal state in which
\texttt{"dev1"} has resource \texttt{"door" }with value
\texttt{"lock"}.

\begin{small}
\begin{verbatim}
    search {raR1(5,0,10,false)} =>! {c:Conf} such that 
        checkRsrc(c:Conf,"dev1","door","lock") .
\end{verbatim}
\end{small}
\noindent
There are 2 solutions in 101 states visited.

To check that the attacker can shut the door before
a later event, we turn on the signal option and search
for a final state where the lock event happened before
the signal event.
\begin{small}
\begin{verbatim}
    search {raR1(5,0,10,true)} =>! {c:Conf} such that 
        checkRsrc(c:Conf,"dev1","door","lock") and 
        subLIL(c:Conf,rcvP("dev1","door","unlock") ;  
                      rcvP("dev1","door","lock") ; 
                      rcvP("dev1","signal","")) .
\end{verbatim}
\end{small}
\noindent
There are 2 solutions from 231 states visited.  The log
of the second solution is the following.
\begin{small}
\begin{verbatim}
    log(rcvP("dev1","door","lock") ; 
        rcvP("dev1", "door", "unlock") ; 
        rcvP("dev1", "door", "lock") ; 
        rcvP("dev1", "signal", "on"))
\end{verbatim}
\end{small}
\noindent
If the attacker uses a delay smaller than 10
or larger than 10 there is no attack (search yields
no solution).

\subsection{R2. Violating ordering or concurrency constraints}

In this example, a sequence of $n$ tasks is to be executed one
after the other. The protocol consists of a controller
(client endpoint) and n servers that execute the tasks. A
task is initiated when the server receives a \texttt{PUT}
\texttt{"sig"} \texttt{"on"} request and terminates when
he server receives a \texttt{PUT} \texttt{"sig"}
\texttt{"off"} request.

\begin{small}
\begin{verbatim}
    dev0        eve        dev1  ... devk
    ----        ---        off
     o -PUTNon  ----------> o           
                 @(cc dev2)> ----->  
     o <-------  <-2.04--   o
     ...                       
     o -PUTNoff  ----------> o   
     o <-------  <-2.04--    o
      ....
      ....
     o -PUTNon-> ---------->         o 
     o <-------  <-2.04---------     o
     o -PUTNoff-> -----------------> o 
     o <-------  <-2.04---------     o
\end{verbatim}
\end{small}

\noindent
The function \texttt{iSysX} generates instances of
the above scenario.  \texttt{n} is the number of
servers, \texttt{d} the task duration, and \texttt{caps}
the attacker capabilities.
\begin{small}
\begin{verbatim}
    op iSysX : Nat Nat Caps -> Sys .
    eq iSysX(n,d,caps) = {CnS(n,mkSigAMs(n,d,nilAM),rb("sig","off"),5,0) 
      mkAtt(caps) } . 
\end{verbatim}
\end{small}

\noindent
It uses the function
\texttt{CnS(n,amsgl,rbnds,mqd,w4ab)} that generates a
configuration with one client and n servers. The
client's application message list is \texttt{amsgl} and
each server has initial resources given by the RMap
\texttt{rbnds}. The client delay between message sends
is \texttt{mqd} and it can only send a new message if
there are no more than \texttt{w4ab} confirmable
messages awaiting an \texttt{ACK}.

The function \texttt{mkSigAMs(n,d,nilAM)}
generates the list of application messages
for the client to send to control the 
execution. The list consists of pairs
\begin{small}
\begin{verbatim}
   mkPutN("putN","dev" + string(j,10),"sig","on") ;
   mkPutN("putN","dev" + string(j,10),"sig","off") .
\end{verbatim}
\end{small}
\noindent
for $j \in [1,n]$ separated by a delay of \texttt{d}.

We consider three levels of attacker capability.
Level $j$ (\texttt{caps-j}) looks for messages from
\texttt{"dev0" }to \texttt{"dev1"} and immediately
(additional delay 0) sends a copy to \texttt{"dev2"} ... \texttt{"dev1+j"}.  
\begin{small}
\begin{verbatim}
    ops caps-1 caps-2 caps-3 : -> Caps .
    eq caps-1 = mc("dev1","dev0",false,act("dev2","",0)) .
    eq caps-2 = mc("dev1","dev0",false,
                   act("dev2","",0) act("dev3","",0)) .
    eq caps-3 = mc("dev1","dev0",false,
                   act("dev2","",0) act("dev3","",0) 
                   act("dev4","",0)) .
\end{verbatim}
\end{small}

\noindent
We consider two attacks that violate the property
that the tasks do not interleave, by starting a task
on \texttt{"dev-k+j"} before the task on \texttt{"dev-k"} 
completes.
One check for violation is
that the number of devices with resource binding
\texttt{rb("sig","on")} is great than 1. 
The function \texttt{epswrb(conf,rbnds)} returns
the number of servers with resource map containing
a match for \texttt{rbnds} is used for this check.
Alternately we can check for executions in which
\texttt{"devj"} receives an \texttt{"on"} signal, and
\texttt{"devj+1"} receives the \texttt{"on"} signal before \texttt{"devj"} receives
the \texttt{"off"} signal.
This is done using the \texttt{subLI} function.
Using \texttt{caps-1} can the attacker can cause
at least 2 simultaneous tasks executions?
\begin{small}
\begin{verbatim}
    search iSysX(3,0,caps-1) =>+ sys:Sys such that 
        size(epswrb(sys:Sys,rb("sig","on"))) > 1 .
\end{verbatim}
\end{small}

\noindent
There are 132 solutions from 767 states visited.

Raising the bar we ask if using \texttt{caps-2} can the attacker cause 3 or more
tasks to execution simultaneously.
\begin{small}
\begin{verbatim}
    search  iSysX(3,0,caps-2) =>+ sys:Sys such that 
        size(epswrb(sys:Sys,rb("sig","on"))) > 2 .
\end{verbatim}
\end{small}
\noindent
There are 182 solutions among 721 states visited.

We can also ask if the attacker can cause
dev2 to start before dev1 finishes.
\begin{small}
\begin{verbatim}
    **** attack 3 dev2 starts before dev1 finishes
    search iSysX(3,0,caps-2) =>+ {c:Conf} such that 
        subLIL(c:Conf, rcvP("dev1","sig","on") ; 
                       rcvP("dev2","sig","on") ; 
                       rcvP("dev1","sig","off")) .
\end{verbatim}
\end{small}

\noindent
There are 342 solutions from 3598 states visited.
The log for the last solution is
\begin{small}
\begin{verbatim}
   log(rcvP("dev3", "sig", "on") ; rcvP("dev1", "sig", "on") ; 
    rcvP( "dev2", "sig", "on") ; rcvP("dev1", "sig", "off") ; 
    rcvP("dev2", "sig", "on") ; rcvP("dev2", "sig", "off") ; 
    rcvP("dev3", "sig", "on") ; rcvP( "dev3", "sig", "off")).
\end{verbatim}
\end{small}

\noindent
Finally, we ask if the attacker can cause 
dev2 to start before dev1 finishes and
dev3 to start before dev2 finishes.

\begin{small}
\begin{verbatim}
    search iSysX(3,0,caps-2) =>+ {c:Conf} such that 
        subLIL(c:Conf, rcvP("dev1","sig","on") ; 
                       rcvP("dev2","sig","on") ; 
                       rcvP("dev1","sig","off")) and 
        subLIL(c:Conf,rcvP("dev2","sig","on") ; 
                      rcvP("dev3","sig","on") ; 
                      rcvP("dev2","sig","off")) .
\end{verbatim}
\end{small}
\noindent
There are 62 solutions among 3594 states visited.
The log for the last solution follows.
\begin{small}
\begin{verbatim}
  log(rcvP("dev1", "sig", "on") ; rcvP("dev2", "sig", "on") ; 
  rcvP( "dev3", "sig", "on") ; rcvP("dev1", "sig", "off") ; 
  rcvP("dev2", "sig", "on") ; rcvP("dev2", "sig", "off") ; 
  rcvP("dev3", "sig", "on") ; rcvP( "dev3", "sig", "off"))
\end{verbatim}
\end{small}

\subsection{R3 Duplication of a process.}

IoT systems may be replicated, for example running
the same protocol in multiple smart rooms, or 
in a manufacturing setting multiple instances of
a manufacturing step.  

\begin{small}
\begin{verbatim}
    dev0        eve        dev1 dev2 dev3 dev4 dev5 dev6
    ----        ---        ---
     o -PUTNgo   @ ---------> o           
                 o -(cc dev4)-------------> o 
     o <-------  <-2.04--   o    
                                 <-2.04--   o
     o -PUTNgo-> @-------------->  o 
                 o -(cc dev5) ---------------->  o 
     o <-------  <-2.04---------  o
                                       <-2.04--  o
     o -PUTNgo   @ -------------------> o 
                 o -(cc dev6)----------------------> o 
     o    <-------  <-2.04---------    o
                                            <-2.04-- o
\end{verbatim}
\end{small}
\noindent 
Here the label \texttt{(cc dev4)} indicates the attacker
coping a message to a different device.

The function \texttt{iSysY} generates example scenarios
with one client (controller) and 2 sets of n servers.
The client will send a \texttt{PUT} message to set the
\texttt{"sig"} resource \texttt{"go"} to each server
in the first set of servers.  Each server initially
has resource map  \texttt{rb("sig","off")}. 
\begin{small}
\begin{verbatim}
    op iSysY : Nat  Caps -> Sys .
    eq iSysY(n,caps) = 
          {CnS(2 * n,mkGoAMs(n,nilAM),rb("sig","off"),5,0) 
           mkAtt(caps)  log(nilLI)} . 
\end{verbatim}
\end{small}

We consider scenarios with $n=2$ and $n=3$ and ask
whether the attacker can drive a the second set
of devices, either concurrently or following the
first.  
The attacker capabilities for the two scenario sizes
are defined as follows, parameterized by the
amount to delay message copies.  
\begin{small}
\begin{verbatim}
    ops caps2-2 caps3-3 : Nat -> Caps .
    eq caps2-2(d) =
      mc("dev1","dev0",act("dev3","dev0",d))
      mc("dev2","dev0",act("dev4","dev0",d)) .
    eq caps3-3(d) =
      mc("dev1","dev0",act("dev4","dev0",d))
      mc("dev2","dev0",act("dev5","dev0",d)) 
      mc("dev3","dev0",act("dev6","dev0",d)) .
\end{verbatim}
\end{small}
  
The four search commands, results (with witness log)
are as follows.

\begin{small}
\begin{verbatim}
     *** search for duplicate 2 server process running concurrently   
    search iSysY(2,caps2-2(0)) =>! {c:Conf} such that 
      subLIL(c:Conf, rcvP("dev3","sig","") ; rcvP("dev4","sig","")) and 
       subLIL(c:Conf, rcvP("dev3","sig","") ; rcvP("dev2","sig","")) .
\end{verbatim}
\end{small}

\noindent
There are 16 solutions among 534 states visited.
\begin{small}
\begin{verbatim}
    log(rcvP("dev3", "sig", "go") ; rcvP("dev1", "sig", "go") ; 
        rcvP( "dev4", "sig", "go") ; rcvP("dev2", "sig", "go"))

\end{verbatim}
\end{small}

\begin{small}
\begin{verbatim}
    *** search for duplicate 2 server process the second after the first
    search iSysY(2,caps2-2(15)) =>! {c:Conf} such that 
      subLIL(c:Conf,rcvP("dev3","sig","") ;  rcvP("dev4","sig","")) 
      and subLIL(c:Conf,rcvP("dev2","sig","") ; rcvP("dev3","sig","")) .
\end{verbatim}
\end{small}
\noindent
There are 4 solutions among 179 states visited.

\begin{small}
\begin{verbatim}
    log(rcvP("dev1", "sig", "go") ; rcvP("dev2", "sig", "go") ;
        rcvP( "dev3", "sig", "go") ; rcvP("dev4", "sig", "go"))
\end{verbatim}
\end{small}

\begin{small}
\begin{verbatim}
     **** search for duplicate 3 server process the second after the first
    search iSysY(3,caps3-3(0)) =>! {c:Conf} such that 
      subLIL(c:Conf, rcvP("dev4","sig","") ; 
        rcvP("dev5","sig","") ; rcvP("dev6","sig","")) and 
      subLIL(c:Conf, rcvP("dev1","sig","") ; 
        rcvP("dev4","sig","") ; rcvP("dev2","sig","")) and 
      subLIL(c:Conf, rcvP("dev2","sig","") ; 
        rcvP("dev5","sig","") ; rcvP("dev3","sig","") ; 
        rcvP("dev6","sig","")) .
\end{verbatim}
\end{small}

\noindent
There are 8 solutions among 2845 states visited.

\begin{small}
\begin{verbatim}
    log(rcvP("dev1", "sig", "go") ; rcvP("dev4", "sig", "go") ; 
        rcvP( "dev2", "sig", "go") ; rcvP("dev5", "sig", "go") ; 
        rcvP("dev3", "sig", "go") ; rcvP("dev6", "sig", "go"))
\end{verbatim}
\end{small}

\begin{small}
\begin{verbatim}
    **** search for duplicate 3 server process the second after the first
    search iSysY(3,caps3-3(15)) =>! {c:Conf} such that 
        subLIL(c:Conf, rcvP("dev4","sig","") ; 
            rcvP("dev5","sig","") ; rcvP("dev6","sig","")) and 
        subLIL(c:Conf, rcvP("dev3","sig","") ; rcvP("dev4","sig","")) . 
\end{verbatim}
\end{small}

\noindent
There are 8 solutions among 683 states visited.
The log of the last solution is the following.
\begin{small}
\begin{verbatim}
   log(rcvP("dev1", "sig", "go") ; rcvP("dev2", "sig", "go") ; 
       rcvP( "dev3", "sig", "go") ; rcvP("dev4", "sig", "go") ; 
       rcvP("dev5", "sig", "go") ; rcvP("dev6", "sig", "go"))
\end{verbatim}
\end{small}

\subsection{R4. Spoofing (Redirect GET request/response)}

Although a reactive attacker can not redirect a message in
transit, as the active attacker can (see Section
\ref{subsec:diversion}),  it can make a copy and direct
it to a different server. The reactive attacker also has to
make a copy of the response from that server to appear to
come from the original server. Thus the client will receive
3 responses. They can arrive in any order. The client will
ignore the one from the alternate server, and its a race
between the two messages apparently from the intended
server.

\begin{small}
\begin{verbatim}
        dev0      eve       dev1     dev2
        ----      ---      unlock    lock
      o -GETN(1)-->@------->  o   
      |            o copy(cc dev2)  -> o
      o   <---unlock---      o 
      o            @    <---lock---    o 
      o <- (cc from dev1)
\end{verbatim}
\end{small}
\noindent
The function \texttt{iSySZ(mqd:Nat,w4b:Nat)} generates
instances of the redirection attack scenario. It is
parameterized by the usual congestion control arguments.

\begin{small}
\begin{verbatim}
    op iSySZ : Nat  Nat -> Sys .
    eq iSySZ(mqd:Nat,w4b:Nat) =
         {tCS3C(0,0,mkGetN("getN0","dev1","door"),
                   mtR,mqd:Nat,w4b:Nat,
              1,1,nilAM,rb("door","unlock"),5,0,
              2,2,nilAM,rb("door","lock"),5,0,
              mc("dev1","dev0",false,act("dev2","",0))
              mc("dev0","dev2",false,act("","dev1",0)) )
              log(nilLI) }  . 
\end{verbatim}
\end{small}

To demonstrate attacks we search for final states in which
the client has received a \texttt{"lock"} response to 
the \texttt{GET} request while \texttt{"dev1"} has the \texttt{"door"}
resource bound to \texttt{"unlock"}.
\begin{small}
\begin{verbatim}
    search iSySZ(5,0) =>! {c:Conf} such that 
      hasGetRsp(c:Conf,"dev0","dev1","getN0","lock") and 
      checkRsrc(c:Conf,"dev1","door","unlock") .
\end{verbatim}
\end{small}

\noindent
There are 4 solutions among 109 states visited.

\section{Dialect functions}
\label{sec:dialect-fns}

Dialects using a moving target strategy must
synchronize on the choice of lingo. In the case of
reliable transport this is often done using time
synchronization. \footnote{Even with reliable transport,
attackers may drop or reorder messages, which can complicate
lingo synchronization.}

For protocols, like CoAP, running on unreliable
transport, the notion of current lingo (or lingo
parameters) or lingo (parameters) being no longer
valid, seems problematic. In principle there is no
bound on the delay of a message, and missing messages
or out of order messages which are common, make
guaranteeing that a receiver can determine the correct
lingo to decode a dialected message a challenge.
Failure would mean dropping a message that should have
been delivered.

We take inspiration from security protocol designs that
use message counters (one counter for each communication
pair and direction) for replay prevention and/or
detection.  These counters appear in the clear in messages making synchronization simpler.
We have to be careful that these counters don't leak
enough information to be useful to an attacker.

For the CoAP dialect specification we propose a scheme of three functions:
\begin{itemize}
  \item $g : \mathtt{String}  \times \mathtt{Nat} \times  \mathtt{Nat} \lra \mathtt{String}$ -- a generator of (pseudo) randomness.  The first argument is a seed, the second
specifies the output length, and the third argument is
the index into the sequence generated by $g$.
Thus $g(\mathtt{seed},\mathit{k},\mathit{ix})$  is the \emph{ix}-th (pseudo) random string (of length \emph{k}) initialized \texttt{seed}.

  \item  $f_1 : \mathtt{String} \times  \mathtt{Content}  \times  \mathtt{Nat} \lra \texttt{DCBits}$ 
    -- the obfuscator/encoder. The first argument is
  the source of randomness, the second argument is a message content and the third an index (a natural number).
  \item
   $f_2 : \mathtt{String} \times (\texttt{DCBits} \times \mathtt{Nat}) \leadsto  \texttt{ContentNat} $ -- the de-obfuscator/decoder. The first argument is again the source of
randomness  and the second argument is a pair consisting
of the encoded message content and the lingo index.
\end{itemize}

\noindent
We can think of a dialect given by the above three functions as having a single parameterized lingo,
or as having a family of lingos indexed by the
third argument to $f1$.

There are two important properties of these functions.
First, $f_2$ recovers the original content encoded by
$f_1$, using random generator $g$
$$ f_2(grand,\lb f_1(grand, \lb \mathtt{content},ix\rb\rb),ix)
 = \lb \mathtt{content},ix\rb.$$
where \emph{grand} is generated independently by each endpoint
in a communicating pair as $g(\mathtt{seed},k,ix)$ where
\texttt{seed} is a shared secret.
Second, if the encoded content or index are modified,
decoding will fail.
Letting 
$f_1(\mathit{grand},\lb \mathit{content},\mathit{ix} \rb) = \mathit{dcbits}$, $\mathit{dcbits1} \neq \mathit{dcbits}$, and $\mathtt{ix1} \neq \mathtt{ix}$
then
$$ f_2(grand,\lb \mathtt{dcbits1},ix\rb)
 = \lb \mathtt{content1},ix\rb.$$
and 
$$ f_2(grand,\lb \mathtt{dcbits},ix1\rb)
 ==\lb \mathtt{content1},ix1\rb.$$
with $\mathtt{content} \neq \mathtt{content1}$ 
and $\mathtt{content1}$  is recognized as ill-formed.
Alternately, $f_2$ could directly return an indication
of failure, but only if the input is not a correctly 
transformed content.

Starting the enumeration of random strings with
a secret seed is important.  
An attacker may well know the three functions.
If he can guess the message content, then
he can compute $f_2(g(ix,k),{\mathit{dcbits},ix}))$.
We claim that with a secret seed  exposing \emph{ix} doesn't really expose much useful information
and avoids the need for the receiving wrapper
to guessing the \emph{ix}.

An alternative for generating the sequence of random
strings is $g^{ix(}seed,k)$ $ seed_0 = g(seed,k,0)$ and
$seed_{j+1} = g(seed_j,k,0)$. If the attacker learns the
value of $g(seed_j,k,0)$ (for example if $f_1$ uses xor
and the attacker can guess the content of some message)
then the attacker can compute $g(seed_{j+n},k,0)$.

What can the attacker learn given the method of
generating a sequence of pseudo random strings/bits? 
Suppose the attacker can derive
$g(\mathit{seed},k,ix)$, $g(\mathit{seed},k,ix+1)$,
\dots, $g(\mathit{seed},k,ix+n)$ from transmitted
messages.  It should be HARD to predict the next
value $g(\mathit{seed},k,ix+1)$ from
$g(\mathit{seed},k,ix)$.

If the attacker knows $g$ but not the secret string  $\mathit{seed}$ he can run $g$ on sample seeds with the
known indexes.  
If seed is 128 bits  with non-0 in the top 64  bits
it can take a long time to check the seed.
If he can do $10^{12}$ samples per day it 
could take $10^6$ days to find the seed.

\omitthis{
      
2^10 1024, 2^20 1mil, 2^80 = 10^24 ~ 10^18 seconds

24 x 3600 = 86400 sec/day  > 10^12 days at 10^6 numbers per day

}

\omitthis{
applyDialect > dc(f1(g(rand,rsize,ix),{c,ix}), ix) 
  
decodeDialect > f2(g(rand,rsize,ix),dc(dcbits,ix))
      
grand == g(rand,rsize,ix)      
f2(grand,dc(f1(grand,{c,ix}),ix)) = {c,ix}   

As an example $f1$ could xor the content,ix pair with
the random string (suppressing the conversion to string of content,ix)

  f1(grand,{c,ix}) = grand xor {c,ix}
and 
 f2(grand,dcbits) = grand xor dcbits
 f2(grand,grand xor {c,ix}) = grand xor grand xor {c,ix}
                 = {c,ix}

}

\section{Specification of a CoAP dialect wrapper}
\label{sec:coap-dialect-spec}

To specify a dialect wrapper for CoAP messaging we need to
specify the data type of dialected messages, the three
functions described in Section \ref{sec:dialect-fns}, the
structure of wrapped CoAP agents (subsection
\ref{subsec:dconf}), and the rules for sending and
receiving messages by the wrapped agents (Subsection
\ref{subsec:drules}). Towards dialects as a theory
transformation, in subsection
\ref{subsec:dialect-transform} we define operations
\texttt{D} and \texttt{UD}. \texttt{D} takes an inital
system configuration and produces the corresponding
dialected configuration. \texttt{UD} extracts the
underlying CoAP system configuration from a dialected
configuration. Using these functions, rewrite and search
commands for a CoAP scenario can be automatically
transformed to rewrite and search commands for the
corresponding dialected scenario (Subsection
\ref{subsec:dscenarios}).

\subsection{Dialect messages and configurations}
\label{subsec:dconf}

A dialected message has a target and a source (sort
\texttt{String} as for normal messages) and a dialect
encoded content (sort \texttt{DContent}). The message
constructor is overloaded to construct dialected
messages.

\begin{small}
\begin{verbatim}
    op m : String String DContent -> Msg .
\end{verbatim}
\end{small}

\noindent
A term of sort \texttt{DContent} is a pair constructed by the operator \texttt{dc}

\begin{small}
\begin{verbatim}
    op dc : DCBits Nat -> DContent [ctor] .
\end{verbatim}
\end{small}

\noindent 
where  the sort \texttt{DCBits} is an opaque sort whose
structure is not further specified.
We also define the sort \texttt{ContentNat}  to be
pairs constructed from ordinary content and a natural number.  

\begin{small}
\begin{verbatim}
    op `{_`,_`} : Content Nat -> ContentNat [ctor] .
\end{verbatim}
\end{small}

\noindent
The three dialect functions of Section \ref{sec:dialect-fns} are specified as follows:

\begin{small}
\begin{verbatim}
    op g : String Nat Nat -> String .  
    op f1 : String ContentNat -> DCBits .
    op f2 : String DContent  ~> ContentNat .
    eq f2(g(rand,k,ix),dc(f1(g(rand,k,ix),{content,ix}),ix))
          = {content,ix} .
\end{verbatim}
\end{small}

\noindent
The partial arrow  $\verb|~>|$ used in the declaration
of \texttt{f2} means that the result of applying \texttt{f2} to a modified encoding will be of kind
\texttt{ContentNat}, but not of sort \texttt{ContentNat}.
This allows to detect modifications in the dialected
content.

As described in \cite{GEMM-2023esorics} we represent
a dialect wrapper by a meta-agent
that has the same identifier as the wrapped agent,
a \texttt{conf} attribute that encapsulates the
base agent and its network stub, and additional
attributes for managing dialect transformations.

\begin{small}
\begin{verbatim}
    *** dialected endpoint
    [eid | conf([eid | devattrs] net), ddevattrs]
\end{verbatim}
\end{small}

\noindent
We refer to the encapsulated network as the internal network and the network at the level of meta entities
as the external network.

The additional attributes for a dialect  meta endpoint  include:
\begin{itemize}
\item
\texttt{used(umap)}  --- a map from endpoint ids to sets of
  lingo indices already received from the identified endpoint.
\item
\texttt{toRcv(dmsgs)}  --- holds the result of decoding-- the original message or the empty (delayed) message set, if decoding fails
\item
\texttt{seedTo(eid,str)} -- the shared secret seed for sending to \texttt{eid}
\item
\texttt{seedFr(eid,str)} -- the shared secret seed for receiving from \texttt{eid}
\item
\texttt{ ixCtr(eid,nat)} --- the counter for generating
   lingo indices for messages to endpoint named \texttt{eid}
\item
\texttt{randSize(k)} -- the size of generated (pseudo) random strings
\end{itemize}

\subsection{Dialect Rules}\label{subsec:drules}

There are two rules to specify dialect behavior.
The rule labeled \texttt{ddevsend} handles applying
a dialect lingo to a message sent by an endpoint
and putting it in the global network.  It selects
a message in the internal network sent by the wrapped
entity, calls \texttt{applyDialect}, and puts the
resulting dialected message in the external network.

\begin{small}
\begin{verbatim}  
    crl[ddevsend]:
      [epid | conf([epid  | devatts ] net(dmsgs0,dmsgs1)) dattrs ]
      net(ddmsgs0,ddmsgs1)
    =>
      [epid | conf([epid  | devatts ] net(mtDM,dmsgs) ) dattrs1 ] 
      net(ddmsgs0 ddmsg,ddmsgs1)
    if dmsgs (msg @ n) := dmsgs0 dmsgs1
    /\ getSrc(msg) = epid
    /\ toSend(ddmsg) dattrs1 := applyDialect(dattrs,msg @ n) .
\end{verbatim}  
\end{small}

\noindent 
The function \texttt{applyDialect} looks up 
seed, random size, and current index values in
the attribute set \texttt{dattrs}.  It computes
the generated string \texttt{grand} using the dialect
function \texttt{g}, and then applies the encoding function
\texttt{f1} to this string, the message content and the
index and constructs the dialected message \texttt{msgd}
to be sent on the external network.  It also increments
the index counter for the message target 
\begin{small}
\begin{verbatim}  
    ceq applyDialect(dattrs,msg @ n) =
         incIxCtr(dattrs,dst,1) toSend(msgd @ n)
    if m(dst,src,content) := msg
    /\ rand := getSeedTo(dattrs,dst)
    /\ rsize := getRandSize(dattrs)
    /\ ix := getIxCtr(dattrs,dst)
    /\ grand := g(rand,rsize,ix)
    /\ dcbits := f1(grand,{content,ix})
    /\ msgd := m(dst,src,dc(dcbits,ix)) .
\end{verbatim}  
\end{small}

The rule labeled \texttt{ddevrcv} handles receipt of a
dialected message. It selects a message with 0 delay and
target the wrapped entity an calls \texttt{decodeDialect}.
If decoding produced a message, the rule calls the \texttt{rcvMsg} directly since the result may include
a logging attribute that needs the full configuration 
to process, which the base CoAP receive rule would not
have.

\begin{small}
\begin{verbatim}  
    crl[ddevrcv]:
    {[epid | conf([epid | devatts ] net(dmsgs0,dmsgs1))
             dattrs ] 
      net(ddmsgs0,ddmsgs1 msgd @ 0) conf}
    =>
    {[epid | conf([epid | devatts2 ]
                  net(dmsgs0 dmsgs2,dmsgs1)) 
            dattrs1 ]
      net(ddmsgs0,ddmsgs1)  conf1}
    if getTgt(msgd) == epid
    /\ toRcv(dmsgs) dattrs1 := decodeDialect(dattrs,msgd)

    /\ toSend(dmsgs2) devatts1 := 
        (if dmsgs :: DMsg
         then rcvMsg(epid, devatts, DMsg2Msg(dmsgs))
         else toSend(mtDM) devatts
         fi)
    /\ conf1 := doLog(conf,devatts1)
    /\ devatts2 := clearToLog(devatts1) } .
\end{verbatim}  
\end{small}
 
In addition to the dialect send and receive rules, the auxiliary functions \texttt{mte} and \texttt{passtime} used by the rule labeled \texttt{tick}
are extended to account for the nesting of configurations. 
If there are messages in the local network then
they need to be processed before time can pass.
These messages will be
messages from \texttt{epid}
the should be dialected and placed in the global network.
\begin{small}
\begin{verbatim}  
  eq mte(atts conf([epid | devattrs] net(dmsgs0,dmsgs1)),ni)
   = (if dmsgs0 dmsgs1 =/= mtDM
      then 0
      else mte(atts, mte(devattrs,ni))
      fi) .
\end{verbatim}  
\end{small}

\noindent
Similarly, the effect of passing time is simply propagated
to the encapsulated agent, ignoring the local network
element.
\begin{small}
\begin{verbatim}  
  eq passTime(conf([epid | devattrs] net(dmsgs0,dmsgs1)),nz)
   = 
  conf([epid | passTime(devattrs,nz,mtA)]  net(dmsgs0,dmsgs1)) .
\end{verbatim}  
\end{small}

\subsection{CoAP dialect transform}
\label{subsec:dialect-transform}

As a step towards realizing the dialects as
theory transformations, we define a function
\texttt{D} that maps a term representing an initial system configuration to its dialected form. \texttt{D} calls the auxilliary function
\texttt{DX} with the system term to transform
and the set of ids of non-attacker endpoints
computed by \texttt{getIds}.  The function
\texttt{DX} uses \texttt{DAs} to transform
the individual agents of the configuration \texttt{aconf},
and adds the initial network element.   
The argument \texttt{"xxxx"} of \texttt{DAs} 
is used to construct abstract representations
of the pairwise shared secrets.   The last
argument of \texttt{DAs} is a configuration
accumulator.   We also define a partial inverse, \texttt{UD}, to the dialect transform \texttt{D}.
The partial inverse is only meaningful when applied
to systems with no dialected messages pending.  
It is generally applied only to inital or terminal
states, which have empty network elements.

\begin{small}
\begin{verbatim}  
    op D : Sys -> Sys .
    op DX : Sys Strings -> Sys .

    eq D(sys) = DX(sys,getIds(sys)) .
    ceq DX({aconf net(mtDM,mtDM)},eids) =
       { daconf net(mtDM,mtDM) }
    if daconf := DAs(aconf,eids,"xxxx",mt) .
\end{verbatim}  
\end{small}
\noindent
The function \texttt{DAs} applies the local
agent transform \texttt{DA} to each agent with
identifier in the set \texttt{eids} and leaves
any other agent (i.e. an attacker) unchanged.
The function \texttt{DA} does the actual
wrapping of a protocol role/endpoint producing a
meta-agent with the same identifier and a
\texttt{conf} attribute containing the original
agent with its local network stubb. 
The attributes of the meta-agent wrapper
consist of the attributes, 
\texttt{sharedDialectAttrs} that are the same
for all wrappers, a set of seed attributes, a pair for each possible communication partner (network link) and a set of index attributes, one for each possible communication partner.
 
\begin{small}
\begin{verbatim} 
    op DA : Agent Strings String -> Agent .
    ceq DA([eid | devatts],eids,str) =
        [eid | conf([eid | devatts] 
                    net(mtDM,mtDM)) ddevatts]
    if ddevatts := sharedDialectAttrs
                   mkSeeds(eid,eids,str,mtA)
                   mkIxCtrs(eids,mtA) . 
\end{verbatim}  
\end{small}

\noindent
To accomodate general scenarios, we assume any
pair of non-attacker agents might communicate.
When called from \texttt{DX} the \texttt{str}
will be \texttt{"xxxx"} and result of \texttt{mkSeeds} will be a set of attribute pairs
\texttt{seedFr(eid,"xxxx" + eid1 + eid)} and
\texttt{seedTo(eid,"xxxx" + eid + eid1)} for
each \texttt{eid} in \texttt{eids}.
 Similarly the function \texttt{mkIxCtrs}  
 produces an attribute \texttt{ixCtr(eid,0)}
 for  each \texttt{eid} in \texttt{eids}.

\omitthis{
red {tCS2C(0,0,mkPutN("putN","dev1","door","lock"), mtR,5,0,1,1,nilAM,rb("door","unlocked"), 2,0, drop)} .
}

As an example, consisder the initial system
\begin{small}
\begin{verbatim} 
initSys =
 { net(mtDM, mtDM)
  ["dev0" | w4Ack(mtDM) w4Rsp(mtM) ...
            rsrcs(mtR) ctr(0)
            sendReqs(...)  config(...) sndCtr(0)]
  ["dev1" | w4Ack(mtDM) w4Rsp(mtM) ...
            rsrcs(rb("door", "unlocked")) ctr(0) 
            sendReqs(nilAM) config(...) sndCtr(1)]
  ["eve" | kb(mtDM) caps(drop)]} 
\end{verbatim}  
\end{small}
\noindent
Its dialect transform, \texttt{D(initSys)} is

\begin{small}
\begin{verbatim} 
{ net(mtDM, mtDM)
  ["dev0" | conf(net(mtDM, mtDM)
          ["dev0" | w4Ack(mtDM) w4Rsp(mtM) ... 
                    rsrcs(mtR) ctr(0)
                    sendReqs(...) config(...) sndCtr(0)])  
    used(mtU) randSize(128)
        seedTo("dev1", "xxxxdev1dev0")
        seedFr("dev1", "xxxxdev0dev1") ixCtr("dev1", 0)]
  ["dev1" | conf(net(mtDM, mtDM)
            ["dev1" | w4Ack(mtDM) w4Rsp(mtM) ...
                     rsrcs(rb("door", "unlocked"))
                     ctr(0) sendReqs(nilAM) config(...) sndCtr(1)])
    used(mtU) randSize(128) 
        seedTo("dev0", "xxxxdev0dev1") 
        seedFr("dev0","xxxxdev1dev0") ixCtr("dev0", 0)]
  ["eve" | kb(mtDM) caps(drop)]} 
\end{verbatim}  
\end{small}

The function \texttt{UD} extracts CoAP endpoints from
\texttt{conf} attributes and 
passes other configuration elements unchanged

\begin{small}
\begin{verbatim} 
op UD : Sys -> Sys .
op UDX : Conf Conf -> Conf .
eq UD({conf}) = {UDX(conf,mt)} .
eq UDX([eid | conf([eid | devatts] conf1) ddevatts] conf0,
               uconf) =
      UDX(conf0, uconf [eid | devatts]) .
eq UDX(conf0, uconf) = uconf conf0 [owise] .

**** produce configurations by omitting the final system construction
op UDC : Conf -> Conf .
**** eq UDC(conf net(dmsgs0,dmsgs1)) = UDX(conf,mt) .
eq UDC(conf) = UDX(conf,mt) .
\end{verbatim}  
\end{small}

\begin{definition}[Initial dialected CoAP configuration]\label{defn:dsysI}
An initial dialected CoAP system, \texttt{dsysI}, is the
result of transforming an initial CoAP system 
\texttt{sysI}
$$\mathtt{dsysI} = \mathtt{D}(\mathtt{sysI}).$$
\end{definition}

\subsection{Dialected  CoAP Scenarios} \label{subsec:dscenarios}


As we prove in Appendix \ref{apx:sbsim-proofs} dialecting
mitigates attacks of a reactive attacker. To illustrate this
claim we can lift the scenarios for reactive attacks of
Section ref{sec:reactive-attacks} to dialected scenarios and
lift the search for attacks to the dialected situation. To
do this we apply the transform \texttt{D} to the initial
system term and apply the inverse (on configurations)
\texttt{UDC} to uses of the search target pattern
\texttt{c:Conf} in the search condition. For example in the
R1 attack (the attacker locks the door after the client has
unlocked it) example the search

\begin{small}
\begin{verbatim}
    search({raR1(5,0,10,false)}) =>! {c:Conf} such that 
        checkRsrc(c:Conf,"dev1","door","lock") .
\end{verbatim}
\end{small}
\noindent
when lifted to a dialected system becomes

\begin{small}
\begin{verbatim}
    search D({raR1(5,0,10,false)}) =>! {c:Conf} such that 
        checkRsrc(UDC(c:Conf),"dev1","door","lock") .
\end{verbatim}
\end{small}

\noindent
There are 2 solutions in 101 states visited in the
original search while no solution is found among
121 states visited in the search for the dialected
system.

As one more example recall the example of R2 type
attacks, where a client is invoking server proceses
one at a time and the attacker manages to 
start some servers early, thus having more than
one server active concurrently.
The original search is
\begin{small}
\begin{verbatim}
    search iSysX(3,0,caps-1) =>+ sys:Sys such that 
          size(epswrb(sys:Sys,rb("sig","on"))) > 1 .
\end{verbatim}
\end{small}
\noindent
and its dialected form is
\begin{small}
\begin{verbatim}
    search D(iSysX(3,0,caps-1)) =>+ sys:Sys such that 
        size(epswrb(UD(sys:Sys),rb("sig","on"))) > 1 .
\end{verbatim}
\end{small}

\noindent
There are 132 solutions from 767 states visited in
the undialected scenario while in the dialected
case no solution is found in 553 visited.

The remaining cases are similar.


\section{Dialect properties}
\label{sec:dialect-properties}

Here we consider the relation of traces of an initial system 
$\mathtt{sysI} = \{ \mathtt{confI} \}$ running the CoAP messaging protocol specification to traces of the 
corresponding dialected system \texttt{D(sysI)} 
in the presence of the different attack capabilities,
including \texttt{mtC} (none).
Although CoAP is designed to run over UDP (unreliable transport) we also consider the case of reliable transport.
We note that the specified dialect wrapper
 (i) drops messages that fail decoding, and
 (ii) has  a notion of used lingo parameter such that
    messages with previously used lingo parameters are dropped.
We consider instances of the following relation:
  $$\Dia(\sysI(\att(c))) ~\relQ~ \sysI(\att(c))$$
where $\sysI(\att(c))$ denotes an initial system
configuration with endpoints in their inital state,
an empty network, and an attacker with capabilities $c$.   The question
is what is $\relQ$.  One candidate is
stuttering bisimulation ($\sbsim$).  
As a second candidate we say that a dialect $\Dia$ is \emph{strongly attack reducing} for attack messages \texttt{M} iff for
any trace of $\sysI(\att(c))$ if an attacker generated message
\texttt{M} is first delivered at step $j$ then there
is a corresponding trace of 
$\Dia(\sysI(\att(c)))$ with
\texttt{M} is delivered at the corresponding step $j'$
and is dropped.  In some more detail we define
$$\Dia(\sysI(\att(c)))\sbsimu[rcv[M]]~\sysI(\att(c))$$
to mean any trace from the lhs without the an attacker event 
has a corresponding trace from the rhs, and dually. If $\tau$ is a
trace of either side with an event $\rcv[M]$ the
first occurrence of receive of  an attacker
produced message, then the lhs trace drops the offending message and any corresponding rhs trace processes the message and possibly diverges from the lhs.
In the cases considered this relation becomes
$$\Dia(\sysI(\att(c)))\sbsim \sysI(\att(\drop)).$$

\paragraph{Dialect basic correctness.}

In the absence of attacks in a reliable or unreliable
network the dialect wrapper simply encodes messages sent by
an endpoint, and decodes them before passing them to the
receiver endpoint.  Thus a dialected CoAP system is
stuttering bisimilar to the original system.
$$\Dia(\sysI(\att(\mtC))) \sbsim \sysI(\att(\mtC))$$
For proof see Appendix \ref{apx:sbsim-proofs} Theorem \ref{theorem:bismmtC}.

\paragraph{Dialecting over unreliable transport.}

Here we consider the question: running over an
unreliable transport, how do $\Dia(\sysI)$ in the
presence of a given attack capability, $\mathtt{cap}$, and $\sysI(\att(c))$ relate?  
  $$\Dia(\sysI(\att(c)))~\relQ~\sysI(\att(mtC))$$
The following is an informal analysis giving
the intuitions.  For proof see Appendix \ref{apx:sbsim-proofs} Theorem \ref{theorem:sbsimC}.

\begin{itemize}
\item[(i)]  $c = \drop$ or $c == \delay(n)$.
In this case $\relQ$ is stuttering bisimilar since the network can drop or delay messages as well! What the attacker can do is significantly change the probability profile of drops/delays.\footnote{This suggests a stronger relation than bisimilarity might be useful}
$$\Dia(\sysI(\att(c)))\sbsim \sysI(\att(mtC))$$
  
\item[(ii)] $c = \divert$ (edit source, or destination or both).
Assuming each communicating pair in the dialected 
system has a unique shared secret, 
$\Dia(\sysI(\att(c)))$ will 
drop the diverted message, thus making it look like a 
network drop.  Without dialecting the new receiver might handle the diverted message, 
possibly causing problems, for example the temperature example (reading room vs oven temperature).   
$$\Dia(\sysI(\att(\divert)))
 \sbsim \sysI(\att(mtC))$$
Note that if all three endpoints involved share the same
secret (the initial seed) then the dialect layer will
not detect the redirection.

\item[(iii)] $c = \replay(n)$.
There are two cases: (a) the original message wasn't delivered---the replay appears as a delay or resend and won't be dropped by the wrapper and will correspond to
a trace on the rhs where the original message was delivered; 
(b) the original message was delivered---the replay is dropped.\footnote{Note that if the lingo policy is to allow a parameter to be used 2 or more times before rejecting, then the second use of a parameter may be rejected after a replay causing dialecting to be complicit in the attack. Thus we chose a single-use policy.}  
  $$\Dia(\sysI(\att(\replay(n)))) \sbsim \sysI(\att(mtC))$$
        
\item[(iv)]  $c = \edit$ a message component (not implemented yet).
There are two cases: 
\begin{itemize}
\item[(iv.a)] the message component is dialected, assuming the attacker  can't break the dialecting, undialecting will fail  and the message will be dropped.    
$$\Dia(\sysI(\att(\edit(a)))) \sbsim \sysI\sysI(\att(mtC))$$
\item[(iv.b)] the message component is not dialected, then
dialecting doesn't change anything
 $$\Dia(\sysI(\att(\edit(b)))) \sbsim \sysI(\att(\edit(b)))$$
\end{itemize}
\item[(v)]  $c = \create$  (message creation not implemented yet).  Assuming the attacker can't break the dialecting, undialecting will fail and the message will be dropped     
 $$\Dia(\sysI(\att(\edit(\create)))) \sbsim \sysI(\att(mtC))$$
\end{itemize}

\paragraph{Dialecting over reliable transport.}

Now we consider the question: In the presence of given
attack capability, $c$, running over a \emph{reliable}
transport, how do $\Dia(\sysI(\att(c)))$ and
$\sysI(\att(c))$ relate?
 $$\Dia(\sysI(\att(c)))~\relQ~\sysI(\att(c))$$

\begin{itemize}
\item[(i)]  $c = \drop$ or $c == \delay(n)$.
Dialecting doesn't affect drop or delay attacks.
 $$\Dia(\sysI(\att(\drop,\delay(n))))\sbsim \sysI(\att(\drop,\delay(n)))$$

\item[(ii)] $c = \divert$ (edit source, or destination or both of message \texttt{M}.   
The dialect wrapper will drop a diverted message,
as it is coded with the wrong secret. The
base CoAP receiver may process the message.
$$\Dia(\sysI(\att(\divert))) ~\sbsimu[rcv(M)]~ \sysI(\att(\divert))$$
 
\item[(iii)] $c = \replay(n)$ of \texttt{M}. Note
that with reliable transport, the original message will
have been delivered. The dialect wrapper will drop the
message, as its lingo (parameter) is already used.
 $$\Dia(\sysI(\att(\replay(n)))) \sbsim \sysI(\att(\mtC))$$

\begin{itemize}
  \item[(iii.a)]  $n$ is small enough that $M$ is delivered within
the original message's (identifier) lifetime , the CoAP receiver may repeat the original response, but will
not process any message requested actions. 

 $$\Dia(\sysI(\att(\replay(na)))) \sbsim \sysI(\att(\replay(na)))$$

\item[(iii.b)] 
$M$ is delivered after the original message's lifetime expires. In this case the CoAP receiver will process the replay.
 $$\Dia(\sysI(\att(\replay(nb)))) \sbsimu[\rcv(M)]~ \sysI(\att(\replay(nb)))$$
\end{itemize}
  
\item[(iv)] $c = \edit$ a component if message \texttt{M} (not implemented yet).  event = \texttt{rcv(M)}, 
There are two cases: 
\begin{itemize}
\item[(iv.a)] the edited message component is dialected.  Assuming the attacker
   can't break the dialecting, undialecting will fail 
   and the message will be dropped,  while the undialected system will
   attempt to process the message.

 $$\Dia(\sysI(\att(\edit(a)))) \sbsimu[rcv[M]]~ \sysI(\att(\edit(a)))$$

\item[(iv.b)] the message component edited is not dialected, then dialecting doesn't change anything.
 $$\Dia(\sysI(\att(\att(\edit(b))))) \sbsim \sysI(\att(\edit(b)))$$
\end{itemize}
\item[(v)]  $c = \create$ message \texttt{M}.
Assuming the attacker can't break the dialecting, the dialect wrapper will drop the message, but $\sysI(\att(\create))$ will attempt to process it.

   $$\Dia(\sysI(\att(\att(\create))))\sbsimu[rcv[M]]~ \sysI(\att(\create))$$
\end{itemize}

\section{Related Work}
\label{sec:related}

The work most closely related to the present work is
\cite{GEMM-2023esorics}, which we will refer to as
ESORICS23 in the following. Complimentary work on
dialects is discussed briefly in Section
\ref{sec:background}. We refer to ESORICS23 for a
more complete review of previous work on dialects.

ESORICS23 addresses three key aspects of dialects:
synchronization mechanisms; protocol and lingo
genericity; and attack model vs dialect choice. It
also defines taxonomies for key dimensions of a
dialect.

A general framework is proposed defining notions of
dialect and lingo as (parameterized) transformations
on protocol theories, where a protocol theory is a
generalized actor theory with a sort
\texttt{Configuration} and a constant
\texttt{initConf} of sort \texttt{Configuration},
representing the initial protocol system state. A
lingo defines a pair $f,g$ of parameterically
inverse functions for obfuscating and deobfuscating
protocol messages: $g(f(m,p),p) = m$. A lingo
transforms a protocol theory to a new protocol
theory by adding the function definitions. The
dialect transform operates on the protocol theory
together with the lingo transformed theories. In the
resulting theory, protocol objects are wrapped with
dialect meta objects with the same identifier and
rules are provided for sending/receiving messages at
the dialect level.

The CoAP Dialect presented in Section
\ref{sec:coap-dialect-spec} has one lingo, defined
in the module \texttt{COAP-DIALECT-LINGOS} which
extends the functional part of the CoAP
specification. The CoAP dialect is specified as the
result of the transform, by including the lingo and
protocol modules and adding rules for receiving and
sending messages. The initial configurations are
defined in protocol and dialect test modules, where
dialected configurations are obtained as in the
ESORICS23 case by wrapping protocol objects by meta
objects.

We note that the definition of protocol theory with
a fixed initial configuration doesn't quite work if
you want it to support a variety of scenarios. For
example MQTT and CoAP support networks of almost
arbitrary size and topology. One can define a
protocol theory for each scenario of interest, with
a dialect being generic with respect to the initial
configuration. But this looses the information that
all the scenarios use the same rules and data types
thus seems not quite satisfactory. Perhaps a notion
of scenario constructor parameterized by a network
topology and a means of specifying initial shared
secrets could be useful.  Also, configurations
are `open' in the sense that each actor or subset
of actors and messages is a configuration and
a configuration is a subconfiguration of infinitely
many larger configurations.  In most cases rewrite
rules apply to a subconfiguration, for example an
actor and a message, independent of what configuration it is part of.  However, some rules
need to know that they are seeing the full configuration.  This can be achieved by introducing
the notion of system as an encapsulated configuration.   The question is whether the notion
of system should be a first class concept for
protocols and dialect transformations.

\paragraph{Synchronization mechanisms.}

Synchronization can happen on different time scales:
when the set of lingos in use changes, and when
current lingo parameters change. The latter is
likely to be more frequent than the former, although
these are not necessarily distinct, since a lingo
parameter could be use to switch between lingo
functions say xor versus shuffle.

Two classes of synchronization mechanism are
identified in ESORICS23: \emph{periodic}--that uses
synchronized clocks and fixed time intervals for
changing to the next phase and its lingo set) or next
parameter values; and \emph{aperiodic}--using
information in the local history to determine
when/what to change.

In the CoAP dialect studied here, there is one
(abstract) lingo with a changing parameter that
selects the next pseudo random string in a
practically unpredictable (i.e. hard to predict)
sequence. The CoAP dialect uses aperiodic
synchronization for reasons discussed in Section
\ref{sec:dialect-fns}. The relevant history is a set
of counters one for each communication partner that
is incremented upon each use.

\paragraph{Genericity.} In the framework proposed in
ESORICS23, dialects are parametric in a list of
lingos and thus can be applied to any protocol
theory for which the lingos are applicable. Assuming
that message content is of type string enables a
wide choice of lingos independent of the meaning of
the strings. The CoAP lingo functions assume an
abstract message content type Content and thus the
CoAP dialect could be applied to any protocol
specification with this content type (and suitable
configuration sort). We note that the example used
in ESORICS23 suggests that the behavior of the
dialect depends on the message type and even uses
information from the CONNECT messages, while the
CoAP dialect does not look inside message content.
Presumably some dependence on message structure
could be made a parameter to the dialect in addition
to the lingo parameters.

\paragraph{Attack model.}

The ESORICS23 attack model starts with two models
proposed in \cite{ren-etal-2023seccom}: the
\emph{off-path }attacker that can send and receive
messages, but can not see messages exchanged between
its targets; and the \emph{on-path} attacker with
all the capabilities of the off-path attacker, plus
the ability to observe communications between honest
endpoints. ESORICS23 adds a third model:
\emph{active on-path} attacker with the capabilities
of the on-path attacker plus the ability to block,
redirect, and alter message content. A refinement is
also introduced--\emph{(m, n, t) attacker} that in
time t can observe \emph{m} messages and block
\emph{n} messages (and send/receive, block, ...).
The goal of a dialect is low probability that an
``attacked'' message will be accepted by intended
receiver.

In this paper we define two attack models for CoAP
messaging systems: \emph{active} and
\emph{reactive}. In either case, an attacker is
allowed a bounded set of actions from a set of
capabilities (depending on the model). The bound on
number of actions is an alternative to the bounds in
the \emph{(m,n,t)} attack model as a means to limit
the attacker strength. The active attacker is based
on vulnerabilities identified in
\cite{coap-attacks}. The capabilities are:
block/drop, delay, replay, or redirect messages. Our
\emph{active} attack model differs from the
\emph{active on-path} model in that attackers don't
receive or create messages. Like \cite{coap-attacks}
we assume communications may be encrypted, so
focuses on attacks that don't involve parsing
messages -- all the attacker learns from a message
is the target, source and a content blob.

The ESORICS23 paper is focused on understanding the
structure of dialects and attack models and gives
some examples and informal statements about the
security guarantees of lingos and dialects. In the
present paper we formally specify an attack model
and carry out reachability analyses concerning
specific undesired states of simple applications
using the CoAP protocol (inspired by
\cite{coap-attacks}). We also prove properties
relating executions of dialected and undialected
scenarios with and without attacks.

\paragraph{Taxonomies.}

ESORICS23 proposes a taxonomy of dialects with four
classes: Protocol Specific; Protocol Generic;
Protocol Modifying; Protocol Wrapping. The CoAP
dialect is essentially generic for single lingo case
(the lingo needs to be suitable for the protocol
message structure) and it is Wrapping.

A lingo taxonomy is proposed with classes
Shuffle/permutation (assumes messages are bit
strings); Split; Flow (via customized messages);
Crypto (HMAC, XOR, \dots); and Use of unused/random
field. The CoAP lingo is abstract, the functions
could be realized using (combinations of) shuffle,
xor, or crypto. We did not consider Splitting as one
of the aims of CoAP is to avoid packet
fragmentation. With an unreliable transport, the
dialect/lingo itself would have to do some extra
work to avoid additional dropped messages.

\paragraph{Other points.}

It appears that in the ESORICS23 framework all
honest actors (protocol actors) share the same
secrets. Examples only consider two actor systems,
so its not completely clear. In section
\ref{sec:dialect-properties} we show examples of
systems with 3 or more endpoints where its important
to have distinct pairwise shared secrets.

ESORICS23 also proposed notions of (generic) vertical
and horizontal composition. That is beyond the scope of
the present work, but certainly an interesting future
direction along with partial evaluation transformations
to eliminate extra overhead where resources are very
limited.

Complimentary to ESORICS23, the current paper
carried out the dialect theory transformation by
hand, formalizing only the wrapping (and unwrapping)
of system configurations, and focused effort on
formalizing attack models and carrying out formal
reachability analyses for a variety of scenarios.

\section{Conclusions}
\label{sec:concl}

Dialects such as ESORIC23 or the CoAP dialect can
protect against replay, editing message source and/or
target, and message forging. They can not protect
against \emph{interference} attacks such as dropping
or delaying. These dialects also don't protect
against \emph{piggy back attack}. In this attack
there are two attackers X,Y where X wants to send
signals to Y. When A sends M to B, X intercepts and
sends to Y (on port j). Y learns the signal
represented by j and forwards M to B.

This suggested the reactive attacker model for
dialects--a non-interference model. The attacker can
observe (copy), construct/edit, transmit (no
block/drop, delay) thus can replay, redirect, spoof.
We showed that dialects satisfying minimal conditions
provide security against non-interference attacks in
the sense that a dialected protocol in the presence
of such attacks is bisimilar to the protocol alone.

We conclude with so ideas for future work. It could
be interesting to consider dialects that emit
apparently random messages between themselves? These
could be versions of previously sent messages or made
up messages of suitable size. They should be similar
to ordinary traffic and not sent too often. Can this
obfuscate patterns the attacker relies on? Another
alternative is deploying intermediate nodes with a
dialect layer only. These would be used to obfuscate
where the message is actually going.   There is still 
work to be done to capture a generic attack model
for messaging protocols: what are the properties 
that applications care about?  how to represent them
generally?  what attack capabilities are needed to
succeed in violating critical properties?  is there
a notion of complexity of property that corresponds
to a minimal bound on successful attacker capability?

\paragraph{Acknowledgments.} 
Talcott thanks Dr. Catherine Meadows for many
insightful discussions that lead to some of the
attack model ideas and other results.
This material is based upon work supported by the Naval Research Laboratory under Contract No. N00173-23-C-2001. Any opinions, findings and conclusions or recommendations expressed in this material are those of the author(s) and do not necessarily reflect the views of the Naval Research Laboratory.


 \newpage
 \bibliographystyle{abbrv}
 \bibliography{bib}

\appendix
\section{Characterization of reachable CoAP system configurations}
\label{apx:exe-props}

As for real-time Maude specifications
\cite{olveczky08tacas} the rewrite rules for the
CoAP specification (and it dialected form) are
either instantaneous rules or the \texttt{tick} rule
that models passing of time. The minimum time to
elapse function, \texttt{mte(sys)}, determines when
time can pass. The following propositions give
essential properties of reachable CoAP and dialected
CoAP system configurations according to the value
\texttt{mte(sys)}.

\begin{proposition}[Mte based characterization of reachable CoAP configurations]\label{prop:mte-coap}
Let \texttt{sysI}  be an initial configuration 
and let  \texttt{sysC} be a CoAP system configuration
reachable from \texttt{sysI}. 
Then  \texttt{mte(sysC)}  is either \texttt{0},
\texttt{nz} a non-zero natural number, or the constant \texttt{infty} (representing infinity)
and one of the following holds. 

\begin{itemize}
\item[1.] \texttt{mte(sysC) = nz:NzNat}.
In this case the tick rule is enabled (and no other
rule is enabled (except possibly \texttt{net} which
commutes with and is independent of \texttt{tick}).
Further more
\begin{itemize}
\item  
 for all \texttt{msg @ d} in the network,  $\mathtt{d} > 0$ (ow \texttt{mte} is 0). 
\item  
for all \texttt{msg @ d} in some \texttt{w4Ack} attribute $\mathtt{d} > 0$  (ow \texttt{mte} is 0). 
\item  
if the \texttt{sendReq} attribute of some endpoint is not empty (\texttt{nilAM}) then either there are more than the \texttt{w4Ack} limit of requests
waiting for acknowledgement
(and hence rule \texttt{devsnd} is not
enabled), or the corresponding \texttt{sndCtr}
attribute is non-zero.
\end{itemize}
\item[2.] \texttt{mte(sysC) = 0} 
In this case some rule is enabled according to
one of the following cases.
\begin{itemize}
\item there is a message with zero delay in the net,
then the rule \texttt{rcv} is enabled (possibly preceeded by \texttt{net}) to move the message from
network input to network output;

\item there is a \texttt{w4Ack} attribute of some
endpoint with a zero delay message, then
the \texttt{ackTimeout} rule is enabled;

\item there is a \texttt{sendReqs} attribute of some
endpoint that is not empty and sending is enabled,
so the rule \texttt{devsnd} is enabled.
\end{itemize}
\item[3.] \texttt{mte(sysC) = infty}. 
In this case \texttt{sysC} is a terminal configuration, no
rules apply since the net is empty, all \texttt{w4Ack} attributes are empty, and all sendReqs are empty buy definition
of \texttt{mte}.
\end{itemize}
\end{proposition}
   
The definition of \texttt{mte} is extended to 
dialected CoAP configurations by adding 
the rule that \texttt{mte} is 0 if
any local nets are non-empty. Thus messages in
local nets must be delivered or dialected and put
in the global net before time can pass.  The
consequences are captured in the next proposition.
 
\begin{proposition}[Mte based characterization of reachable dialected CoAP configurations.]
\label{prop:mte-dialected}
As for the non-dialected case there are three cases.
\begin{itemize}
\item[1.] $\mathtt{mte(dsysC)} = \mathtt{nz:NzNat}$.
In this case the value of \texttt{mte} is determined
by the global network and the \texttt{sysC}
level attributes and only the \texttt{tick} rule (and  possibly \texttt{net}) is enabled.

\item[2.] $\mathtt{mte(dsysC)} = \mathtt{0}$.
In this case, if there is a message
some local net the either \texttt{ddevsnd} or \texttt{rcv} is enabled.  Otherwise 
either \texttt{ddevrcv} is enabled or some rule
at the \texttt{sysC}  level as shown in 
Proposition \ref{prop:mte-coap}.

\item[3.] $\mathtt{mte(dsysC)} = \mathtt{infty}$.
In this case \texttt{dsysC} is terminal.
 As for \texttt{sysC} the net is empty (global and local), \texttt{w4Ack} attributes are empty, and 
 \texttt{sendReq} attribues are empty so no rule is enabled.
\end{itemize}
\end{proposition}

  \omitthis{
  if M is sent it will be received if not dropped
  if sv hasRspTSnt to cl tok matching aid then
  previously sv rcvd req src cl tok matching aid
  If cl hasPend sv aid then cl previously sent
  req to sv w tok matching aid
  If cl hasRspTRcd from sv mathing aid then cl previously sent req to sv w tok matching aid
  }

 \begin{small}
 \begin{verbatim}
 \end{verbatim}
 \end{small}
\omitthis{
\begin{proposition}[Message based characterization of reachable states]\label{prop:msg-based}

(M=m(ep0,ep1,content) in net then
 if isReq(M) then ep1 has corresponding w4 (ack or rsp) or rspRcd
 if isResp(M) then ep1 has corresponding rspSnt
      and ep0 had corresponding w4
 if isAck(M) then ep0 has matching w4Ack or a separate response
 arrived earlier and ep0 has matching rspRcd
 For every w4Ack there is the request, or a matching ack or rsp 
 in the net,  [if unreliable net modeled then the req or ack or rsp
 could have been dropped]
 if epid has rspRcd from epid0 then epid0 has corresponding rspSnt
 \end{proposition}
 
 }     

\section{CoAP vulnerabilities}
\label{apx:coap-vulnerabilities}

The paper \cite{coap-attacks} motivates an extension to
the CoAP standard by describing several vulnerabilities
where attacks do not depend on being able to see message
content, i.e. these attacks can in principle be carried
out against CoAP running over DTLS or other transport
providing encryption.

Four classes of vulnerability are discussed using message
diagrams to illustrate attack scenarios. Three of the
classes rely only on the ability to drop or delay
messages. In some cases the vulnerability is due to reuse
of tokens resulting in requests connected to multiple
responses. These scenarios do not exhibit vulnerabilities
in the Maude specification, but in some cases we give
alternative scenarios that realize a similar
vulnerability. With the execption of the example that
relys on mismatching responses to GETs of different
resources on one server, dialecting does not change the
success of attacks as it can not detect blocking/dropping
or delay. Our modified version of the mismatch case uses
two servers and in this case dialecting prevents the
attack becuase lingos are pairwise disjoint. The fourth
class concerns message fragmentation in the network. This
is not modeled in the current Maude CoAP specification.

In the following we show how the three relevant
vulnerabilities can be formally modeled and analyzed
using the Maude specification. In the diagrams, the
vertical lines are participant send/receive event
timelines. The client is labelled \textsc{dev0}, the
server labelled \texttt{dev1} and the attacker labelled
\texttt{eve}. The state of the server resource
\texttt{"door"} is indicated at the beginning and
end of the server timeline. In a timeline,
\texttt{o}s mark normal send or receive events,
\texttt{X} indicates a message drop, and \texttt{@} 
indicates a message capture (for later release).

We use the function \texttt{tCS2C} defined
in Section \ref{subsec:coap-scenarios} to construct
scenario initial system configurations.  For convenience
we repeat the definition here.
\begin{small}  
\begin{verbatim}  
op tCS2C : Nat Nat AMsgL RMap Nat Nat
           Nat Nat AMsgL RMap Nat Nat
          Caps -> Conf .
eq tCS2C(n0,j0,amsgl0,rbnds0,mqd0:Nat,w4ab0:Nat,
         n1,j1,amsgl1,rbnds1,mqd1:Nat,w4ab1:Nat,
        caps) = 
   net(mtDM,mtDM)
   mkDevC(n0,j0,amsgl0,rbnds0,mqd0:Nat,w4ab0:Nat)  
   mkDevC(n1,j1,amsgl1,rbnds1,mqd1:Nat,w4ab1:Nat)  
   (if caps == mtC then mt else mkAtt(caps) fi)  .
 \end{verbatim}  
 \end{small}  
\noindent
It generates a configuration with two endpoint agents (constructed using \texttt{mkDevC}) and an attacker (if \texttt{caps} is non empty). 
The arguments specify initial application messages
(\texttt{amsglj}) and resource maps (\texttt{rbndsj}).  The first two numbers for a
device construction determine the device id and inital delay for message sending.  The last two numbers for a device construction specify the
congestion control behavior: a new message can not be sent if there are more than \texttt{w4ab}  requests not yet acknowledged, and there must be
a delay of a least \texttt{mqd} between requests.

We also use three application message
constructors to generate requests.\footnote{Recall that a client will retry a confirmable request if no acknowlegment
or response is received within a specified time. Non-confirmable requests are not acknowledged and are not retried.}
\begin{itemize}
\item
\texttt{mkPutN(tag,server,resource,val)}:
a non-confirmable request to \texttt{server} to set
\texttt{resource} to \texttt{val}.
\item
\texttt{mkPutC(tag,server,resource,val)}:
a confirmable request to \texttt{server} to set
\texttt{resource} to \texttt{val}.
\item
\texttt{mkGetN(tag,server,resource)}
a non-confirmable request to get the \texttt{resource} value from \texttt{server}.
\end{itemize}    
\noindent
See Section \ref{subsec:coap-scenarios} for definitions.

\noindent
To check if an attack is possible in a given
scenario, we search from the constructed initial
configuration \texttt{initSys} to a \textbf{terminal}
configuration \footnote{See section
\ref{apx:exe-props} for a characterization of
terminal configurations.} that satisfies a property
\texttt{P(c:Conf, params)} characterizing the
effects of the attack typically in terms of
properties of unperturbed executions that are 
violated.  Requiring solutions to be terminal allows us to rely on a number of consequences such as
if a message is sent, then it will be received
unless it is dropped. 

\begin{small}  
\begin{verbatim}  
search initSys =>! {c:Conf ["eve" | a:Attrs caps(mtC)]}
    such that P(c:Conf, params)
\end{verbatim}  
\end{small}  
\noindent
This lists all the ways the property P can be satisfied.
In some cases, we also search for executions in which the
attacker fails, i.e. the expected execution properties
hold. Including the term \texttt{["eve" | a:Attrs
caps(mtC)]} in the search pattern ensures that the
attacker will use all its capabilities in reaching a
solution state.

We use the following predicates to form the search
properties. These are very basic properties that
characterize the state of endpoints at the end of a
scenario.
\paragraph{Properties of the server state.}
\begin{itemize}
\item
\texttt{checkRsrc(c:Conf,sv,r,v)}: The
resource \texttt{r} of server \texttt{sv} has value \texttt{v}.
\item
\texttt{hasRspTSnt(c:Conf,sv,cl,aid)}:
Server \texttt{sv} has sent a response to client
\texttt{cl} in response to a request with
token matching \texttt{aid} (the application provided
message identifier, which is part of the message id
and token strings).

\item
\texttt{rspTSntBefore(c:Conf,sv,cl,aid0,aid1)}:
Server \texttt{sv} sent responses to client \texttt{cl}  in response to a request with
token matching \texttt{aid0} and  a requesst with token matching \texttt{aid1} with the \texttt{aid0}
matching response being sent first.

\end{itemize}    
 
\paragraph{Properties of the client state.}
 \begin{itemize}
\item
\texttt{hasRspTRcd(c:Conf,cl,sv,aid)}:
Client \texttt{cl} had received a response from server \texttt{sv} to a request with token matching \texttt{aid}.   (The empty string matches any token.)
\item
\texttt{hasGetRsp(c:Conf,cl,sv,aid,val)}:
Client \texttt{cl} had received a response from server \texttt{sv} to a GET request with token matching \texttt{aid}. The response body contains
value \texttt{val}.  (If \texttt{val} is the empty
string, the value check is omitted.)
\item
\texttt{rspPend(c:Conf,cl,sv,aid)}:
Client \texttt{cl} is waiting for a response
from server \texttt{sv} to a request with token matching \texttt{aid}.
\end{itemize}

The following subsections treat the different
vulnerabilities. They begin with a diagram of the
proposed attack, followed by a formal definition of
the corresponding scenario, search commands and
results demonstrating attack success and failure
situations. We show the scenario definitions and
search commands for the original CoAP case. The
search command for the corresponding dialected
scenario is given by applying the dialect
transformation \texttt{D} (defined in Section
\ref{subsec:dialect-transform}) to the initial system
configuration and lifing the search pattern to the
dialected structure to expose the relevant elements.
Thus the dialected version of the CoAP search command

\begin{small}  
\begin{verbatim}  
        search initSys =>! 
          {c:Conf ["eve" | a:Attrs caps(mtC)]
                  [epid | devatts:Attrs]}
          such that prop(c:Conf,epid) .
\end{verbatim}  
\end{small}  
\noindent 
is
\begin{small}  
\begin{verbatim}  
       search D(initSys) =>! 
         {c0:Conf  ["eve" | a:Attrs caps(mtC)]
                   [epid | conf(c:Conf)]
          ddattrs0:Attrs]}
        such that prop(c:Conf,epid) .
\end{verbatim}  
\end{small}  
\noindent
Note that the goal property doesn't change.

\subsection{The Selective Blocking Attack.}
This scenario illustrates the ability of an
attacker to block a single message.  
How the attacker recognizes the intended target is 
not discussed in \cite{coap-attacks}.
Th scenario is adapted from Figures 1,2 of 
\cite{coap-attacks}. 

\begin{small}  
\begin{verbatim}  
      dev0      eve    dev1        dev       eve     dev1 
      ---      ---   unlock        ---       ---    unlock
       o -PUTN->X      |            o -PUTN->  ------>  o 
                       |            |         X <-2.04- o 
                    unlock          |                 lock

          (1) drop message            (2) drop response

          Figures 1,2 scenarios          
\end{verbatim}  
\end{small} 
\noindent
\texttt{-PUTN->} a non-confirmable request to lock
the door. \texttt{-2.04->} the \texttt{PUTN}
response. (1) Attacker blocks the request. (2) Attacker blocks the response.

The function \texttt{caFig1.2} specifies the scenario for the message
dropping attacks of Figures 1 and 2, parameterized by
congestion control parameters (\texttt{mqd:Nat} the minimum time
between message sends; more that \texttt{w4b:Nat} requests in
\texttt{w4Ack} blocks sending additional requests).

\begin{small}  
\begin{verbatim}  
    op caFig1.2 : Nat Nat -> Conf .
    eq caFig1.2(mqd:Nat,w4b:Nat) =  
         tCS2C(0,0,mkPutN("putN","dev1","door","lock"), 
                   mtR,mqd:Nat,w4b:Nat,
               1,1,nilAM,rb("door","unlocked"),2,0,
               drop) .
\end{verbatim}  
\end{small}  

\paragraph{Figure 1 successful attack.}
Dropping a request is specified by checking that the server 
did not make any response (hence it did not
receive the request) and the client
awaits a response (it sent a request). 
Note this is redundant given the semantics, but a systematic characterization.

\begin{small}  
\begin{verbatim}  
    search {caFig1.2(5,0)} =>! 
      {c:Conf ["eve" | a:Attrs caps(mtC)]} such that 
      not(hasRspTSnt(c:Conf,"dev1","dev0","putN")) and
      rspPend(c:Conf,"dev0","dev1","putN")  .
\end{verbatim}  
\end{small}  

\noindent
or by checking that the server resource has the
original value.
\begin{small}  
\begin{verbatim}  
      search {caFig1.2(5,0)} =>! 
      {c:Conf ["eve" | a:Attrs caps(mtC)]} such that  
      checkRsrc(c:Conf,"dev1","door","unlocked") .
\end{verbatim}  
\end{small}  
\noindent
In either case there are 2 solutions among 7 states. The
results for the corresponding two dialected cases are
similar, 2 solutions among 9 states visited.

\paragraph{Figure 2 successful attack.}  Dropping the response
can be specified by checking that the server 
sent a response and the client did not receive one.

\begin{small}  
\begin{verbatim}  
    search {caFig1.2(5,0)} =>! 
    {c:Conf  ["eve" | a:Attrs caps(mtC)]} such that   
      hasRspTSnt(c:Conf,"dev1","dev0","putN") and 
      checkRsrc(c:Conf,"dev1","door","lock") and 
      rspPend(c:Conf,"dev0","dev1","putN") .
\end{verbatim}  
\end{small}  
\noindent
There are 2 solutions out of 13 states visited.

\paragraph{Figure 1,2  attack failures.} 
To check that 
there are good executions in the presence of
an attacker we check that the server sends
a response and the client receives it and the
attacker used all its capabilities.  

\begin{small}  
\begin{verbatim}  
    search {caFig1.2(5,0)} =>! 
      {c:Conf ["eve" | a:Attrs caps(mtC)]} such that 
         hasRspTSnt(c:Conf,"dev1","dev0","putN") and 
         hasRspTRcd(c:Conf,"dev0","dev1","putN") .
\end{verbatim}  
\end{small}  
\noindent
There is no solution, some message has to be dropped
and there are only two messages.

\subsection{The Request Delay Attack.}
This scenario is based on figure 3 of \cite{coap-attacks}.
\begin{small}  
\begin{verbatim}  
          dev0        eve        dev1
          ----        ---        lock
           o -PUTNDU-> @    
           o -PUTNSO->  ------>    o  
           o <-------  <-2.01--    o
                       o -PUTNDU-> o
                       X <-2.04-   o
                                unlock
           Figure 3 scenario                              
\end{verbatim}  
\end{small}  
The client intends to unlock
the door, \texttt{PUTNDU}, then turn a signal on
\texttt{PUTSO}.  However the attacker delays the door unlock request and drops the \texttt{PUTNDU} response.
The signal goes on with the door locked.

The Figure 3 attack scenario is constructed by
the function  \texttt{caFig3}.
\begin{small}  
\begin{verbatim}  
      op caFig3 : Nat Nat Nat -> Conf .
      eq caFig3(d:Nat,mqd:Nat,w4b:Nat) = 
       tCS2C(0,0,mkPutN("putND","dev1","door","unlock") ;
                 mkPutN("putNS","dev1","signal","on"),
                mtR,mqd:Nat,w4b:Nat,
            1,1,nilAM,rb("door","lock"),2,0,
            drop delay(d:Nat)) .
\end{verbatim}  
\end{small}  

\paragraph{Figure 3 successful attacks.}
From the diagram we characterize the attack
by server properties saying a request with token
matching \texttt{putNS} was received and processed
before the request with token matching  \texttt{putND}; the door is unlocked, and
the client has not received a response to the
message with token matching \texttt{putND}.

\begin{small}  
\begin{verbatim}  
      search {caFig3(15,5,0)} =>! 
      {c:Conf ["eve" | a:Attrs caps(mtC)]} such that 
      checkRsrc(c:Conf,"dev1","door","unlock") and 
      rspTSntBefore(c:Conf,"dev1","dev0","putNS","putND") and 
      rspPend(c:Conf,"dev0","dev1","putND") .
\end{verbatim}  
\end{small}  
\noindent
There are 4 solutions in 330 states visited.
In the dialected version of the scenario, there
are 4 solutions among 496 states visited.

\paragraph{Figure 3 failed attacks.}
Failed attacks are those executions where
the attacker uses all its capabiities, but
the endpoint events are as expected in
absence of attacker.  That is, the server
has responded to \texttt{putND} before \texttt{putNS} and the client has received
the two expected responses.  If we require
the client receive all responses there
will be no solution, as something needs
to be dropped.  So we only require the client
receive a response to \texttt{putND}.

\begin{small}  
\begin{verbatim}  
      search {caFig3(15,5,0)} =>! 
      {c:Conf  ["eve" | a:Attrs caps(mtC)]} such that  
      checkRsrc(c:Conf,"dev1","door","unlock") and 
      rspTSntBefore(c:Conf,"dev1","dev0","putND","putNS") and
      hasRspTRcd(c:Conf,"dev0","dev1","putND")  .
\end{verbatim}  
\end{small}  
\noindent
There are 4 solutions, for example the attacker
might delay the \texttt{putNS} request, and drop
the response to that request.

\subsection{The Delay with Reordering Attack.}
This scenario is based on Figure 4 of \cite{coap-attacks}.
Here the door is initially locked.  The client
requests unlock, possibly carries out some other
interactions, and then requests lock.  The attacker
captures the unlock request. Since the unlock request is confirmable, the client resends it and the server
sends a response to the resent request.  The following
lock request is processed as expected.  Then
the attacker releases the captured unlock request
and drops the response. Thus the client believes the
door is locked when in fact it is unlocked.
\begin{small}  
\begin{verbatim}  
        dev0        eve        dev1
        ----        ---        lock
         o -PUTCU->  @    
         o -PUTCU->  -------->  o  --- retry
         o <----     <---2.04-- o
                  .. ... ..
         o -PUTNL->    ------>  o  
         o <-------  <--2.04--  o
                     o -PUTCU-> o  --- release
                     X <-2.04-  o
                              unlock  
           Figure 4 scenario                              
\end{verbatim}  
\end{small}  

This scenario is constructed by the function
\texttt{caFig4x} giving the attack a drop and a delay
capability.  The function is parameterized by the delay ammount
(\texttt{n:Nat}) and the usual congestion controls.
\begin{small}  
\begin{verbatim}  
      op caFig4x : Nat Nat Nat -> Conf .
      eq caFig4x(n:Nat,mqd:Nat,w4b:Nat) = 
         tCS2C(0,0,mkPutC("putC","dev1","door","unlock") ;
                  mkPutN("putN","dev1","door", "lock"),
                  mtR,mqd:Nat,w4b:Nat,
             1,1,nilAM,rb("door","lock"),2,0,
             drop delay(n:Nat)) .
\end{verbatim}  
\end{small}  

\paragraph{Figure 4 Successful Attack}
The attack is specified by requiring
the server to send a response to a request
with token matching \texttt{"putC"} followed
by a \texttt{"putN"} followed by another \texttt{"putC"} with the door unlocked at the end.
We require the client to have received responses
to \texttt{"putC"} and \texttt{"putN"} requests
thus leading the client to believe it succeeded.
\begin{small}  
\begin{verbatim}  
      search {caFig4x(10,5,0)} =>! 
      {c:Conf ["eve" | a:Attrs caps(mtC)]} such that 
      checkRsrc(c:Conf,"dev1","door","unlock") and 
      hasRspTSnt(c:Conf,"dev1","dev0","putC") and 
      hasRspTSnt(c:Conf,"dev1","dev0","putN") and 
      rspTSntBefore(c:Conf,"dev1","dev0","putC","putN") and 
      rspTSntBefore(c:Conf,"dev1","dev0","putN","putC") and 
      hasRspTRcd(c:Conf,"dev0","dev1","putN") and 
      hasRspTRcd(c:Conf,"dev0","dev1","putC") .
\end{verbatim}  
\end{small}  
\noindent
There are 8 solutions from 2600 states visited.
In the dialected version there are 16 solutions
from 10167 states visited.

\paragraph{Figure 4 Failed Attack.}
The attacker can fail even though both capabilities
are used.  In this case we require the server
response sequence to be simply \texttt{putC} then \texttt{putN} and the final door state to be
\texttt{"lock"} and requier the client state as
for the successful attack.
\begin{small}  
\begin{verbatim}  
          search {caFig4x(10,5,0)} =>! 
          {c:Conf ["eve" | a:Attrs caps(mtC)]} such that 
          checkRsrc(c:Conf,"dev1","door","lock") and 
          hasRspTSnt(c:Conf,"dev1","dev0","putC") and 
          hasRspTSnt(c:Conf,"dev1","dev0","putN") and 
          rspTSntBefore(c:Conf,"dev1","dev0","putC","putN") and  
          hasRspTRcd(c:Conf,"dev0","dev1","putN") and 
          hasRspTRcd(c:Conf,"dev0","dev1","putC") .
\end{verbatim}  
\end{small}  
\noindent
There are 45 solutions from 2586 states visited.

\subsection{The Response Delay and Mismatch Attack.}
This scenario is based on Figure 5 of \cite{coap-attacks}.
In this attack the client sends an unlock request,
the attacker captures the response, since the resquest is non-confirmable the client doesn't
worry about the lack of response.  Later the
client sends a lock request that is dropped by
the attacker, who releases the unlock response
to trick the client into thinking the lock
request succeeded, when in fact the door remains
unlocked.
       
This works for clients that reuse tokens.
It fails using the Maude client specification.

\begin{small}  
\begin{verbatim}  
            dev0      eve       dev1
            ----      ---       lock
             o -PUTNU-->   -----> o
             |         @  <-2.04- o                  
                    ......
             o -PUTNL->   X       |
             o <-2.04-    o       |
                               unlock
           Figure 5 scenario  
\end{verbatim}  
\end{small}  

This attack scenario is specified by the function
\texttt{caFig5x}.
\begin{small}  
\begin{verbatim}  
        op caFig5x : Nat Nat Nat -> Conf .
        eq caFig5x(d:Nat,mqd:Nat,w4b:Nat) = 
           tCS2C(0,0,mkPutN("putNU","dev1","door","unlock") ;
                     mkPutN("putNL","dev1","door", "lock"),
                     mtR,mqd:Nat,w4b:Nat,
                 1,1,nilAM,rb("door","lock"),2,0,
                 drop delay(d:Nat)) .
\end{verbatim}  
\end{small}  

\paragraph{Figure 5 intented attack fails}
\begin{small}  
\begin{verbatim}  
            search {caFig5x(10,5,0)} =>!
            {c:Conf ["eve" | a:Attrs caps(mtC)]} such that 
            hasRspTSnt(c:Conf,"dev1","dev0","putNU") and 
            not(hasRspTSnt(c:Conf,"dev1","dev0","putNL")) and 
            checkRsrc(c:Conf,"dev1","door","unlock") and 
            hasRspTRcd(c:Conf,"dev0","dev1","putNL") .
\end{verbatim}  
\end{small}  
\noindent
There is no solution.  Also no solution in the
dialected case.

\paragraph{Figure 5 alternative attack}
The result that the client believes the door is locked
when it is not can be realized by an attack
where the unlock request is delayed and
the unlock response is dropped.
\begin{small}  
\begin{verbatim}  
          search {caFig5x(10,5,0)} =>! 
          {c:Conf ["eve" | a:Attrs caps(mtC)]} such that
          hasRspTSnt(c:Conf,"dev1","dev0","putNU") and 
          hasRspTSnt(c:Conf,"dev1","dev0","putNL") and 
          rspTSntBefore(c:Conf,"dev1","dev0","putNL","putNU") and 
          checkRsrc(c:Conf,"dev1","door","unlock") and 
          hasRspTRcd(c:Conf,"dev0","dev1","putNL") .
\end{verbatim}  
\end{small}  
\noindent
There are 4 solution 4 among 330 states visited.
In the dialected case there are 4 solutions among
499 states visited.

\subsection{Delaying and mismatching response to \texttt{GET} requests}
This attack is based on Figure 6 of \cite{coap-attacks}.
In this attack the attacker captures the response
to a GET request, drops a subsequent GET request
and releases the original response.
Between GET requests the door state may have changed.
A vulnerable client will accept the delayed response to the first GET as response to the second GET, and thus in the case where the door is locked initially, and becomes unlocked, client incorrectly thinks the door is locked, assuming the unlock request failed to arrive.
The Maude client will accept the delayed response, but as response to the first GET and will not have any knowledge about the later state.

\begin{small}  
\begin{verbatim}  
            dev0      eve        dev1
            ----      ---        lock
             o -GETN->   ------>  o    getN0
             |         @  <-lock- o                  
             o -PUTNU-->  ----->  |    putNU
             |         X  <-2.04- o                  
             o -GETN-> X          |    getN1
             o <-lock- o          |    
                               unlock
            Figure 6 scenario  
\end{verbatim}  
\end{small}  

This attack scenario is constructed by the function
\texttt{caFig6x}.
\begin{small}  
\begin{verbatim}  
        op caFig6x : Nat Nat Nat -> Conf .
        eq caFig6x(d:Nat,mqd:Nat,w4b:Nat) = 
           tCS2C(0,0,mkGetN("getN0","dev1","door") ;
                     mkPutN("putNU","dev1","door","unlock") ;
                     mkGetN("getN1","dev1","door") ,
                     mtR,mqd:Nat,w4b:Nat,
                 1,1,nilAM,rb("door","lock"),2,0,
                 drop drop delay(d:Nat)) .
\end{verbatim}  
\end{small}  

\paragraph{Figure 6 intended attack fails.}
To specify the intended attack we require that
the server sends a response to \texttt{"getN0"} before
sending a response to \texttt{"putNU"}, at the end
the door state is \texttt{"unlock"}, and the client
has a response to \texttt{"getN1"} with value
\texttt{"lock"}.
\begin{small}  
\begin{verbatim}  
              search {caFig6x(10,5,0)} =>! 
              {c:Conf ["eve" | a:Attrs caps(mtC)]} such that 
              hasRspTSnt(c:Conf,"dev1","dev0","getN0") and  
              hasRspTSnt(c:Conf,"dev1","dev0","putNU") and 
              rspTSntBefore(c:Conf,"dev1","dev0","getN0","putNU") and 
              checkRsrc(c:Conf,"dev1","door","unlock") and 
              not(hasRspTSnt(c:Conf,"dev1","dev0","getN1") ) and 
              hasGetRsp(c:Conf,"dev0","dev1","getN1","lock") .
\end{verbatim}  
\end{small}  
\noindent
There are no solutions, also no solutions in the
dialected case.

\paragraph{Figure 6 alternate attack succeeds}
In the alternate attack, the server requirement
is the same as for the intended attack, but
the client has a response to \texttt{"getN0"}
but no response to either \texttt{"putNU"} or
\texttt{"getN1"}, thus the door is unlocked at the end,
but the client does not know the state.
\begin{small}  
\begin{verbatim}  
            search {caFig6x(10,5,0)} =>! 
            {c:Conf ["eve" | a:Attrs caps(mtC)]} such that  
            hasRspTSnt(c:Conf,"dev1","dev0","getN0") and 
            hasRspTSnt(c:Conf,"dev1","dev0","putNU") and 
            rspTSntBefore(c:Conf,"dev1","dev0","getN0","putNU") and 
            checkRsrc(c:Conf,"dev1","door","unlock") and 
            not(hasRspTSnt(c:Conf,"dev1","dev0","getN1") ) and 
            hasGetRsp(c:Conf,"dev0","dev1","getN0","lock") and 
            rspPend(c:Conf,"dev0","dev1","getN1") and 
            rspPend(c:Conf,"dev0","dev1","putNU") .
\end{verbatim}  
\end{small}  
\noindent
There are 18 solutions in 2742 states visited.  In
the dialected case there are 18 solutions among 4675
states visited.

\subsection{Delaying and mismatching response from other resource}\label{subsec:diversion}
This attack is based on Figure 7 of \cite{coap-attacks}.
In the original attack the client sends temperature requests to two
resources (oven,room), and the 
attacker switches responses, dropping the the room
temperature
response, making the client think the room is on fire.
We adapt this to the case of doors to two rooms, to
keep to the door theme.

\begin{small}  
\begin{verbatim}  
              dev0      eve        dev1
              ---       ---     d1/lock d2/unlock
               o  ----GETN(d1)--->   o  getN0
               |             @ <-d1- o                  
               o -GETN(p2)-> X       |  getN1
               o <-d1-       o       |

              Figure 7 scenario      
\end{verbatim}  
\end{small}  

\noindent
With a vulnerable client the attack succeeds because the client uses same token for both GETs
and will accept the response to the first GET
as a response to the second GET request. 
The Maude  client does not reuse tokens and
this attack fails.

The attack scenario is constructed by the function
\texttt{caFig7x}.
\begin{small}  
\begin{verbatim}  
          op caFig7x : Nat Nat Nat -> Conf .
          eq caFig7x(d:Nat,mqd:Nat,w4b:Nat) = 
             tCS2C(0,0,mkGetN("getN0","dev1","door1") ;
                       mkGetN("getN1","dev1","door2") ,
                       mtR,mqd:Nat,w4b:Nat,
                   1,1,nilAM,
                   rb("door1","lock") rb("door2","unlock"),2,0,
                   drop delay(d:Nat)) .
\end{verbatim}  
\end{small}  

\paragraph{Figure 7 attack fails.}
\begin{small}  
\begin{verbatim}  
          search {caFig7x(10,5,0)} =>! 
            {c:Conf ["eve" | a:Attrs caps(mtC)]} such that
          hasRspTSnt(c:Conf,"dev1","dev0","getN0") and 
          not(hasRspTSnt(c:Conf,"dev1","dev0","getN1")) and 
          hasGetRsp(c:Conf,"dev0","dev1","getN1","lock") .
\end{verbatim}  
\end{small}  
\noindent
There is no solution.
Also no solution in the dialected case.

\subsection{Delaying and mismatching response from other resource variant}
This attack is a variant of the attack based on Figure 7 of \cite{coap-attacks} using two servers
rather than one server with two resources. 
\begin{small}  
\begin{verbatim}  
                dev0         eve      dev1    dev2
                ----         ---     unlock   lock
                o -GETN(1)--> #        |        |
                |             >  ---GETN(2)-->  o
                |             #  <---lock---    o 
                o <-1:lock-   <        |        |
\end{verbatim}  
\end{small}  

\noindent
 The client asks
for the door state at \texttt{dev1}.  The attacker redirects
the request to \texttt{dev2} and (un)redirects the response
to appear to come from \texttt{dev1}.   This scenario is
generated by the function \texttt{caFig7mod}.

\begin{small}  
\begin{verbatim} 
            op caFig7mod : Nat  Nat -> Conf .
            eq caFig7mod(mqd:Nat,w4b:Nat) =
                 tCS3C(0,0,mkGetN("getN0","dev1","door"),
                           mtR,mqd:Nat,w4b:Nat,
                      1,1,nilAM,rb("door","unlock"),5,0,
                      2,2,nilAM,rb("door","lock"),5,0,
                      redirect("dev1","dev2")       --- change tgt
                      unredirect("dev1","dev2")) .  --- change src
\end{verbatim}  
\end{small}  

\paragraph{Figure 7mod successful attack}
We search for a final state in which dev0 has a response from dev1
that the door is locked while in fact the door of dev1 is unlocked,
dev1 sent no response, and dev2 did send a response.
\begin{small}  
\begin{verbatim} 
              search {caFig7mod(5,0)} =>! 
              {c:Conf ["eve" | a:Attrs caps(mtC)]} such that 
              not(hasRspTSnt(c:Conf,"dev1","dev0","getN0")) and 
              hasRspTSnt(c:Conf,"dev2","dev0","getN0") and 
              checkRsrc(c:Conf,"dev1","door","unlock") and 
              hasGetRsp(c:Conf,"dev0","dev1","getN0","lock") .
\end{verbatim}  
\end{small}  
\noindent
There are 4 solutions out of 33 states visited.
In the dialected case there are no solutions
among 21 states visited. This is because the GET request is
dialected using the secret shared by dev0 and dev1, and dev2
can not decode the message, so it is dropped.


\section{Bisimulation  proofs}
\label{apx:sbsim-proofs}

We consider  CoAP system configurations \texttt{sysC} of the form

 \begin{small}
 \begin{verbatim}
   { [eid-j | devatts-j ] | 0 <= j < k 
     net(dmsgs0,dmsgs1)
     ["eve" | kb(dmsgs) caps(caps)]
     } 
 \end{verbatim}
 \end{small}

\noindent
and dialected CoAP system configurations, \texttt{dsysC}, of the form
 \begin{small}
 \begin{verbatim}
   { [eid-j | conf([eid-j | devatts'-j] 
                   net(dmsgsO-j,dmsgsI-j))
              ddevatts-j ]   
     | 0 <= j < k 
     net(ddmsgs0,ddmsgs1)
     [eve | kb(ddmsgs) caps(capsd)]
     } 
 \end{verbatim}
 \end{small}

In Section \ref{subsec:dialect-transform}
we defined transforms relating initial CoAP system
configurations and dialected versions
\texttt{D(sysI)} and 
\texttt{UD(dsysC)}.
We use these transforms to
define a simulation relation $\sim$ and prove two
stuttering bisimulation relations:

\begin{itemize}
\item[1.] $\mathtt{sysI} \sbsim
D(\mathtt{sysI})$ for initial
\texttt{sysI} with attack capabilities restricted
to \texttt{drop} or \texttt{delay}.
\item[2.]
  $\mathtt{dsysI} \sbsim \mathtt{UDA(dsysI)}$ for
initial \texttt{dsysI} with attack capabilites
\texttt{drop}, \texttt{delay}, \texttt{replay}, 
\texttt{redirect} where \texttt{UDA} extends \texttt{UD}
to remove \texttt{replay} and \texttt{redirect}
capabilities.
\end{itemize}
\noindent
Recall the notion of initial configuration is defined in Section \ref{subsec:coap-scenarios} definition \ref{defn:sysI}.

\begin{definition}[$\sim$-bisimulation]
Given a simulation relation $\mathtt{sysC} \sim \mathtt{dsysC}$ between system
configurations and dialected system configurations,
we say that $\mathtt{iSys}$ is $\sim$-stuttering bisimilar
to $\mathtt{iDSys}$ 
( $\mathtt{iSys} \sbsim_{\sim} \mathtt{iDSys}$)
iff the following hold:
\begin{itemize}
\item[1.]  
 $\mathtt{iSys} \sim \mathtt{iDSys}$
\item[2.]  
for reachable configurations such that 
$\mathtt{sysC} \sim \mathtt{dsysC}$ 
\begin{itemize}
\item[2a.]
if $\mathtt{sysC} \Rightarrow \mathtt{sysC1}$ then
there is some $\mathtt{dsysC1}$  with
$\mathtt{sysC1} \sim \mathtt{dsysC1}$ 
such that
   $\mathtt{dsysC} \Rightarrow^* \mathtt{dsysC1}$,
and
\item[2b.]
if $\mathtt{dsysC} \Rightarrow \mathtt{dsysC1}$ then
there is some $\mathtt{sysC1}$  with
$\mathtt{dsysC1} \sim \mathtt{sysC1}$ 
such that
   $\mathtt{sysC} \Rightarrow^* \mathtt{sysC1}$.
\end{itemize}
\end{itemize}
\noindent
We say $\mathtt{iSys} \sbsim \mathtt{iDSys}$ if there
is a simulation relation $\sim$ such that
$\mathtt{iSys} \sbsim_{\sim} \mathtt{iDSys}$.
\end{definition}

In the no-attacker case, intuitively the simulation
relation, $\mathtt{sysC} \sim \mathtt{dsysC}$,
holds if the two configurations differ only in
that in \texttt{dsysC}  there are meta-level attributes
and some messages are in the global net and are
dialected and some are in a local net while there is only one net is \texttt{sysC}.  

\begin{definition}[UDX and simulation]
\label{defn:ud-sim}
The simulation relation $\sim$ is defined by
the property
$$\mathtt{dsysC} \sim \mathtt{sysC} 
\equiv (\mathtt{UDX(dsysC)} = \mathtt{sysC})$$

\noindent
where, using the notation above, the function
$\mathtt{UDX}$ extends $\mathtt{UD}$ and the
result is derived as follows.
 \begin{small}
 \begin{verbatim}
  devatts-j = devatts'-j 0 <= j < k
  caps = capsd
  dmsgs0 = decode(ddmsgs0,U_{0<=j<k}ddevatts-j)  
           U_{0<=j<k} dmsgsO-j
  dmsgs1 = decode(ddmsgs1,U_{0<=j<k}ddevatts-j)  
             U_{0<=j<k} dmsgsI-j
 \end{verbatim}
 \end{small}
\noindent 
Note that \texttt{dmsgsI-j} has only delayed 
messages with target \texttt{eid-j} and delay 0,
and \texttt{dmsgsO-j} has only delayed 
messages with source \texttt{eid-j}.
\end{definition}

\begin{lemma}[Inverse relation]\label{lemma:DUD} The transformations \texttt{D} and \texttt{UD} are inverses.   For initial (dialected) systems
$$\mathtt{UD}(\mathtt{D}(\mathtt{sysI},\mathtt{ids})) = \mathtt{sysI}$$
and
$$\mathtt{D}(\mathtt{UD}(\mathtt{dsysI})) = \mathtt{dsysI}$$
\end{lemma}

Theorem \ref{theorem:bismmtC} tells us that
dialecting preserves properties of the underlying
protocol, in this case CoAP, expressible in
LTL (with out the next operator).  The corollary
says that allowing drop and/or delay attack capabilities does not change this.  That is,
dialecting does not protect against these
attack capabilities.
Theorem \ref{theorem:sbsimC} tells us that
dialecting  protects CoAP messaging
against replay an message redirection.

\begin{theorem}[Stuttering bisimulation in the absence of attacks.]\label{theorem:bismmtC}
For any initial system configuration \texttt{sysI}
with empty attack capabilities, 
  \texttt{sysI} is $\sim$-stuttering bisimilar to
  \texttt{D(sysI)} written
$$\mathtt{sysI} \sbsim \mathtt{D}(\mathtt{sysI})$$  
\end{theorem}    

\begin{proof}
Let \texttt{sysI} be an initial system
configuration with $\mathtt{caps} = \mathtt{mtC}$
and endpoint identifiers \texttt{eids}.
Let \texttt{dsysI} be \texttt{D(sysI)}.
According to the definition above, it suffices to show, 
  $\mathtt{sysI} \sbsim \mathtt{dsysI}$ using the relation $\sim$ the following.
\begin{itemize}
\item[(1)] $\mathtt{sysI} \sim \mathtt{dsysI} $ 
\item[(2)]  
  If \texttt{sysC} is reachable from \texttt{sysI}, \texttt{dsysC} is reachable from \texttt{dsysI}, and $\mathtt{sysC} \sim \mathtt{dsysC}$ then
 \begin{itemize}
   \item[](2.i)
     if  $\mathtt{sysC} \Rightarrow \mathtt{sysC1}$ 
     (application of one rewrite rule) then there is some \texttt{dsysC1} such that
  $\mathtt{dsysC} \Rightarrow^{<3} \mathtt{dsysC1}$ 
    and $\mathtt{sysC1} \sim \mathtt{dsysC1}$
    \item[](2.ii)  
   if $\mathtt{dsysC} \Rightarrow \mathtt{dsysC1}$ 
         then there is some \texttt{sysC1} such that
      $\mathtt{sysC} \sim \mathtt{dsysC1}$ or 
 $\mathtt{sysC} \Rightarrow \mathtt{sysC1}$ 
  and $\mathtt{sysC1} \sim \mathtt{dsysC1}$.
\end{itemize}
\end{itemize}

\noindent
(1) is by definition of $\sim$ and lemma 
\ref{lemma:DUD}.  
(2.i) is proved by cases on the rule applied.
Using the notation above, assume the rule applies to endpoint \texttt{eid-j}. 
\begin{itemize}
\item[]\texttt{crl[devsend]:}
By
similarity, \texttt{eid-j} in \texttt{dsysC} has
the same attributes and hence the same
application message to send. Therefore the same
rule applies to the wrapped \texttt{eid-j} with
the same attribute updates and new delayed
messages added to the local network.

\item[]\texttt{crl[rcv]:}
The
corresponding message in \texttt{dsysC} has the
same delay and is either in the local net of
\texttt{eid-j} or in the global net. In the
former case, the rule applies in \texttt{dsysC}
with $\mathtt{dsysC} \Rightarrow
\mathtt{dsysC1}$. In the latter case, the rule
sequence \texttt{ddevrcv ; rcv} applies with
$\mathtt{dsysC} \Rightarrow^2 \mathtt{dsysC2}$.
In either case the new state is given by the
function \texttt{rcvMsg(epid,devatts,msg)} in
\texttt{dsysC} and \texttt{sysC}. Thus
$\mathtt{dsysC1} = \mathtt{dsysC2} \sim
\mathtt{sysC1}$.

\item[]\texttt{crl[ackTimeout]:}
By $\sim$ this rule is enabled in endpoint \texttt{eid-j} in \texttt{dsysC}.
Since the rule outcome depends only \texttt{devatts-j} we have
  $\mathtt{dsysC} \Rightarrow_{\mathtt{ackTimeout}} \mathtt{dsysC1} \sim \mathtt{sysC1}$.

\item[]\texttt{crl[net]:}
If the corresponding message is in the global net
then 
$\mathtt{dsysC} \Rightarrow_{\mathtt{net}}  \mathtt{dsysC} \sim \mathtt{sysC}$.
Otherwise the corresponding message is in the local net and 
$\mathtt{dsysC} \Rightarrow_{\mathtt{ddevsend;net}} \mathtt{dsysC1} \sim \mathtt{sysC1}$.

\item[]\texttt{crl[tick]:}
By definition of \texttt{mte}, if \texttt{tick} is not enabled in \texttt{dsysC} then the local net stub is not empty, and by $\sim$  the only local messages are outgoing (incoming messages have delay 0).  Thus 
$$\mathtt{dsysC} \Rightarrow_{\mathtt{tick;ddevsend+}}
\mathtt{dsysC1} \sim \mathtt{sysC1}$$
\end{itemize}

(2.ii) is also proved by cases on the rule applied.
\begin{itemize}
\item[]\texttt{crl[ddevsend]:}
A message is moved between a local net and the
global net, thus $\mathtt{dsysC1} \sim \mathtt{sysC}$
by just revising the network message correspondence.
          
\item[]\texttt{crl[ddevrcv]:}
A message is moved between  the global net and a local net, thus
 $\mathtt{dsysC1} \sim \mathtt{sysC}$.
\item[]\texttt{crl[devsend]:}
 by similarity, the corresponding rule instance is enabled in 
 \texttt{sysC} and the effect depends only on the base endpoint attributes and the message
 which are the same in \texttt{sysC} and \texttt{dsysC}.  Thus 
 $\mathtt{sysC} \Rightarrow_{\mathtt{devsnd}} 
 \mathtt{sysC1} \sim \mathtt{dsysC1}$
\item[]\texttt{crl[rcv]:}
  similar to \texttt{devsend} and (2.i)
\item[]\texttt{crl[ackTimeout]:}
  similar to  (2.1)
\item[]\texttt{crl[net]:}
  by $\sim$ the \texttt{sysC} net also has the message to move.
\item[]\texttt{crl[tick]:}
  if \texttt{mte} is a non-zero Nat in \texttt{dsysC} then the local stub networks
  are all empty and the \texttt{mte} for sysC has the same value. Thus
  $\mathtt{sysC} \Rightarrow_{\mathtt{tick}} \mathtt{sysC1} \sim \mathtt{dsysC1}$ since the function \texttt{passTime} has the same effect in both systems.
\end{itemize}
\end{proof}

\begin{corollary}
The bisimulation result of Theorem \ref{theorem:bismmtC} extends to the case that \texttt{caps} allows delay and/or drop.
\end{corollary}
\begin{proof}
We only need to consider the rule \texttt{attack}.
delay/drop attacks do not concern the message
content or the deliverability of dialected
messages, thus an attack rule application in
\texttt{sysC} has a corresponding application in
\texttt{dsysC} (possible combined with a ddevsend
instance) and vice-versa, with the result being
$\sim$ related configurations.
\end{proof}

\begin{theorem} \label{theorem:sbsimC}
Let \texttt{DC(sysI,caps)} be  \texttt{D(sysI)} with
\texttt{caps} added to the attacker capabilities
of the dialected system configuration for any
initial system configuration \texttt{sysI}.  Then
for \texttt{sysI} with attacker capabilities
\texttt{drop} and \texttt{delay} 
and caps containing at most \texttt{replay},\texttt{redirect} capabilities
$$\mathtt{sysI} \sbsim \mathtt{DC}(\mathtt{sysI},\mathtt{caps})$$
\end{theorem}
\begin{proof}

Imagine execution of dialected systems annotates each message with the number of the step that
generated/sent it (the \texttt{devsend} rule). Copies of messages generated by
the attack rule using the \texttt{replay} capability are annotated by the creation step of the original and the replay
attack step that created the copy; and redirected messages are annotated by the creation step and
the redirection step. The CoAP rules are unchanged.
Define the function \texttt{UDA} to adapt \texttt{UD}
so that the annotation is dropped from single annotation messages and for doubly annotated messages,  
only the original message (sans annotation)
is in the network image. 
Copies are dropped
and redirected messages are dropped from the network image.  Furthermore, attacker capabilities are reduced to \texttt{delay},\texttt{drop}.

Now we show $\mathtt{dsysC} \sim
\mathtt{UDA}(\mathtt{dsysC})$ is a stuttering
bisimulation on configurations reachable from
\texttt{\texttt{D}(\texttt{sysI},\texttt{dcaps})}
where \texttt{dcaps} contains only
\texttt{replay},\texttt{redirect} capabilities and
the attack capabilities of \texttt{sysI} are among
\texttt{drop}, \texttt{delay}.

\noindent
\textbf{1.} For the $\mathtt{dsysC} \Rightarrow \mathtt{dsysC1}$ direction let $\mathtt{sysC} = \mathtt{UDA}(\mathtt{dsysC})$ and show
$\mathtt{sysC} \Rightarrow^* \mathtt{sysC1} = \mathtt{UDA}(\mathtt{dsysC1})$ If the
rule used is not the attack rule using replay or redirect capability or
delivery of the target message of such an attack, 
this follows by the proof of Theorem \ref{theorem:bismmtC} and its corollary. 
We consider the rules
in remaining cases.
\begin{itemize}
\item 
\texttt{[attack(replay(n))]} at step k.
Let $\texttt{ddm}^j$ be the message being replayed ($j < k$) and 
$\mathtt{ddm1}^{j,k} = delay(\mathtt{ddm},n)^{j.k}$ be the copy. Then $\mathtt{UDA}(\mathtt{dsysC}) => \mathtt{UDA}\mathtt{(dsysC1)}$ by the \texttt{net} rule since $\mathtt{ddm1}^{j,k}$
  is not in the \texttt{UDA} image but the attack
  moves the original from the first to the second
  network component.

\item \texttt{[ddrcv]} of a replayed message that is delivered to the wrapped endpoint.
If this is the original, then \texttt{sysC} can also receive the message. If this is the copy (but the original has not been received, thus the copy was not further delayed) then again \texttt{sysC} can receive the original.

\item \texttt{[ddrcv]} and drop a replayed message.
Either the original or copy has already been received and the replayed message will be discarded
by \texttt{UDA} thus $\mathtt{sysC} = \mathtt{sysC1} = \mathtt{UDA}(\mathtt{dsysC1})$.

 \item \texttt{[attack(redirect(epid0,epid1))]}.
 Here $\mathtt{sysC} \rightarrow_{\mathtt{drop}} \mathtt{sysC1} = \mathtt{UDA}(\mathtt{dsysC1})$  since redirected messages are dropped.

\item \texttt{[ddrcv]} of a redirected message.
This has the effect of a \texttt{drop}. The 
redirected message has already been dropped 
in \texttt{sysC}, and is not in the \texttt{UDA} image.  Thus  
$\mathtt{sysC} = \mathtt{UDA}(\mathtt{dsysC}) = \mathtt{UDA}(\mathtt{dsysC1}) = \mathtt{sysC1}$.
\end{itemize}

\noindent 
\textbf{2.} For the  $\mathtt{sysC} \Rightarrow \mathtt{sysC1}$ direction, for any \texttt{dsysC} such that $\mathtt{UDA(dsysC)} = \mathtt{sysC}$
the argument that $\mathtt{dsysC} \Rightarrow^{<3}  \mathtt{dsysC1}$  with $\mathtt{UDA}(\mathtt{dsysC1}) = \mathtt{sysC1}$ 
is the same as the proof of Theorem \ref{theorem:bismmtC} as the same rules are available
for \texttt{sysC}.
\end{proof}

\begin{corollary}
With notation as for Theorem \ref{theorem:sbsimC} and letting
\emph{caps} be a reactive attacker capability set we
have:
$\mathtt{sysI} \sbsim \mathtt{DC}(\mathtt{sysI}, caps)$
\end{corollary}
\begin{proof}
Reactive caps have the form \texttt{mc(tpat,spat,false,acts)} where
acts is a set of elements of the form \texttt{act(tp,sp,n)} with \texttt{tpat,tp} being message
target patterns
and \texttt{spat,sp} being message source patterns.
An attack using  capability
\texttt{mc(tpat,spat,false,acts)} at step $k$ with matching target
$M^j$ generated at step $j < k$ is a set of messages
$M^j, M_1^{j,k} \ldots M_l^{j,k}$ ($M_i^{j,k}$ is the result of action $i$ and behave the same as a delayed
or redirected message).
Thus the bisimulation argument for reactive attacks
is the same as that used 
to prove Theorem \ref{theorem:sbsimC}.
\end{proof}



\section{Application Level Experiments}
\label{apx:app-experiments}

To study attack models in scenarios with richer
message passing patterns we define a simple
application layer to drive CoAP messaging, and a
generic reactive attack model to support a more
symbolic search for attacks. We also define basic
state properties useful in specifying application
properties vulnerable to attack. Two simple
applications are defined with associated
safety invariants. The bridge application
models automatic control of a moveable bridge over
a waterway. The Pick-n-Place (PnP) application
models a smart manufacturing tool that picks up
parts from one place an puts them down at another
place -- for example conveyor belts. Reachability
analysis is used to verify the invariants in the
absence of attack, and to find attacks given
different levels of attacker capability. Repeating
the attack scenarios with dialected versions of
the applications shows that dialecting protects
against reactive attackers.

\subsection{Application Level Specification}
\label{asubsec:app-level}

An application consists of a set of CoAP devices
(endpoints), each with its set of resources, and
one or more having an application layer playing a
specific role in the application. When a message
is received, it is processed at the CoAP level as
usual and then passed to the application layer, if
any, for additional processing. The application
layer has a knowledge base (AKB) binding names to
values and a set of rules specifying response
to messages. A rule consists of a message pattern
and a set of conditional actions. A message is
processed at the application level by executing
the conditional actions associated to each rule
whose pattern matches the message. An action can
update the application knowledge base or the CoAP
layer resource mapping, or create an application
message to be sent.
In the following we give a little more detail
of the application layer.

The AKB has the same structure as the CoAP layer
resource map--a multiset of elements of the form
\begin{small}
\begin{verbatim}
       rb(name,val)
\end{verbatim}
\end{small}
\noindent
where \texttt{name} and \texttt{val} are 
strings.  

An application rule (sort \texttt{ARule}) has the
form
\begin{small}
\begin{verbatim}
        ar(mpat,cacts)
\end{verbatim}
\end{small}
\noindent
where \texttt{mpat} is a message pattern (sort
\texttt{MPat}) and \texttt{cacts} is a multiset
of conditional actions (sort \texttt{CAct}).

A message pattern is either a pattern to match
request messages
or a pattern to match response messages.  (Recall that a CoAP message is either a request
or a response or an acknowledgment.  We do not let
the application layer mess with the CoAP layer
reliability mechanism, so no rule will match
a purely ACK message.

A request  pattern
\begin{small}
\begin{verbatim}
        req(srcP,methP,pathP,valP) 
\end{verbatim}
\end{small}
\noindent
is used to match a request message. A message matches
this pattern if the message source matches
\texttt{srcP}, the message method matches
\texttt{methP}, resource path of the message matches
\texttt{pathP}, and the body, if any, matches
\texttt{valP}. The argument patterns are either
variables (\texttt{v(vname)}) or string constants.
Constants must match exactly, variables match any
string, and if the message pattern matches a message,
the variable bindings are used to instantiate variables
in the conditional actions. (In the case of an
empty body, the value pattern is ignored.)

A response message matches the response pattern
\begin{small}
\begin{verbatim}
        rsp(srcP,amidP,bool,valP) 
\end{verbatim}
\end{small}
\noindent
if, as for requests, \texttt{srcP} matches the message source, \texttt{valP} matches the message body, \texttt{amidP} matches the message token (that identifies the associated request) and bool is \texttt{true} iff the response code is a success code.

A conditional action has the form 
\texttt{ca(cond,acts)}. A condition
is an equality, a disequality, a disjunction
or a conjunction:
\begin{small}
\begin{verbatim}
        eq(pat0,pat1)
        neq(pat0,pat1)
        disj(conds)
        conj(conds)
\end{verbatim}
\end{small}
\noindent
where \texttt{pat0}, \texttt{pat1} are variables
or string constants as above, and \texttt{conds}
is a set of conditions. \texttt{conj(none)}
corresponds to \texttt{true} and
\texttt{disj(none)} corresponds to \texttt{false}

An action has one of the following forms
\begin{small}
\begin{verbatim}
        send(amid,tgt,type,meth,path,val)
        set(var,val)  
        put(var,val)
\end{verbatim}
\end{small}
\noindent

The action arguments are variables or string
constants. When an action is executed the
variables are interpreted in a context
consisting of binding resulting from the message
pattern matching, the AKB, and the CoAP resource
map. The \texttt{send} action constructs an
application message and puts it in the CoAP
layer \texttt{sendReqs} attribute. The
\texttt{set} action updates the AKB, while the
\texttt{put} action updates the CoAP layer
resource map.

As an example, here is the rule from the bridge controller rule set for handling a message from the boat sensor signalling that a boat has arrived 
(by putting the \texttt{"boat"} resource
to be \texttt{"here"}).
\begin{small}
\begin{verbatim}
        op rcvBoatArr : -> ARule .
        eq rcvBoatArr  =
           ar(req("bs","PUT","boat","here"),
              ca(eq(v("status"),"idle"),
                 send("GateCL","ga","NON", 
                      "PUT","gate", "close")
                 set("status","working") ) ) .
\end{verbatim}
\end{small}
\noindent
The rule \texttt{rcvBoatArr} has one conditional
action, the condition being that the value of
\texttt{"status"} is \texttt{"idle"} (in the AKB).
If so, the controller sends a request to the gate
guarding traffic over the bridge to close and sets
the controller \texttt{"status"} variable to
\texttt{"working"}. The string \texttt{"GateCL"}
is used to construct the token that identifies the
gate close request message and its response.

The rule \texttt{rcvGateClose} specifies what
the controller should do when the gate responds that the gate is successfully closed.
\begin{small}
\begin{verbatim}
          op rcvGateClose : -> ARule .
          eq rcvGateClose =
             ar(rsp("ga","GateCL",true,""),
                ca(conj(none), 
                   send("BridgeOp","br","NON",
                        "PUT","bridge","open"))) .
\end{verbatim}
\end{small}
\noindent

The response message is identified by the
expected sender, \texttt{"ga"}, and the message
identifier, \texttt{"GateCL"}. The condition is
\texttt{true} and the action is to request the
bridge to open, now that the gate is closed (and
any traffic cleared).

To add an application layer to CoAP messaging, an
attribute \texttt{aconf} is added to the attributes of
devices with a non-trivial application layer. This
attribute has two arguments: \texttt{akb}, the
application knowledge base; and \texttt{arules}, the
application rule set. The CoAP receive rule (and the
dialect receive rule) are extended with a call to the
\texttt{doApp} function following the call to the
\texttt{rcvMsg} function.

\begin{small}
\begin{verbatim}
    /\ toSend(dmsgs) devatts1 := rcvMsg(epid, devatts, msg) --- original
    /\ devatts0 := doApp(msg,devatts1)  --- new
\end{verbatim}
\end{small}
\noindent

The \texttt{doApp} function finds the matching
rules from the application layer, if any, and
executes the associated conditional actions,
returning the updated attribute set, which may
include updates to the AKB, the resource map,
and/or the \texttt{sendReqs} attribute.

\subsection{Generic Reactive Attack Model}
\label{asubsec:generic-attack}

We generalize the reactive attacker model
 to allow search to do more of the work
of finding attacks.
For this we introduce the generic attack capability
\texttt{mcX(n)} that operates on \texttt{GET} or
\texttt{PUT} request messages. Given a request
message, \texttt{m(tgt,src,c) @ d}, in the network
input part, and an attack capability \texttt{mcX(n)},
the attack rule non-deterministically selects 
a message to add to the net output part from
a set messages of
the form \texttt{m(tgt0,src0,c) @@ d + n} where
\texttt{tgt0} identifies an endpoint that has the
resource path mentioned in the request content.
The delay construct \texttt{@@} behaves like
the original \texttt{@} except for delay 0. A
message \texttt{msg @@ 0} may be delivered at the
current time or at any later time.  In the case
of a GET request, a redirect capability
\texttt{mc(src0,tgt0,false,act(src,tgt,0)} is
added to the attacker capability set so that
the redirected version of \texttt{GET} response will be accepted by the sender, otherwise the attack will  have no effect.  

The ability to observe the method and path of a
message is not an ability of the attacker, but
rather a mechanism to prune the search space to
eliminate useless branches. This is consistent
with the goal of simply finding attacks, rather
than being concerned with the difficulty or
probability of a successful attack.

\subsection{Basic Application State Properties}
\label{asubsec:basic-app-properties}

The application model is based on the state of
CoAP resources and an interpretation of  PUTing
of some resource values as actuator commands.
We interpret receipt and processing of such PUTs
as initiating a command, and receipt of
the success response as completion of the
command.  This interpretation is formalized
by the following properties for constructing application invariants.
\begin{itemize}
\item \texttt{hasV(conf,epid,path,val)} is true in
a configuration \texttt{conf} if the device with
identifier \texttt{epid} has the binding
\texttt{rb(path,val)} in its resource map.

\item \texttt{hasAV(conf,epid,path,val)} is true
in a configuration \texttt{conf} if the device
with identifier \texttt{epid} has the binding
\texttt{rb(path,val)} in its application layer AKB

\item \texttt{isV(conf,ctl, epid, aid, path, val)}
is true in a configuration \texttt{conf} if
\texttt{hasV(conf,epid,path,val)} and the device
with identifier \texttt{ctl} does not have a
pending response for a message with token matching
\texttt{aid}. This is used in a context where
\texttt{ctl} is the device that sent a message
setting the resource at \texttt{path} and the
response has been received, indicating the
associated action/command has completed.

\item \texttt{becomeV(conf,ctl,epid,aid, path,
val)} is true in a configuration \texttt{conf} if
\texttt{hasV(conf,epid,path,val)} holds and  the
device with identifier \texttt{ctl} has pending
response for a message with token matching
\texttt{aid}. This is used in a context where
\texttt{ctl} is the device that sent a message
setting the resource at \texttt{path} and while
the response is in transit, the associated
action/command is ongoing.

\item \texttt{aKbNotTok(conf,eveid, tok)} is true
in a configuration \texttt{conf} if the attacker
(identifier \texttt{eveid}) has a message with
token \texttt{tok} in its knowledge base (of
seen/attacked messages). This can be used to force
search to find alternative attacks.

\end{itemize}

\subsection{Moving Bridge Control Application}
\label{asubsec:bridge-scenario}

The context of this application is a waterway, for
example a canal in Amsterdam, with a bridge used by
vehicles and pedestrians to cross the waterway. For
boats to pass under the bridge, the bridge must open.
Of course, the bridge must be cleared of traffic and
traffic must be blocked while the bridge is open. The
bridge application automates this process.
The roles of the application are
\begin{itemize}
\item \texttt{bctl}: Bridge Controller
\item \texttt{bs}: Boat Sensor 
-- notifies bctl when a boat wants to pass
\item \texttt{ga}: Gate--controls bridge traffic at each end of the bridge
\item \texttt{br}: Bridge movable part
\end{itemize}

Only the bridge controller, \texttt{bctl}, has an application layer. The other roles are modeled by the CoAP message semantics. The bridge controller has states \emph{idle} (initially) or \emph{working}. The gate is either \emph{open} (initial state), \emph{closing}, \emph{closed}, or \emph{opening}. The bridge is either \emph{closed} (initial state), \emph{opening}, \emph{open}, or \emph{closing}.

Informally the bridge application protocol is given
by the following message sequence 

\begin{small}
\begin{verbatim}
        bs -> bctl : boat here -- a boat wants to pass
        bctl becomes working
        bctl -> bs : received
        bctl -> ga : close -- clear traffic and close
        ga -> bctl : success -- gate is closed
        bctl -> br : open
        br -> bctl : sucess  -- bridge is open
        bctl -> bs : pass -- boat can pass
        bs -> bctl : success -- boat has passed
        bctl -> br : close 
        br -> bctl : success  -- bridge is closed
        bctl -> ga : open
        ga -> bctl : success  -- gate is open
        bctl becomes idle
\end{verbatim}
\end{small}
\noindent
To simplify specification of bridge application scenario initial configurations we define a function  \texttt{mkDevA} to construct
initial configurations for the participating
endpoints.  The arguments consist of 
\texttt{amsgl} (application messages to be sent);
\texttt{rbnds} (the CoAP level resource map);
\texttt{(abnds,arules)} (the application layer knowledge base and rules); \texttt{j} the initial
delay for sending messages; and \texttt{msgSD}
(the delay before a newly sent message can be delivered).  The function \texttt{mkInitDevAttrs}
generates the CoAP layer attributes as for the
experiments described in Sections  \ref{sec:reactive-attacks} and \ref{apx:coap-vulnerabilities}.

\begin{small}
\begin{verbatim}
       mkDevA(epid,j,amsgl,rbnds,msgSD,abnds,arules) 
         =
        [ epid | sendReqs(amsgl) rsrcs(rbnds)
                 aconf(abnds,arules)
                 sndCtr(j) mkInitDevAttrs(msgSD) ]  .
\end{verbatim}
\end{small}
\noindent
In the experiments we used two initial configurations,\
\texttt{brInit} that executes a single round of the protocol and \texttt{brInit2(n)} that executes two
rounds of the protocol, the second round delayed by
\texttt{n}.
 
\begin{small}
\begin{verbatim}
        eq brInit = bctl bs ga br net(mtDM,mtDM) .
        eq brInit2(n:Nat) = 
              bctl bs2(n:Nat) ga br net(mtDM,mtDM) .
\end{verbatim}
\end{small}
\noindent
where the bridge controller, \texttt{bctl}, is defined
by
\begin{small}
\begin{verbatim}
         bctl = mkDevA("bctl",1,nilAM,rb("boat", "none"),
                       2,rb("status","idle"), bridge-rules) .
\end{verbatim}
\end{small}
\noindent
The boat sensor versions are defined by
\begin{small}
\begin{verbatim}
  eq bs = mkDevA("bs",1,boatHereAMsg,mtR,6,mtR,none) .
  eq bs2(n:Nat) = mkDevA("bs",1,
           boatHereAMsg ; amsgd(n:Nat) ; boatHereAMsg, 
           mtR,6, mtR,none) .
\end{verbatim}
\end{small}
\noindent
where \texttt{boatHereAMsg} is a \texttt{PUT} request to
\texttt{"bctl"} to set \texttt{"boat"} to \texttt{"here"}.
The initial configurations of the gate, \texttt{ga}
initially \texttt{"open",} and bridge, \texttt{br}
initially \texttt{"close"}, are defined as follows.

\begin{small}
\begin{verbatim}
      eq ga = 
        mkDevA("ga",1,nilAM,rb("gate","open"),4,mtR,none) .
      eq br = 
        mkDevA("br",1,nilAM,rb("bridge","close"),6,mtR,none) .
\end{verbatim}
\end{small}
\noindent
There are 6 rules for the controller application layer.
We list the message ids of any messages sent by each
rule as they are used in specifying properties of application
executions.
\begin{itemize}
\item
\texttt{rcvBoatArr}: in response to message signaling a boat
waiting to pass.  The controller requests the gate to close
(message id \texttt{"GateCL"}). Discussed in Section
\ref{asubsec:app-level}.

\item
\texttt{rcvGateClose}: in response to confirmation that the
gate is closed, the controller requests the bridge to open
(message id \texttt{"BridgeOp"}). Discussed in Section
\ref{asubsec:app-level}.

\item
\texttt{rcvBridgeOpen}: in response to confirmation
that the bridge is open, the controller sends a message
to the boat sensor to allow the boat to pass (message id 
\texttt{"BSPass"}).

\item \texttt{rcvBoatPass}:  in response to confirmation
that the boat has passed, the controller requests the
bridge to close (message id \texttt{"BridgeCl"}).

\item \texttt{rcvBridgeClose}:  in response to confirmation
that the bridge is closed, the controller requests the
gate to open (message id \texttt{"GateOp"}).

\item \texttt{rcvGateOpen}:  in response to confirmation
that the gate is open, the controller sets its status to
\texttt{"idle"} ready for the next boat.
\end{itemize}

We require that the following invariants hold
for bridge application instances. For each
invariant we define the boolean function
corresponding to the negation and use that to
search for violations under different
conditions (without or with attacker, without
or with a dialect layer).

\begin{itemize}
\item \texttt{bclIdle}: If the bridge controller status is \texttt{"idle"}
then the bridge is closed and the gate is open.
The negation of this property is defined by
\begin{small}
\begin{verbatim}
        bclIdleInv(conf,bcid,brid,gid)
\end{verbatim}
\end{small}
\noindent
where \texttt{bcid} is the controller identifier,
\texttt{brid} is the bridge identifier and \texttt{gid} is
the gate identifier.
 
\item \texttt{gateNCl}: If the gate is not closed (thus it is opening, open,
or closing) then the bridge is closed. The negation is
formalized by the function
\begin{small}
\begin{verbatim}
        gateNClInv(conf,bcid,brid,gid) .
\end{verbatim}
\end{small}
\noindent

\item  \texttt{brNCl}: If the bridge is opening, open, or closing then the gate is closed. The negation is formalized by the function
\begin{small}
\begin{verbatim}
        brNClInv(c:Conf,bcid,brid,gid) .
\end{verbatim}
\end{small}

\item \texttt{boatPass}: If a boats is passing then the bridge is open and the gate is closed. The negation is formalized by the function

\begin{small}
\begin{verbatim}
         boatPassInv(conf,bcid,bsid,brid,gid) . 
\end{verbatim}
\end{small}
\noindent
Here \texttt{bsid} is the boat sensor identifier.
\end{itemize}

As an example, the function \texttt{bcIdleInv}
is defined as follows.

\begin{small}
\begin{verbatim}
  op bclIdleInv : Conf String String String -> Bool .
  eq bclIdleInv(conf,bcid,brid,gid) = 
       bcIdle(conf,bcid) and
        (not(brClose(conf,bcid,brid)) or 
         not(gaClose(conf,bcid,gid))) .
\end{verbatim}
\end{small}
\noindent

For each invariant, we search for violations for different initial configurations: \texttt{brInit} and \texttt{brInit(40)} without and with addition of an attacker, and dialected versions of each of the above. 
The search command for violation of the \texttt{bclIdle} invariant during one round of the protocol with no attacker is
\begin{small}
\begin{verbatim}
    search {brInit} =>+ 
           {c:Conf} such that 
           bclIdleInv(c:Conf,"bctl","br","ga") .
\end{verbatim}
\end{small}
\noindent
Replacing \texttt{brInit} by \texttt{brInit2(40)} we obtain
the search command for two rounds of the protocol.
\begin{small}
\begin{verbatim}
    search {brInit2(40)} =>+ 
           {c:Conf} such that 
           bclIdleInv(c:Conf,"bctl","br","ga") .
\end{verbatim}
\end{small}
\noindent
The search command for attacks on the \texttt{bclIdle}
invariant during one round of the protocol is obtained
by adding an attacker agent with capability \texttt{mcX(20)}.
The configuration element \texttt{log(nilLI)} records 
message receive events, thus giving a simplified trace
for successful attacks.
\begin{small}
\begin{verbatim}
    search {brInit log(nilLI) ["eve" | kb(mtDM) caps(mcX(20))]} =>+ 
           {c:Conf} such that 
            bclIdleInv(c:Conf,"bctl","br","ga") .
\end{verbatim}
\end{small}
\noindent
The dialected form of search for attacks on the
\texttt{bclIdle} invariant (one round) is
\begin{small}
\begin{verbatim}
    search D({brInit log(nilLI) ["eve" | kb(mtDM) caps(mcX(20))]}) =>+ 
          {c:Conf} such that 
          bclIdleInv(UDC(c:Conf),"bctl","br","ga") .
\end{verbatim}
\end{small}
\noindent
Recall that the functions \texttt{D} and \texttt{UDC} are the dialecting and undialecting tranform defined in Section \ref{subsec:dialect-transform}. 
Searches for violations in the absence of attacker
and for attacks in dialected case return \texttt{no solution}.  Table \ref{table:bridgeexp}
summarizes the attacks found in the other cases. 

\begin{figure}[ht]
\begin{small}
\begin{verbatim}
    Invariant   nRnd  mcX  msg
    bclIdleInv    1   20   GateCl,BridgeOp
                  2   40   GateCl,BridgeOp
    brNClInv      1   20   BridgeOp
                  2   20   BridgeOp
    gateNClInv    1   20   BridgeOp
                  2   20   BridgeOp
    boatPassInv   1  20-40 none
                  2   20   BridgeCl,GateOp 
\end{verbatim}
\end{small}
\caption{Summary of bridge application attacks. The column \texttt{nRnd} is the number of rounds, \texttt{mcX} is delay argument to the \texttt{mcX} attack capability, and \texttt{msg} is the message
identifier of the attacked message.}  
\label{table:bridgeexp}
\end{figure}
We note that the gate and bridge open and close commands
are vulnerable to attack, while
the  boat here and  boat pass 
messages have no attacks against the movable bridge
invariants.

\subsection{Pick-n-Place Case Study}
\label{asubsec:pnp-scenario}

The basic Pick-n-place application consists of a
controller, an arm (moving on a track), and a
gripper. When the controller is told there is an
item available, it tells the arm to move to the
opposite end of its track. When the arm reports
success, the controller tells the gripper to
close (to pick up the item below it). Then the
controller tells the arm to return to its
original position and then tells the gripper to
open, thus placing the item picked. Copies of
this basic application can be combined in
different ways for more complex transfer of
items.

We present the PnP protocol in terms of an
abstract initial position of the arm
\texttt{atI} and abstract arm commands
\texttt{goNI} (go from the initial position on
the arm track to the opposite end) and
\texttt{goI} (return to the initial position).
If \texttt{I} is left the \texttt{goI} becomes
\texttt{goL} and \texttt{goNI} becomes
\texttt{goR}. Dually, if the initial position is
right.
An informal version of the application protocol follows.
\begin{small}
\begin{verbatim}
    Roles and states:
      pctl : Controller  
            status: "idle", "working"
      arm : 
         -- Arm resource values
            arm inital position, drop location
              "arm" = "goI", 
            opposite position, pick location
              "arm" = "goNI", 
         -- arm states
            goingI -- put response in transit
            atI  --   put response received
                        (or initial state)
            goingNI -- put response in transit
            atNI    -- put response received
      gr : 
        -- gripper resource values
             "grip"  == "open"
             "grip"  == "close"
        -- gripper states   
           opening --- dropping if loaded
                      put response in transit
           open       put response received
           closing --- gripping if at pick location
                       put response in transit
           close      put response received
      ps : Part Sensor -- external driver
\end{verbatim}
\end{small}
\noindent
Pick-n-place message exchange for arm
with initial position at the left.
\begin{small}
\begin{verbatim}
      pctl status is "idle"
      ps->pctl:  Put start
      pctl status becomes "working"
      pctl->arm: Put arm goR
      arm->pctl: success  (at right)
      pctl->grip: PUT grip close
      grip->pctl: success
      pctl->arm: Put  arm goL
      arm->pctl: success  (at left)
      pctl->grip: PUT grip open
      grip->pctl: success
      pctl status becomes "idle"
      pctl->ps success
\end{verbatim}
\end{small}
\noindent

Only the pnp controller has an application
layer, the other roles simply use
CoAP messaging.
There are six controller rules.

\begin{itemize}
\item  \texttt{rcvStart}:, upon receipt of a \texttt{"start"} message, if the control status is \texttt{"idle"}, then
the controller sends a \texttt{goNI} request to its arm (message id \texttt{"ArmGoNI"}).  It also  sets \texttt{"status"} to\texttt{ \texttt{"working"}} and
\texttt{"source"} to the message sender.

\item \texttt{rcvAtNI}: upon receipt of a response to the \texttt{"ArmGoI"} request, the controller sends a \texttt{"close" } request to the gripper (message id \texttt{"GripCl"}).
            
\item \texttt{rcvGrCl}: upon receipt of a response to the \texttt{"GripCl"} request, the controller sends a
a \texttt{goI} request to the  arm (message id \texttt{"ArmGoI")}.

\item \texttt{rcvAtI}:  upon receipt of a response to the \texttt{"ArmGoI"} request, the controller sends a
a \texttt{goI} request to the  arm (message id \texttt{"GripOp"}).

\item  \texttt{rcvGrOp}: upon receipt of a response to the \texttt{"GripOp"} request, the controller sends a
a message to the  source of the start message to indicate that the item has been transferred (message id \texttt{"PnPDone"}).

\item \texttt{rcvPnPDone}:  upon receipt of a response to the \texttt{"PnPDone"} request, the controller sets its status
to \texttt{"idle"} (and waits for a\texttt{"start"} 
request).
\end{itemize}
\noindent

To experiment with applications with a single PnP we define
the function
\begin{small}
\begin{verbatim}
      initRL(pid,gid,aid,amsgl) 
\end{verbatim}
\end{small}
\noindent
that constructs an initial configuration  with 
four endpoints.
\begin{itemize}
\item
  \texttt{ctlRL}: a controller with identifier 
  \texttt{pid} with initial knowledge base
  \begin{small}
  \begin{verbatim}
       rb("status","idle")
       rb("myarm",aid)  rb("mygrip",gid)
       rb("goNI","goR") rb("goI","goL")
   \end{verbatim}
   \end{small}
\item
  \texttt{arm}: an arm with identifier 
  \texttt{aid} with  initial resource bindings rb(\texttt{"arm"},\texttt{"goL"})
 indicating the initial position at the left.
\item
\texttt{grip}: a gripper with identifier  \texttt{gid} with
 initial resource bindings
       rb(\texttt{"grip"},\texttt{"open"}) 
\item  
  \texttt{ps}: a part sensor with identifier \texttt{"ps"}.
 and attribure \texttt{sendReqs(amsgl)}, one or more
 start messages for the controller, separated by a delay.
\end{itemize}
\noindent

The following are application invariants required for
correct / safe execution. We formalize the negation as the
property to search for in the presence of an attacker.

\begin{itemize}
\item \texttt{pnpIdle}: If the PnP controller status is \texttt{"idle"}
then the arm is in its initial position \texttt{rb("arm",goI)}
and the grip is open.
The negation of this property is formalized by
 \texttt{ pnpIdleInv(c:Conf,pid,gid,aid,goI)}
where \texttt{pid} is the controller identifier,
\texttt{gid} is the grip identifier and \texttt{aid} is
the arm identifier.  This term equationally reduces to
\begin{small}
\begin{verbatim}
   hasAV(c:Conf,pid,"status","idle")
   and    
  (not(isV(c:Conf,pid,aid,"ArmGoI","arm",goI)) 
    or 
  not(isV(c:Conf,pid,gid,"GripOp","grip","open"))) .
\end{verbatim}
\end{small}
\noindent
These properties are defined in 
Section~\ref{asubsec:basic-app-properties}.

\item \texttt{armGoingI}:
If the arm is going from the opposite position
to the initial position then the grip is closed
(presumable gripping the picked item).
The negation is formalized by 
 \texttt{armGoingIInv(c:Conf,pid,gid,aid,goI)} 
which reduces to
\begin{small}
\begin{verbatim}
    becomeV(c:Conf,pid,aid,"ArmGoI","arm", goI)
    and      
    not(isV(c:Conf,pid,gid,"GripCl","grip","close")) .
\end{verbatim}
\end{small}

\item \texttt{armGoingNI}: 
If the arm is going from the initial position
to the opposite position then the grip is open.
The negation is formalized by 
\texttt{armGoingNIInv(c:Conf,pid,gid,aid,goNI)} 
which reduces to
\begin{small}
\begin{verbatim}
    becomeV(c:Conf,pid,aid,"ArmGoNI","arm",goNI) 
    and
    not(isV(c:Conf,pid,gid,"GripOp","grip","open")) .
\end{verbatim}
\end{small}
   
\item \texttt{gripClosing}:  If the grip is closing then
the arm is at the pickup location (\texttt{goNI}).
The negation is formalized by
 \texttt{gripClosingInv(c:Conf,pid,gid,aid,goNI)} 
which reduces to
\begin{small}
\begin{verbatim}
    becomeV(c:Conf,pid,gid,"GripCl","grip","close")
    and 
    not(isV(c:Conf,pid,aid,"ArmGoNI","arm",goNI)) .
\end{verbatim}
\end{small}
  
\item \texttt{gripOpening}:  If the grip is opening then
the arm is at the initial (placing) location.
The negation is formalized by
 \texttt{gripOpeningInv(c:Conf,pid,gid,aid,goI)}
which reduces to   
\begin{small}
\begin{verbatim}
    becomeV(c:Conf,pid,gid,"GriOp","grip","open")}
    and 
    not(isV(c:Conf,pid,aid,"ArmGoI","arm",goI)) .
\end{verbatim}
\end{small}
\end{itemize}

Using the initial configuration constructor for a single PnP
application with arm starting at the left, for each invariant
we search for a violation. For example the search command for
violation of the PnP Idle invariant is

\begin{small}
\begin{verbatim}
search  {initRL("pctl","gr","arm",startAMsg("pctl","PUTS")) log(nilLI)}
       =>+ {c:Conf} such that 
           pnpIdleInv(c:Conf,"pctl","gr","arm","goL") .
\end{verbatim}
\end{small}
\noindent
In each case the result is no solution.

We then consider configurations with an attacker and one or
two rounds of the PnP protocol. (Two rounds are specified by
giving the part sensor two messages to send.) For example,
the search command for attack on PnP Idle invariant is

\begin{small}
\begin{verbatim}
search {initRL("pctl","gr","arm",startAMsg("pctl","PUTS")) log(nilLI) 
          ["eve" | kb(mtDM) caps(mcX(20))]}
      =>+ {c:Conf} such that 
          pnpIdleInv(c:Conf,"pctl","gr","arm","goL") .
\end{verbatim}
\end{small}

\noindent
Table \ref{table:pnpAttacks} summarizes the attacks found.

\begin{figure}[ht]
\begin{small}
\begin{verbatim}
      Invariant  nRnd    mcX     Msg
      pnpIdle      1      20     ArmGoNI,GripCl
                   2      20     ArmGoNI,GripCl
      armGoingI    1      20,40  none
                   2      20     GripOp
      armGoingNI   1    0,20,40  none
                   2      20     GripCl
      gripClosing  1      20,40  none
                   2      20     ArmGoI
      gripOpening  1      20,40  none
                   2      20     nonTerm
\end{verbatim}
\end{small}
\caption{Summary of attacks on PnP.
The column \texttt{nRnd} is the number of rounds,
\texttt{mcX} is delay argument to the \texttt{mcX} attack
capability, and \texttt{msg} is the message identifier of the
attacked message.}
\label{table:pnpAttacks}
\end{figure}

As for the Bridge application, communications between the
controller and the part sensor are vulnerable to attacks
on the invariants, and search for attacks in
dialected forms of the PnP scenarios find no solution.

Note there are many more scenarios to be constructed
from the PnP application including pairs of PnP
systems that coordinate in various ways.  Here
one can search for attacks that violate combinations
of invariants likely requiring attacks of multiple
messages.


\end{document}